\documentclass[11pt]{article}
\usepackage[utf8]{inputenc}
\usepackage{amsmath,amssymb}
\usepackage{picins}
\usepackage{wrapfig}
\usepackage{amsfonts}
\usepackage{amssymb,epic,eepic}
\usepackage{dcolumn}
\usepackage{bm}
\usepackage{bbm}
\usepackage{verbatim}
\usepackage{color}
\usepackage{amsthm}
\usepackage{tikz}
\usepackage{enumitem}
\usepackage{mathabx}
\usepackage{tikz}
\usetikzlibrary{arrows,positioning,decorations.pathmorphing,decorations.markings,shapes,fadings}

\newtheorem{thm}{Theorem}
\newtheorem{lem}{Lemma}
\newtheorem{cor}{Corollary}
\newtheorem{conjecture}{Conjecture}
\newtheorem{rem}{Remark}
\newtheorem{defn}{Definition}
\newtheorem{example}{Example}

\DeclareMathOperator{\conv}{conv}

\title{The Arbitrarily Varying Wiretap Channel --\\
Secret Randomness, Stability and Super-Activation}
\author{J. N\"otzel$^{(1,3)}$, M. Wiese$^{(2)}$, H. Boche$^{(1)}$\\
\scriptsize{Electronic addresses: janis.noetzel@tum.de, moritzw@kth.se, boche@tum.de}
\vspace{0.2cm}\\
{\footnotesize $^{(1)}$ Lehrstuhl f\"ur Theoretische Informationstechnik, Technische Universit\"at M\"unchen,}\\
{\footnotesize 80290 M\"unchen, Germany.}
\vspace{0.2cm}\\
{\footnotesize $^{(2)}$ ACCESS Linnaeus Center, KTH Royal Institute of Technology, Stockholm, Sweden.}
\vspace{0.2cm}\\
\footnotesize{$^{(3)}$F\'{\i}sica Te\`{o}rica: Informaci\'{o} i Fen\`{o}mens Qu\`{a}ntics, Universitat Aut\`{o}noma de Barcelona,}\\
\footnotesize{ES-08193 Bellaterra (Barcelona), Spain.}
}

\begin{document}
\maketitle

\begin{abstract}
We define the common randomness assisted capacity of an arbitrarily varying channel (AVWC) when the Eavesdropper is kept ignorant about the common randomness. We prove a multi-letter capacity formula for this model. We prove that, if enough common randomness is used, the capacity formula can be given a single-shot form again.\\
We then consider the opposite extremal case, where no common randomness is available, and derive the capacity. It is known that the capacity of the system can be discontinuous under these circumstances. We prove here that it is still \emph{stable} in the sense that it is continuous around its positivity points. We further prove that discontinuities can only arise if the legal link is symmetrizable and characterize the points where it is positive. These results shed new light on the design principles of communication systems with embedded security features.\\
At last we investigate the effect of super-activation of the message transmission capacity of AVWCs under the average error criterion. We give a complete characterization of those AVWCs that may be super-activated. The effect is thereby also related to the (conjectured) super-activation of the common randomness assisted capacity of AVWCs with an eavesdropper that gets to know the common randomness.\\
Super-activation is based on the idea of ``wasting'' a few bits of non-secret messages in order to enable provably secret transmission of a large bulk of data, a concept that may prove to be of further importance in the design of communication systems. In this work we provide further insight into this phenomenon by providing a class of codes that is capacity-achieving and does not convey any information to the Eavesdropper.
\end{abstract}
\setcounter{tocdepth}{2}
\tableofcontents
\begin{section}{Introduction}
Just like in our previous work \cite{wiese-noetzel-boche-I}, we investigate a model on the intersection between the two areas of secrecy and robust communication in information theory: the arbitrarily varying wiretap channel (AVWC). The communication scenario is depicted in Figure \ref{figurethree}.\\
In this model, a sender (Alice) would like to send messages to a legitimate receiver (Bob) over a noisy channel. Involved into the scenario are
two other parties: a jammer (James) who can actively influence the channel and a second but illegitimate receiver (Eve). Alice's and Bob's goal
is to achieve reliable and secure communication:\\
First, Bob should be able to decode Alice's messages with high probability (with respect to the average error criterion) no matter what the
input of James is.\\
Second, the mutual information between the messages and Eve's output should be close to zero. Again, this has to be the case no matter what the
input of James is.\\
Like in our previous work, we add the option of Alice and Bob having access to perfect copies of the outcomes of a random experiment $\mathcal
G$ (a source of common randomness). While in our previous work \cite{wiese-noetzel-boche-I} we considered the case where Eve gets an exact copy
of the outcomes received by Alice and Bob, we now extend our study to the case where Eve remains completely ignorant.\\
The only party which has no access to $\mathcal G$ in all the scenarios we study is James. We call the capacities which we derive from the two scenarios the ``correlated random coding mean secrecy capacity'' if Eve has information about $\mathcal G$ and ``secret common randomness assisted secrecy capacity'' if Eve has no information about it. When no common randomness is present at all, we speak of the ``uncorrelated coding secrecy capacity''. For the sake of an extended discussion of secrecy criteria we also define a ``capacity with public side-information'' which is the data transmission benchmark for systems where Eve gets to know a part of the messages.\\
From now on, we use the label $C_{\mathrm{S}}$ for the uncorrelated coding secrecy capacity (when no shared randomness is available between Alice and Bob) and $C^{\mathrm{mean}}_{\mathrm{S,ran}}$ for the correlated random coding mean secrecy capacity (just as in our previous work \cite{wiese-noetzel-boche-I} we restrict attention to the case where common randomness is used. To the reader which is not familiar with that work we apologize, as some of our results rely on that previous work). The secret common randomness assisted secrecy capacity is labelled $C_{\mathrm{key}}$ and the capacity with public side information $C_{\mathrm{pp}}$.
\begin{figure}[hhh]
\begin{center}
\begin{tikzpicture}
\node (A) at (-0.2,1) {Alice};
\node (A-in) at (1,0) {$\blacksquare$};
\node (J-in) at (1,-0.3) {$\blacksquare$};
\draw[line width=1pt] (0.8,0.2) rectangle (1.2,-0.5);
\node (G) at (3.2,2.2) {$\mathcal G$};
\node (E) at (3,-1) {Eve};
\node (J) at (0,-1) {James};
\node (B) at (2.5,1) {Bob};
\draw[->, line width=2pt, shorten >=-8pt, color=gray] (A) to (A-in);
\draw[->] (G) to (A);
\draw[->] (G) to (B);
\draw[->] (G) to (E);
\draw[->, line width=2pt, shorten <=0.5pt, shorten >=-2pt, -triangle 45] (1.2,0) to[bend angle = -1, bend left] (B);

\draw[->, line width=2pt, shorten <=0.5pt, shorten >=-2pt, -triangle 45] (1.2,0) to[bend angle = -1, bend left] (E);
\fill[fill=black] (1.2,0.1) to[bend angle = 1, bend left] (E) to[bend angle = 1, bend right] (1.2,-0.4) to (1.2,0.1);
\fill[fill=black] (1.2,0.1) to[bend angle = 1, bend left] (B) to[bend angle = 1, bend right] (1.2,-0.4) to (1.2,0.1);

\path[->, line width=2pt, shorten >=-6pt, color=gray] (J) edge[bend left] (J-in);
\draw[->, line width=2pt, color=gray] (J) to[bend angle =5, bend right] (E);
\end{tikzpicture}
\begin{tikzpicture}
\node at (0,0) {};
\node at (1,0) {};
\end{tikzpicture}
\begin{tikzpicture}
\node (A) at (-0.2,1) {Alice};
\node (A-in) at (1,0) {$\blacksquare$};
\node (J-in) at (1,-0.3) {$\blacksquare$};
\draw[line width=1pt] (0.8,0.2) rectangle (1.2,-0.5);
\node (G) at (3.2,2.2) {$\mathcal G$};
\node (E) at (3,-1) {Eve};
\node (J) at (0,-1) {James};
\node (B) at (2.5,1) {Bob};
\draw[->, line width=2pt, shorten >=-8pt, color=gray] (A) to (A-in);
\draw[->] (G) to (A);
\draw[->] (G) to (B);
\draw[->, line width=2pt, shorten <=0.5pt, shorten >=-2pt, -triangle 45] (1.2,0) to[bend angle = -1, bend left] (B);

\draw[->, line width=2pt, shorten <=0.5pt, shorten >=-2pt, -triangle 45] (1.2,0) to[bend angle = -1, bend left] (E);
\fill[fill=black] (1.2,0.1) to[bend angle = 1, bend left] (E) to[bend angle = 1, bend right] (1.2,-0.4) to (1.2,0.1);
\fill[fill=black] (1.2,0.1) to[bend angle = 1, bend left] (B) to[bend angle = 1, bend right] (1.2,-0.4) to (1.2,0.1);

\path[->, line width=2pt, shorten >=-6pt, color=gray] (J) edge[bend left] (J-in);
\draw[->, line width=2pt, color=gray] (J) to[bend angle =5, bend right] (E);
\end{tikzpicture}
\caption{Secure coding schemes for correlated random coding (left) and secret common randomness assisted coding (right)} \label{figurethree}
\end{center}
\end{figure}
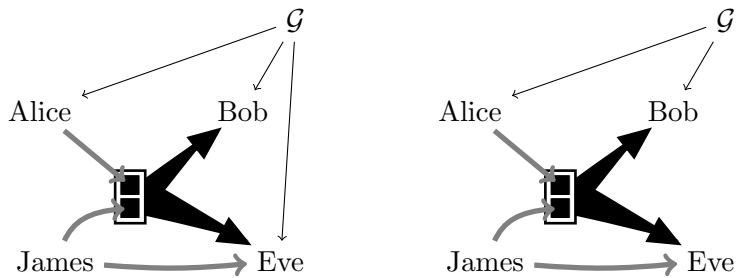
As is depcited in Figure \ref{figurethree}, it is of vital importance that Eve cannot communicate to James.\\
We give a unified treatment of the subject which allows us to observe the behaviour of the system while we change the amount of and the access to the common randomness: for common randomness set to zero one observes instabilities of the system (in the sense that the capacity is not a continuous function of the channel parameters anymore) and the effect of super-activation. Roughly speaking, two channels show super-activation when each of them cannot be used for a certain task (e.g. reliable communication under average error, maximal error or zero error criterion or, as in this work, secure communication) alone, but if a joint use is allowed the task becomes feasible. A more precise formulation is given in equations \eqref{eqn:super-activation-1} to \eqref{eqn:super-activation-3}, while the definition is part of Definition \ref{def:super-activation} which is followed by a short discussion of super-activation in the scenario treated here. If common randomness is used between Alice and Bob but Eve gets to know it as well, it is known from the results in \cite{wiese-noetzel-boche-I} that already small (a logarithmic number of bits, counted in block-length) amounts of common randomness resolve the instabilities (in the sense that the correlated random coding capacity is a continuous function of the channel parameters). It remains unknown whether super-activation is possible when common randomness is present, and this question is the content of Conjecture \ref{conjecture}.\\
The full advantage from common randomness can only be gained if Eve is kept ignorant of it. If common randomness is used at a nonzero rate, this rate adds linearly to the capacity of the system. All the capacity formulas which can be proven to hold in the various nontrivial scenarios are given by multi-letter formulae. Only if the common randomness exceeds the maximal amount of information which can be leaked to Eve do we recover a single-letter description. At that point, the linear increase in capacity stops:
\begin{figure}[hhh]
\begin{center}
\begin{tikzpicture}
\draw[->,line width=2pt,color=black] (0,-0.1) to (0,3);
\draw[->,line width=2pt,color=black] (-0.1,0) to (5,0);
\node (C) at (0,3.5) {$C_{\mathrm{key}}(\mathfrak W,\mathfrak V,G)$};
\node (G) at (5.5,0) {$G$};
\node (C_t) at (-1.3,2) {$C_{\mathrm{S,ran}}^{\mathrm{mean}}(\mathfrak W,\mathfrak T)$};
\node (C_m) at (-1.3,0) {$C_{\mathrm{S,ran}}^{\mathrm{mean}}(\mathfrak W,\mathfrak V)$};
\draw[-,line width=1pt,color=black] (0,0) to (2,2);
\draw[-,line width=1pt,color=black] (2,2) to (4.5,2);
\draw[-,line width=1pt,color=black] (-0.1,2) to (0.1,2);
\draw[-,line width=1pt,color=black] (2,0.1) to (2,-0.1);
\node (Null) at (0,-0.5) {$0$};
\node (X) at (2,-0.5) {$X$};
\end{tikzpicture}
\end{center}
\caption{Scaling of secrecy capacity with the rate $G$ of secret common randomness. It holds $X=C_{\mathrm{S,ran}}^{\mathrm{mean}}(\mathfrak W,\mathfrak T)-C_{\mathrm{S,ran}}^{\mathrm{mean}}(\mathfrak W,\mathfrak V)$, where $\mathfrak T$ is defined below after equation \eqref{eqn:first-equation}. \label{figrefour}}
\end{figure}
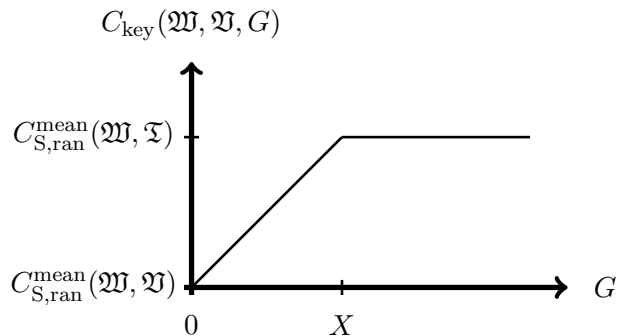
In order to carve out these principal features of secure data transmission in a both exact and elegant mathematical framework we let the number $n$ of channel uses go to infinity.\\
We will now sketch the connections of our work with some of the highlights and landmarks in the earlier literature. While we do not attempt to work in full rigour in the introduction, we will nonetheless gradually introduce some mathematical
notation.\\
The probabilistic law which governs the transmission of codewords sent by Alice and jamming signals sent by James to Eve and Bob is, for $n$
channel uses, given by
\begin{align}\label{eqn:first-equation}
w^{\otimes n}(y^n|x^n,s^n)v^{\otimes n}(z^n|x^n,s^n)=\prod_{i=1}^nw(y_i|x_i,s_i)v(z_i|x_i,s_i).
\end{align}
Here, $s^n=(s_1,\ldots,s_n)$ are the inputs of James, $x^n=(x_1,\ldots,x_n)$ those of Alice and $z^n=(z_1,\ldots,z_n)$ the outputs of Eve,
while $y^n=(y_1,\ldots,y_n)$ are received by Bob. All letters are assumed to be taken from finite alphabets. The action of the channel is, for
each natural number $n$ and therefore also as a whole, completely described by the pair $(W,V)$ of matrices of conditional probabilities and this could rightfully be called an
interference channel with non-cooperating senders and receivers. With respect to the historical development we will nonetheless prefer to use a description via the pair $(\mathfrak W,\mathfrak V)=((w(\cdot|\cdot,s))_{s\in\mathcal S},(v(\cdot|\cdot,s))_{s\in\mathcal S})$ and the label ``AVWC''.\\
This model has two important restrictions which are widely known: The case where $\mathfrak V$ does not convey any information about either one of its inputs is the arbitrarily varying channel (AVC). We will denote this special channel by $\mathfrak T=(T)$, where $t(z|x,s)=\frac{1}{|\mathcal Z|}$ for all $z$, $x$ and $s$. Before we give some credit to the historical developments in the area, we would like to emphasize that the notion introduced in (\ref{eqn:first-equation}) extends to products of arbitrary channels from $\mathcal X_1$ to $\mathcal Y_1$ and $\mathcal X_2$ to $\mathcal Y_2$, let them be denoted $W_1$ and $W_2$ with respective transition probability matrices $(w_1(y|x))_{x\in\mathcal X,y\in\mathcal Y}$ and $(w_2(y|x))_{y\in\mathcal Y,x\in\mathcal X}$ as follows: The transition probability matrix of $W_1\otimes W_2$ is defined by $w(y_1,y_2|x_1,x_2):=w_1(y_1|x_1)\cdot w_2(y_2|x_2)$ (for all $x_1\in\mathcal X_1$, $x_2\in\mathcal X_2$, $y_1\in\mathcal Y_1$ and $y_2\in\mathcal Y_2$).\\
The notation then carries over to arbitrarily varying channels, where we set
\begin{align}\label{eqn:def-of-tensor-power}
\mathfrak W\otimes\mathfrak W':=(W_s\otimes W_{s'}')_{s\in\mathcal S,s'\in\mathcal S'}.
\end{align}
The model of an arbitrarily varying channel has been introduced by Blackwell, Breiman and Thomasian \cite{bbt-avc} in 1960. They derived a formula for the capacity of an AVC with shared randomness-assisted codes under the average error criterion, and we will restrict our discussions to this criterion, although important nontrivial results concerning message transmission under the maximal error criterion have been obtained e.g. in \cite{kiefer-wolfowitz,ahlswede-note}. In \cite{ahlswede-note} it was shown that an explicit formula for the (weak) capacity of an AVC under maximal error criterion would imply a formula for the zero-error capacity of a discrete memoryless channel. The latter problem is open now for half a century.\\
In \cite{ahlswede-elimination}, Ahlswede developed an elegant and streamlined method of proof that, together with the random coding results of \cite{bbt-avc}, enabled him to prove the following: the capacity of an AVC (under the average error probability criterion) is either zero or equals its random coding capacity. This dichotomic behaviour is extended in the present work to the case where there is a (nontrivial) eavesdropper that has access to the shared randomness.\\
After the discoveries made in \cite{ahlswede-elimination}, an important open question was, when exactly the deterministic capacity with
vanishing average error is equal to zero, and in some sense the corresponding question for the AVWC is left open by us as well. In 1985, a first step towards a solution was made by Ericson \cite{ericson}, who came up with a sufficient condition that was proven to be necessary by Csiszar and Narayan \cite{csiszar-narayan}
in 1989.\\
The condition which was developed by Ericson, called \emph{symmetrizability}, reads as follows: An AVC $\mathfrak W$ is called symmetrizable if
there is a set $(u(\cdot|x))_{x\in\mathcal X}$ of probability distributions on $\mathcal S$ such that for every $x,x'\in\mathcal X$ and $y\in\mathcal Y$ we have
\begin{align}\label{eqn:definition-of-symmetrizability}
\sum_{s\in\mathcal S}u(s|x)w(y|x',s)=\sum_{s\in\mathcal S}u(s|x')w(y|x,s).
\end{align}
An arbitrarily varying channel $\mathfrak W$ that is symmetrizable cannot be used for reliable transmission of messages, as any input $x$ can, at least in an average sense, be made to look as if it had been another input $x'$. An example for a symmetrizable AVC that cannot be used for reliable transmission of messages just by using one encoder-decoder pair but still has a positive capacity for correlated random codes was given in \cite{bbt-avc} and later used again in \cite[Example 1]{ahlswede-elimination}. This exemplary AVC also serves as an important ingredient to the super-activation results in \cite{bs} and is, as an important example, also to be found in Remark \ref{rem:super-activation-explanation} of this document.\\
On the technical side, this work makes heavy use of the results that were obtained in the work \cite{csiszar-narayan} by extending one of their central results to the
situation where Eve gets some information via $V$. Namely, we are able to prove the following: If $\mathfrak W$ is non-symmetrizable, then
$C_{\mathrm{S}}(\mathfrak W,\mathfrak V)=C^{\mathrm{mean}}_{\mathrm{S,ran}}(\mathfrak W,\mathfrak V)$ for all possible $\mathfrak V$. We do not
attempt to give a necessary and sufficient condition for $C_{\mathrm{S}}$ to be positive, since a geometric characterization in the spirit of the symmetrizability condition \ref{eqn:definition-of-symmetrizability} is not even known for the usual wiretap channel. Rather, when speaking about the wiretap channel one usually refers to the concept of ``less noisy'' channels that was developed in \cite{csiszar-koerner-broadcast}.\\
The wiretap channel has been studied widely in the literature. The analysis started with the celebrated work \cite{wyner} of Wyner, an important follow-up work was \cite{csiszar-koerner-broadcast}, by Csiszar and K\"orner. While Wyner only treated the degraded case, Csiszar and K\"orner derived the capacity for the general discrete memoryless wiretap channel. The wiretap channel in the presence of common randomness which is kept secret from Eve (in this scenario, one could equally well speak of a secret key) was studied by Kang and Liu in \cite{ITW2010}.\\
In recent years there has been a growing interest in more elaborate models which combine insufficient channel state information with secrecy
requirements. Probably the earliest publications which came to our attention are the work \cite{liang-kramer-poor-shamai} by Liang, Kramer,
Poor and Shamai and the paper \cite{bloch-laneman} by Bloch and Laneman. Shortly after, the papers \cite{bbs-secrecy} and \cite{bbs-capacity} by Bjelakovi\'c, Boche and Sommerfeld got published. The work \cite{bbs-secrecy} provides a lower bound on the secrecy capacity of the compound wiretap channel with channel state information at the transmitter that matches an upper bound on the secrecy capacity
of general compound wiretap channels given provided in \cite{liang-kramer-poor-shamai}, establishing a full coding
theorem in this case. Important contributions of the work \cite{bbs-capacity} are a lower bound on what is called the ``random code secrecy capacity'' there, as well as a multi-letter expression for the secrecy capacity in the case of a best channel to the eavesdropper. The approach taken in this publication is closely related to the one taken in \cite{bbs-capacity}, but the use of different proof techniques enables us to provide much stronger results. An interesting parallel development is the work \cite{he-khisti-yener} by He, Khisti and Yener studies a two-transmitter Gaussian multiple access wiretap channel with multiple antennas at each  of the nodes. A characterization of the secrecy degrees of freedom region under a strong secrecy constraint is derived.\\
A surprising result that was discovered only recently by Boche and Wyrembelski in \cite{bs} is that of super-activation of AVWCs. We will explain this example in more detail in Remark \ref{rem:super-activation-explanation}. This effect was until then only known for information transmission capacities in quantum information theory, where it was proven by Smith and Yard in \cite{smith-yard-super-activation} that there exist channels which have the property that each of them alone has zero capacity but the two together have a positive capacity.\\
Before the work \cite{bs} this was assumed to be an effect which only shows up in quantum systems, where it was observed e.g. in \cite{smith-yard-super-activation}.\\
The work \cite{bs} gave an explicit example of super-activation which we repeat in Remark \ref{rem:super-activation-explanation}, but a deeper understanding of the effect was not achieved. Based on our finer analysis, we are now able to provide the following results: First, we give a much clearer characterization of super-activation of the \emph{uncorrelated}\footnote{Note that, due to the presence of an eavesdropper, it makes sense to allow the use of randomized encodings. Using, in such cases, the term ``random code'' is much too imprecise due to the potential presence of shared randomness between sender and receiver. Thus, we prefer to use the term uncorrelated codes. The random choice of codewords within an uncorrelated code represents lack of knowledge both for Eve and James. Analysing the case where James gains additional knowledge provides an interesting research opportunity, but care has to be taken when modelling the information flow from James to Eve.} coding secrecy capacity in Theorem \ref{theorem:symmetrizability-properties-of-C1det}. Second, and more for the sake of a clean discussion of coding and secrecy concepts, we define the capacity $C_{\mathrm{pp}}$ which explicitly keeps a part of the messages public (such that it may be that Eve is able to decode them). We do not attempt to give a further characterization of $C_{\mathrm{pp}}$ in this work, but we show that this capacity does as well show super-activation by use of the code concepts that were developed in \cite{bs}. Details are given in Subsection \ref{subsec:operational-definitions}, together with the exact definition of $C_{\mathrm{pp}}$.
\\\\
We will now give a broad sketch of our results concerning $C_{\mathrm{S}}$ and $C^{\mathrm{mean}}_{\mathrm{S,ran}}$, before we start concentrating
on $C_{\mathrm{key}}$. It was proven in \cite{wiese-noetzel-boche-I} that $C^{\mathrm{mean}}_{\mathrm{S,ran}}$ is a continuous quantity, and while the statement may
seem trivial at first sight, it becomes highly nontrivial when the following are taken into account:\\
There is at least no obvious way to deduce this statement directly just from the definition of the capacity, without first proving a coding
result, and the latter route was taken in \cite{wiese-noetzel-boche-I}, where an explicit formula for $C^{\mathrm{mean}}_{\mathrm{S,ran}}$ was found:
\begin{align}\label{eqn:capacity-formula}
C^{\mathrm{mean}}_{\mathrm{S,ran}}(\mathfrak W,\mathfrak V)=\lim_{n\to\infty}\frac{1}{n}\max_{p\in\mathcal P(\mathcal U_n)}\max_{U\in C(\mathcal
U_n,\mathcal X^n)}\left(\min_{q\in\mathcal P(\mathcal S)}I(p;W_q^{\otimes n}\circ U)-\max_{s^n\in\mathcal S^n}I(p;V_{s^n}\circ U)\right).
\end{align}
Explicit bounds on $|\mathcal U_n|$ were given as well. While one may argue that this is not an efficient description since one is forced to
compute the limit of a series of convex optimization problems, it turns out to be an incredibly useful characterization in the following sense:
First, it enables one to prove that $C^{\mathrm{mean}}_{\mathrm{S,ran}}$ is a continuous function in the pair $(\mathfrak W,\mathfrak V)$ and this result
was obtained in \cite{wiese-noetzel-boche-I}.\\
As has already been pointed out in \cite{bn-positivity}, the continuous dependence of the performance of a communication system on the relevant
system parameters is of central importance. To give just one example, consider recent efforts to build what is called ``smart grids''. Such
systems do certainly have high requirements both concerning reliability and stability of the communication in order to avoid potentially damaging consequences for its users.\\
While it is very interesting from a mathematical point of view, it certainly comes as an unpleasant surprise then that $C_{\mathrm{S}}$
does not grant us the favour of being a continuous function of the channel. On the other hand, this casts a flashlight on the importance of
distributed resources in communication networks - in this case the use of small amounts of common randomness. While one may now be tempted to
think that the transmission of messages over AVWCs without the use of common randomness is a rather adventurous task, we are also able to prove
that such a perception is wrong: Our analysis shows that $C_{\mathrm{S}}$ is continuous around its positivity points (this has been observed for classical-quantum arbitrarily varying channels in \cite{bn-positivity} already), and we are able to give an exact characterization of the discontinuity points as well. An example of a point of discontinuity has been given in \cite{boche-schaefer-poor}.\\
Moreover, our characterization of discontinuity relies purely on the computation of functions which are \emph{continuous} themselves, so that a
calculation of such points is at least within reach also from a computational point of view.\\
Further, the deep interconnection between continuity and symmetrizability which shows up in our work enables us to give a characterization of
pairs $(\mathfrak W_i,\mathfrak V_i)$ ($i=1,2$) for which super-activation is possible only in terms of $C^{\mathrm{mean}}_{\mathrm{S,ran}}$. In order to
be very explicit about super-activation, let us note the following:\\
The inequality
\begin{align}\label{eqn:super-activation-1}
C_{\mathrm{S}}(\mathfrak W_1\otimes\mathfrak W_2,\mathfrak V_1\otimes\mathfrak V_2)&\geq C_{\mathrm{S}}(\mathfrak W_1,\mathfrak
V_1)+C_{\mathrm{S}}(\mathfrak W_2,\mathfrak V_2)
\end{align}
follows trivially from the definition of $C$. It is common to all notions of capacity which are known to the authors. In contrast, if the
inequality
\begin{align}
C_{\mathrm{S}}(\mathfrak W_1\otimes\mathfrak W_2,\mathfrak V_1\otimes\mathfrak V_2)&> C_{\mathrm{S}}(\mathfrak W_1,\mathfrak
V_1)+C_{\mathrm{S}}(\mathfrak W_2,\mathfrak V_2)
\end{align}
holds, we speak of \emph{super-additivity} and only if we can even find AVWCs $(\mathfrak W_1,\mathfrak V_1)$ and $(\mathfrak W_2,\mathfrak
V_2)$ such that we have
\begin{align}\label{eqn:super-activation-3}
C_{\mathrm{S}}(\mathfrak W_1,\mathfrak V_1)=C_{\mathrm{S}}(\mathfrak W_2,\mathfrak V_2)=0,\qquad \mathrm{but}\qquad
C_{\mathrm{S}}(\mathfrak W_1\otimes\mathfrak W_2,\mathfrak V_1\otimes\mathfrak V_2)>0
\end{align}
we speak of \emph{super-activation}.\\
While it is clear from explicit examples in that super-activation of $C_{\mathrm{S}}$ is possible, it turns out in our work
via Theorem \ref{thm:full-characteriation-of-super-activation} that the effect is connected to the super-activation of
$C^{\mathrm{mean}}_{\mathrm{S,ran}}$, if the latter occurs. We would therefore like to take the opportunity of spurring future research by stating the
following conjecture:
\begin{conjecture}\label{conjecture}
There exist pairs $(\mathfrak W_1,\mathfrak V_1)$ and $(\mathfrak W_2,\mathfrak V_2)$ of (finite) AVWCs such that
\begin{align}
C^{\mathrm{mean}}_{\mathrm{S,ran}}(\mathfrak W_1,\mathfrak V_1)=C^{\mathrm{mean}}_{\mathrm{S,ran}}(\mathfrak W_1,\mathfrak V_1)=0,
\end{align}
but
\begin{align}
C^{\mathrm{mean}}_{\mathrm{S,ran}}(\mathfrak W_1\otimes\mathfrak W_2,\mathfrak V_1\otimes\mathfrak V_2)>0.
\end{align}
\end{conjecture}
An initial definition of objects such as $\mathfrak W_1\otimes\mathfrak W_2$ has been given in equation (\ref{eqn:def-of-tensor-power}) and repeated again in Subsection \ref{subsec:operational-definitions}. As a last introductory statement concerning super-additivity, let us mention the connection of super-activation to information transmission in networks:
Consider two orthogonal channels in a mobile communication network. Not taking into account the issues on the physical layer, on may end up in
a description of these channels via $\mathfrak W_1,\mathfrak W_2$ from Alice to Bob and $\mathfrak V_1,\mathfrak V_2$ from Alice to Eve. The
surprising result then is that, while it may be completely impossible to send information securely over each one of them, there exist coding
schemes which enable Alice to send her information securely if both she and Bob have access to both $\mathfrak W_1$ and $\mathfrak W_2$!\\
We will argue later in Subsection \ref{subsec:operational-definitions} how this effect works for the capacity $C_\mathrm{pp}$. While
this capacity offers an insightful view on the topic, we nevertheless concentrate on the interplay between $C_{\mathrm{S}}$, $C_{\mathrm{S,ran}}^{\mathrm{mean}}$ and $C_{\mathrm{key}}$ in this work.\\
\\
Let us now switch our attention to further results presented in this work. As mentioned already, we also extend earlier research to the case
where \emph{lots} of common randomness can be used (exponentially many random bits, to be precise) during our investigation of $C_{\mathrm{key}}$. We do
not dive into the issues arising when sub-exponentially many random bits are available, although the repeated appearance of the activating
effect of common randomness in arbitrarily varying systems seems to deserve a closer study. Our method of proving the direct part does again
yield nothing more than the statement that any number of random bits which scales asymptotically as $\mathrm{const.}+(1+\epsilon)\log(n)$ (for
some $\epsilon>0$) is sufficient for evading all issues which may arise from symmetrizable $\mathfrak W$.\\
Our restriction to positive rates $G$ of common randomness allows us to give an elegant formula for $C_{\mathrm{key}}$ as follows: For every $G>0$, it holds
\begin{align}
C_{\mathrm{key}}(\mathfrak W,\mathfrak V,G)=\min\{C^{\mathrm{mean}}_{\mathrm{S,ran}}(\mathfrak W,\mathfrak V)+G,C^{\mathrm{mean}}_{\mathrm{S,ran}}(\mathfrak W,\mathfrak T)\}.
\end{align}
Here, $\mathfrak T$ denotes the AVC consisting only of the memoryless ``trash'' channel $T$ mapping every legal input $x$ and jamming input $s$ onto an arbitrary element of $\mathcal Z$ with equal probability ($t(z|s,x)=|\mathcal Z|^{-1}$). While the reader familiar with the topic would certainly have guessed the validity of a formula of this form it is worth noting that this
formula is generally ``hard to compute'' in the sense that it requires one to calculate the limit in the formula (\ref{eqn:capacity-formula}) -
as long as $G<C^{\mathrm{mean}}_{\mathrm{S,ran}}(\mathfrak W,\mathfrak T)-C^{\mathrm{mean}}_{\mathrm{S,ran}}(\mathfrak W,\mathfrak V)$. If this condition is not met, then
$C_{\mathrm{key}}(\mathfrak W,\mathfrak V)=C^{\mathrm{mean}}_{\mathrm{S,ran}}(\mathfrak W,\mathfrak T)$. Since the latter is the usual capacity of the AVC $\mathfrak W$,
we conclude the following: If enough common randomness is available, the capacity of the system can be much more efficiently described - by a
formula which does not require regularization anymore!\\
Again, a look into the area of quantum information theory shows a striking resemblance: The capacity formula for the usual memoryless quantum channel has been proven to be given by regularized quantities in the general cases, both for entanglement transmission and for message
transmission. Without going into too much detail about quantum systems we cite here the work \cite{De05} by Devetak as our main reference underlining this statement, although this work has been both preceded and followed by important results dealing with the topic.\\
Apart from specific classes of quantum channels which were shown to have non-regularized capacity formulae \cite{devetak-shor} by Devetak and Shor, it has also been proven that the entanglement assisted capacity for message transmission over quantum channels is given by a one-shot formulae \cite{bsst} by Bennet, Shor, Smolin and Thapliyal.\\
To the best of our knowledge, a quantification of the amount of entanglement assistance which is necessary in order to turn the capacity
formula into a one-shot formula has not been given yet.
\end{section}
\begin{section}{Notation and Definitions\label{sec:notation-and-definitions}}
This section contains notation, conventions, as well as operational definitions and technical definitions
\begin{subsection}{Notation and Conventions\label{subsec:notation-and-conventions}}
In the context presented in this work, every finite set will equivalently be called an alphabet. Such alphabets are denoted by script letters such as $\mathcal A,\ \mathcal B,\ \mathcal S,\ \mathcal X,\ \mathcal Y,\ \mathcal Z$. The cardinality of a set $\mathcal A$ is denoted by $|\mathcal A|$. Every natural number $N\in\mathbb N$ defines a set $[N]:=\{1,\ldots,N\}$. The set of all permutations on such $[N]$ is written $S_N$. The function
$\exp:\mathbb R\to\mathbb R_+$ is defined with respect to base $2$: $\exp(t):=2^t$. The logarithm $\log$ is defined with respect to the same
base. For any $c\in\mathbb R$ we define $|c|^+$ by setting $|c|^+:=c$ if $c>0$ and $|c|^+:=0$ otherwise. A function $f:\mathcal A\to\mathbb R$ is nonnegative ($f\geq0$) if $f(a)\geq0$ holds for all $a\in\mathcal A$. To each finite set
$\mathcal A$ we associate the corresponding set $\mathcal P(\mathcal A):=\{p:\mathcal A\to[0,1]:p\geq0,\sum_{a\in\mathcal A}p(a)=1\}$ of
probability distributions on $\mathcal A$. Each random variable $A$ with values in $\mathcal A$ is associated to the unique $p\in\mathcal
P(\mathcal A)$ satisfying $\mathbb P(A=a)=p(a)$ for all $a\in\mathcal A$. An important subset of $\mathcal P(\mathcal A)$ is the set of its
extreme points. Every such extreme point is a point measure $\delta_a(a'):=\delta(a,a')$ where $\delta(\cdot,\cdot)$ is the usual
Kronecker-delta. The one-norm distance between two probability distributions $p,p'\in\mathcal P(\mathcal A)$ is $\|p-p'\|_1=\sum_{a\in\mathcal
A}|p(a)-p'(a)|$.\\
The expectation of a function $f:\mathcal A\to\mathbb R$ with respect to a distribution $p\in\mathcal P(\mathcal A)$ is written $\mathbb E_p
f:=\sum_{s\in\mathcal A}p(a)f(a)$ or, if $p$ is clear from the context, simply $\mathbb Ef$.\\
For each alphabet $\mathcal A$ and natural number $n\in\mathbb N$ we can build the corresponding product alphabet $\mathcal A^n:=\mathcal
A\times\ldots\times\mathcal A$, where $\times$ is the usual Cartesian product and there are exactly $n$ copies of $\mathcal A$ involved in the
definition of $\mathcal A^n$. The elements of $\mathcal A^n$ are denoted $a^n=(a_1,\ldots,a_n)$. Each such element gives rise to the
corresponding empirical distribution or \emph{type} $\bar N(\cdot|a^n)\in\mathcal P(\mathcal A)$ defined via $N(a|a^n):=|\{i:a_i=a\}|$ and
$\bar N(\cdot|a^n):=\frac{1}{n}N(\cdot|a^n)$. Given $\mathcal A$ and $n\in\mathbb N$, the set of all empirical distributions arising from an
element $a^n\in\mathcal A^n$ is $\mathcal P_0^n(\mathcal A):=\{\bar N(\cdot|a^n):a^n\in\mathcal A^n\}$. Each type $p\in\mathcal P_0^n(\mathcal A)$ defines the \emph{typical set}
$T_p:=\{a^n:\bar N(\cdot|a^n)=p(\cdot)\}$.\\
Channels are given by affine maps $W:\mathcal P(\mathcal A)\to\mathcal P(\mathcal B)$. The set of channels is denoted $C(\mathcal A,\mathcal
B)$. Every channel is uniquely represented (and can therefore be identified with) its set $\{w(b|a)\}_{a\in\mathcal A,b\in\mathcal B}$ of
transition probabilities, which are defined via $w(b|a):=W(\delta_a)(b)$. It acts as
\begin{align}\label{eqn:def-of-W(p)}
W(p):=\sum_{a\in\mathcal A}\sum_{b\in\mathcal B}w(b|a)p(a)\delta_b,
\end{align}
where both $W(p)\in\mathcal P(\mathcal B)$ and $\{\delta_b\}_{b\in\mathcal B}\subset\mathcal P(\mathcal B)$ (another way of writing the above formula would be to set $W(p)(\cdot)=\sum_{a\in\mathcal A}\sum_{b\in\mathcal B}w(b|a)p(a)\delta_b(\cdot)$ or even $W(p)(y)=\sum_{a\in\mathcal A}w(b|a)p(a)$).
As a shorthand, we may occasionally also write $Wp$ to denote $W(p)$, in analogy to linear algebra (every channel is naturally associated to its
representing stochastic matrix $(w(a|b))_{a,b}$ and can therefore be extended to a linear map on the appropriate vector spaces).\\
When operating on product alphabets such as $\mathcal A\times\mathcal B$ we define $p\otimes q\in\mathcal P(\mathcal A\times\mathcal B)$ to be
the distribution defined by $(p\otimes q)(a,b):=p(a)q(b)$. Correspondingly, $p^{\otimes n}\in\mathcal P(\mathcal A^n)$ is defined via
$p^{\otimes n}(a^n):=\prod_{i=1}^np(a_i)$. The same conventions hold for channels: if $V:\mathcal P(\mathcal A)\to\mathcal P(\mathcal B)$ and
$W:\mathcal P(\mathcal A')\to\mathcal P(\mathcal B')$, then $V\otimes W:\mathcal P(\mathcal A\times\mathcal A')\to\mathcal P(\mathcal
B\times\mathcal B')$ is defined via its transition probabilities as $(v\otimes w)((b,b')|(a,a')):=v(b|a)w(b'|a')$ and the notation carries over
to $n$-fold products $W^{\otimes n}$ of $W:\mathcal P(\mathcal A)\to\mathcal P(\mathcal B)$ as before by setting $w^{\otimes
n}(b^n|a^n):=\prod_{i=1}^nw(b_i|a_i)$.\\
For channels $W\in C(\mathcal A\times\mathcal B,\mathcal C)$ it will become important to derive a short notation for cases where one input
remains fixed while the other is arbitrary. Such induced channels will be denoted, in case that this is unambiguously possible, by $W_p$ where
\begin{align}
W_p(\delta_a):=W(\delta_a\otimes p).
\end{align}
At times it will, in order to straighten out notation, also be necessary to write the transition probabilities as either $w_p(b|a)$ or even $w(b|a,p)$.\\
The Shannon entropy of $p\in\mathcal P(\mathcal A)$ is $H(p):=-\sum_{a\in\mathcal A}p(A)\log p(a)$, the relative entropy between two
probability distributions $p,q\in\mathcal P(\mathcal A)$ is $D(p\|q):=\sum_{a\in\mathcal A}p(a)\log(p(a)/q(a))$, if $q(a)=0\ \Rightarrow
p(a)=0$ for all $a\in\mathcal A$, and $D(p\|q):=+\infty$, else.\\
Every $p\in\mathcal P(\mathcal A)$ and channel $W:\mathcal P(\mathcal A)\to\mathcal P(\mathcal B)$ define a joint random variable which we call
$(A,B)$ for the moment and which is defined via $\mathbb P((A,B)=(a,b))=p(a)w(b|a)$ (for all $a\in\mathcal A,\ b\in\mathcal B$). This enables
us to use an equivalent formulation for the mutual information:
\begin{align}
I(p;W):=I(A;B).
\end{align}
A more operational definition of this quantity can be achieved by noting that the distribution of $(A,B)$ in this scenario arises from defining
$p^{(2)}\in\mathcal P(\mathcal A\times\mathcal A)$ by $p^{(2)}(a,a'):=p(a)\cdot\delta_a(a')$ for all $a,a'\in\mathcal A$ - it then holds
$\mathbb P((A,B)=(a,b))=((Id\otimes W)p^{(2)})(a,b)$ for all $a\in\mathcal A,b\in\mathcal B$. The operational interpretation of this
probability distribution is that Alice observes the outcomes $a$ of some random process. Given any such outcome, she makes one copy of it and
sends that copy over to Bob via the channel $W$, keeping the original data with herself.\\
We will go one step further and define mutual information on pairs of sequences $a^n\in\mathcal A^n$, $b^n\in\mathcal B^n$, this time by
defining a random variable $(A,B)$ with values in $\mathcal A\times\mathcal B$ via $\mathbb P((A,B)=(a,b)):=\bar N(a,b|a^n,b^n)$ and then
setting
\begin{align}
I(a^n;b^n):=I(A;B).
\end{align}
In addition, we will need a suitable measure of distance between AVWCs. Our object of choice is the Hausdorff distance which we define as follows: For two channels $W,\tilde W\in C(\mathcal A,\mathcal B)$, set
\begin{align}
\| W-\tilde W\|:=\max_{a\in\mathcal A}\| W(\delta_a)-\tilde W(\delta_a)\|.
\end{align}
Now we define for a given $\mathfrak W=(W_s)_{s\in\mathcal S}$, and $\mathfrak W'=(W'_{s'})_{s'\in \mathcal S'}$
\begin{align*}
g(\mathfrak W,\mathfrak W'):=\max_{s\in\mathcal S}\min_{s'\in\mathcal S'}\|W_s-W'_{s'}\|.
\end{align*}
Then we can ultimately define
\begin{align}
  d((\mathfrak W,\mathfrak V),(\mathfrak{ W}',\mathfrak{ V}')):=\max\{g(\mathfrak W\otimes\mathfrak V,\mathfrak{W}'\otimes\mathfrak V'),g(\mathfrak W'\otimes\mathfrak V',\mathfrak W\otimes\mathfrak{ V})\}.
\end{align}
This is a metric on the set of finite-state AVWCs with the corresponding alphabets $\mathcal{A,B,C}$. Another ingredient in the following is the notion of the convex hull of a set of channels, which can for e.g. AVCs $\mathfrak W=(W_s)_{s\in\mathcal S}$ be defined as
\begin{align}
\conv(\mathfrak W):=\left\{W=\sum_{s\in\mathcal S}q(s)W_s:q\in\mathcal P(\mathcal S)\right\}.
\end{align}
At last, we would like to mention that for any given $W\in C(\mathcal A,\mathcal B)$, $a\in\mathcal A$ and subset $\mathcal B'\subset\mathcal B$ we use the notation
\begin{align}
w(\mathcal B'|a):=\sum_{b\in\mathcal B'}w(b|a).
\end{align}
\end{subsection}
\begin{subsection}{Models and operational definitions\label{subsec:operational-definitions}}
At first, we give a formal definition of an arbitrarily varying channel. This extends our informal definition from the introduction, without any change in notation.
\begin{defn}[AVWC] Let $\mathcal X,\ \mathcal Y,\ \mathcal Z,\ \mathcal S$ be finite sets and for each $s\in\mathcal S$, let $W_s\in C(\mathcal X,\mathcal Y)$ and $V_s\in C(\mathcal X,\mathcal Z)$. Define $\mathfrak W:=(W_s)_{s\in\mathcal
S}$ and $\mathfrak V:=(V_s)_{s\in\mathcal S}$. The corresponding arbitrarily varying wiretap channel is denoted $(\mathfrak W,\mathfrak V)$.
Its action is completely specified by the sequence $(\{W_{s^n},V_{s^n}\}_{s^n\in\mathcal S^n})_{n\in\mathbb N}$, where
$W_{s^n}:=W_{s_1}\otimes\ldots\otimes W_{s_n}$ and $V_{s^n}:=V_{s_1}\otimes\ldots\otimes V_{s_n}$.\\
\end{defn}
\begin{rem}
The AVWC $(\mathfrak W,\mathfrak V)$ can equivalently be represented by defining $W\in C(\mathcal S\times\mathcal X,\mathcal Y)$ via $w(y|x,s):=w_s(y|x)$ and $V\in C(\mathcal S\times\mathcal X,\mathcal Z)$ via $v(z|x,s):=v_s(y|x)$. We
will use both representations interchangeably.\\
Whenever necessary, we will (for $n\in\mathbb N$ and $q\in\mathcal P(\mathcal S^n)$) also use the abbreviations
\begin{align}
W_q^{\otimes n}:=\sum_{s^n\in\mathcal S^n}q(s^n)W_{s^n},\qquad V_q^{\otimes n}:=\sum_{s^n\in\mathcal S^n}q(s^n)V_{s^n},
\end{align}
and the corresponding conditional probabilities are defined in the obvious way for all $x^n\in\mathcal X^n$, $y^n\in\mathcal Y^n$, $z^n\in\mathcal Z^n$ as
\begin{align}
w_q^{\otimes n}(y^n|x^n):=W_q^{\otimes n}(\delta_{x^n})(y^n),\qquad v_q^{\otimes n}(z^n|x^n):=V_q^{\otimes n}(\delta_{x^n})(z^n).
\end{align}
\end{rem}
Since a central part of our work is to study AVWCs under joint use, we have to carefully define what we mean here with ``joint use''. Let
$(\mathfrak W_1,\mathfrak V_1)$ and $(\mathfrak W_2,\mathfrak V_2)$ be two AVWCs. Since state alphabets are finite in all of our work, we will
without loss of generality assume that they have a joint state set $\mathcal S$. We then define
\begin{align}
(\mathfrak W_1\otimes\mathfrak W_2,\mathfrak V_1\otimes \mathfrak V_2):=((W_1(\cdot|\cdot,s)\otimes W_2(\cdot|\cdot,s'),V_1(\cdot|\cdot,s)\otimes V_2(\cdot|\cdot,s'))_{s,s'\in\mathcal S}.
\end{align}
We now come to a more ``classic'' topic: The definition of codes, rates and capacities. From the start, we will include the possibility of adding
some extra variables like shared randomness or common randomness between Alice and Bob, but also the possibility for Alice to divide her
message set into two parts: One which is to be kept secret from Eve and one which does not necessarily have to remain secret.\\
We introduce three different classes of codes, which are defined in the following and related to each other as follows: The class of shared randomness assisted codes contains those which use common randomness and these again contain the uncorrelated codes. Formal definitions are as follows:
\begin{defn}[Shared randomness assisted code] A shared randomness assisted code $\mathcal K_n$ for the AVWC $(\mathfrak W,\mathfrak V)$ consists of: a set
$[K]$ of messages, two finite alphabets $[\Gamma],[\Gamma']$ and a set of stochastic encoders $e^\gamma\in C([K],\mathcal X^n)$ (one for every value $\gamma\in[\Gamma]$) together with a collection $((D^{\gamma'}_k)_{k=1}^K)_{\gamma'=1}^{\Gamma'}$ of sets satisfying  $\bigcup_{k=1}^{K}D^{\gamma'}_k\subset\mathcal
Y^n$ and $D^{\gamma'}_k\cap D^{\gamma'}_{k'}=\emptyset$ for all $k\neq k'$ and for each $\gamma'$. In addition to that, there is a probability
distribution $\mu\in\mathcal P([\Gamma]\times[\Gamma'])$. Every such code defines the joint random variables $\mathfrak S_{s^n}:=(\mathfrak K_n,\mathfrak
K'_n,\mathfrak \Gamma_n,\mathfrak \Gamma'_n,\mathfrak Z_{s^n},\mathfrak X_n,\mathfrak Y_{s^n})$ ($s^n\in\mathcal S^n$) which are distributed
according to
\begin{align}
\mathbb
P(\mathfrak S_{s^n}=(k,k',\gamma,\gamma',z^n,x^n,y^n))&=\frac{1}{K}\mu(\gamma,\gamma')e^\gamma(x^n|k)\mathbbm1_{D^{\gamma'}_{k'}}(y^n)w_{s^n}(y^n|x^n)v_{s^n}(z^n|x^n)
\end{align}
The average error of $\mathcal K_n$ is
\begin{align}
\mathrm{err}(\mathcal K_n)=1-\max_{s^n\in\mathcal S^n}\sum_{k,\gamma,\gamma'=1}^{K,\Gamma,\Gamma'}\frac{\mu(\gamma,\gamma')}{K}\sum_{x^n\in\mathcal
X^n}e^\gamma(x^n|k) w_{s^n}(D^{\gamma'}_k|x^n).
\end{align}
\end{defn}
\begin{defn}[Common randomness assisted code]\label{def:CR-assisted-code} A common randomness assisted code $\mathcal K_n$ for the AVWC $(\mathfrak W,\mathfrak V)$ consists of: a set
$[K]$ of messages, a set $[\Gamma]$ of values for the common randomness and a set of stochastic encoders $e^\gamma\in C([K],\mathcal X^n)$ (one for each element $\gamma\in[\Gamma]$), together with a collection
$(D^\gamma_k)_{k,\gamma=1}^{K,\Gamma}$ of subsets $D^\gamma_k$ of $\mathcal Y^n$ satisfying $D^\gamma_k\cap D^\gamma_{k'}=\emptyset$ for all
$\gamma\in[\Gamma]$, whenever $k\neq k'$. Every such code defines the joint random variables $\mathfrak S_{s^n}:=(\mathfrak K_n,\mathfrak K'_n,\mathfrak
\Gamma_n,\mathfrak X_n,\mathfrak Y_{s^n},\mathfrak Z_{s^n})$ ($s^n\in\mathcal S^n$) which are distributed according to
\begin{align}
\mathbb P(\mathfrak S_{s^n}=(k,k',\gamma,x^n,y^n,z^n))&=\frac{1}{\Gamma\cdot
K}e^\gamma(x^n|k)\mathbbm1_{D^\gamma_{k'}}(y^n)w_{s^n}(y^n|x^n)v_{s^n}(z^n|x^n)
\end{align}
The average error of $\mathcal K_n$ is
\begin{align}
\mathrm{err}(\mathcal K_n)=1-\max_{s^n\in\mathcal S^n}\frac{1}{K\cdot\Gamma}\sum_{k,\gamma=1}^{K,\Gamma}\sum_{x^n\in\mathcal X^n}
e^\gamma(x^n|k)w_{s^n}(D^\gamma_k|x^n).
\end{align}
For technical reasons we also define, for all state sequences $s^n$, the corresponding average success probability of the code by
\begin{align}
d_{s^n}(\mathcal K_n)=\frac{1}{K\cdot\Gamma}\sum_{k,\gamma=1}^{K,\Gamma}\sum_{x^n\in\mathcal X^n}
e^\gamma(x^n|k)w_{s^n}(D^\gamma_k|x^n).
\end{align}
\end{defn}
One particularly interesting feature of AVCs is that it may be impossible to transmit any whatsoever small number of messages reliably from Alice to Bob without using shared randomness - but if one is willing to only spend a polynomial amount of common randomness, the capacity of the channel jumps to the maximally attainable value, an effect which was discovered in \cite{ahlswede-elimination}.\\
If a whole
communication network is being utilized it may be possible to use one part of the network to establish common randomness between the legal
parties (one could equally well speak of a secret key here) which is then used to send messages over another part of the system which may be
symmetrizable. This idea was first established in \cite{bs}. In this work, we will give a more careful analysis of the underlying structure, an undertaking which motivates the following definition:
\begin{defn}[Private/public code] A private/public code $\mathcal K_n$ for the AVWC $(\mathfrak W,\mathfrak V)$ consists of: two sets $[K],\ [L]$ of messages, an encoder $E\in C([K]\times[L],\mathcal X^n)$, and a collection $(D_{kl})_{k,l=1}^{K,L}$ of subsets of $\mathcal Y^n$ satisfying $D_{kl}\cap D_{k'l'}=\emptyset$ whenever $(k,l)\neq(k',l')$. Every such code defines the joint random variables
$\mathfrak S_{s^n}:=(\mathfrak K,\mathfrak L,\mathfrak K',\mathfrak L',\mathfrak X^n,\mathfrak Y_{s^n},\mathfrak Z_{s^n})$ ($s^n\in\mathcal S^n$) which are distributed according to
\begin{align}
\mathbb P(\mathfrak S_{s^n}=(k,l,k',l',x^n,y^n,z^n))&=\frac{1}{K\cdot L}e(x^n|k,l)\mathbbm1_{D_{k'l'}}(y^n)w_{s^n}(y^n|x^n)v_{s^n}(z^n|x^n).
\end{align}
The average error of $\mathcal K_n$ is
\begin{align}
\mathrm{err}(\mathcal K_n)=1-\max_{s^n\in\mathcal S^n}\sum_{k,l=1}^{K,L}\sum_{x^n\in\mathcal X^n}\frac{1}{K\cdot L}e(x^n|k,l)w_{s^n}(D_{k,l}|x^n).
\end{align}
\end{defn}
With this definition we can formalize the idea of ``wasting'' a few bits in order to guarantee secret communication. We would like to compare this approach to the case of a compound channel, where a sender that knows the channel parameters can send pilot sequences to the receiver in order to let him estimate the channel. The pilot sequences do not carry information from sender to receiver. With such a scheme, a sender with state information can transmit at strictly higher rates than one without. The higher capacity is reached by ``wasting'' some transmissions for the estimation. Since the number of channel uses that have to be used for estimation grows only sub-exponentially in the number of channel uses, there is no negative impact on the message transmission rate in asymptotic scenarios.\\
In the case treated here it turns out that sending a small amount of non-secret messages is the key to increase the secrecy capacity in specific situations. We would like to extend the formal background of this idea by allowing for a joint transmission of secret and non-secret messages:
\begin{defn}[Private/public coding scheme]\label{def:private/public_coding_scheme}
A private/public coding scheme operating at rates $(R_{\mathrm{pri}},R_\mathrm{pub})$ consists of a sequence $(\mathcal K_n)_{n\in\mathbb N}$ of private/public codes such that
\begin{align}
\lim_{n\to\infty}\mathrm{err}(\mathcal K_n)=0,\\
\liminf_{n\to\infty}\frac{1}{n}\log(K_n)=R_{\mathrm{pri}},\\
\liminf_{n\to\infty}\frac{1}{n}\log(L_n)=R_{\mathrm{pub}},\\
\limsup_{n\to\infty}\max_{s^n\in\mathcal S^n}I(\mathfrak K_n;\mathfrak Z_{s^n}|\mathfrak L_n)=0.
\end{align}
\end{defn}
A more restricted class of codes arises when there is only one type of messages, which ought to be kept secret, and in addition allows the use of common randomness.
\begin{defn}[Common randomness assisted coding scheme satisfying mean secrecy criterion]\label{def:mean-secrecy-coding-scheme}
A common randomness assisted coding scheme satisfying the mean secrecy criterion operating at rate $R$ consists of a sequence $(\mathcal K_n)_{n\in\mathbb
N}$ of common randomness assisted codes such that
\begin{align}
\lim_{n\to\infty}\mathrm{err}(\mathcal K_n)=0,\\
\liminf_{n\to\infty}\frac{1}{n}\log(K_n)=R,\\
\limsup_{n\to\infty}\max_{s^n\in\mathcal S^n}I(\mathfrak K_n;\mathfrak Z_{s^n}|\mathfrak\Gamma_n)=0.
\end{align}
\end{defn}
Note that both Definition \ref{def:private/public_coding_scheme} and Definition \ref{def:mean-secrecy-coding-scheme} require the mutual information between the secret messages and the output at Eve's site to be small on average, either over the public messages or over the common randomness. One may argue that this is a somewhat weak criterion. In our earlier paper \cite{wiese-noetzel-boche-I} we compared the capacity arising from the use of coding schemes under Definition \ref{def:mean-secrecy-coding-scheme} to a capacity derived under more severe requirements on the secrecy criterion. We were able to demonstrate that the respective capacities coincide. It is not known to us whether a more strict requirement in Definition \ref{def:private/public_coding_scheme} would lead to a different capacity.
\begin{defn}[Secure uncorrelated coding scheme]\label{def:uncorrelated-coding-scheme}
A secure uncorrelated coding scheme operating at rate $R$ consists of a sequence $(\mathcal K_n)_{n\in\mathbb N}$ of common
randomness assisted codes with $\Gamma_n=1$ for all $n\in\mathbb N$ such that
\begin{align}
\lim_{n\to\infty}\mathrm{err}(\mathcal K_n)=0,\\
\liminf_{n\to\infty}\frac{1}{n}\log(K_n)=R,\\
\limsup_{n\to\infty}\max_{s^n\in\mathcal S^n}I(\mathfrak K_n;\mathfrak Z_{s^n})=0.
\end{align}
\end{defn}
\begin{defn}[Secure coding scheme with secret common randomness\label{def:secure-coding-scheme-of-second-kind}]
A secure coding scheme with secret common randomness $\mathfrak K$ operating at rate $R$ and using an amount $G_{\mathfrak K}>0$ of common randomness
consists of a sequence $\mathfrak K:=(\mathcal K_n)_{n\in\mathbb N}$ of common randomness assisted codes with
$\lim_{n\to\infty}\frac{1}{n}\log\Gamma_n=G_{\mathfrak K}$ such that
\begin{align}
\lim_{n\to\infty}\mathrm{err}(\mathcal K_n)=0,\\
\liminf_{n\to\infty}\frac{1}{n}\log(J_n)=R,\\
\limsup_{n\to\infty}\max_{s^n\in\mathcal S^n}I(\mathfrak K_n;\mathfrak Z_{s^n})=0.
\end{align}
\end{defn}
\begin{rem}
The reader may wonder why the common randomness is only being quantified for secrecy schemes where the common randomness is kept secret. The reason for this becomes clear when reading \cite{bs-finite-coordination}, where it is proven that any shared randomness needed in order to achieve the correlated random coding mean secrecy capacity can always be assumed to not be larger than polynomially many bits of common randomness. These small amounts are not counted in the definition of the respective capacity. This result from \cite{bs-finite-coordination} got applied in our earlier paper \cite{wiese-noetzel-boche-I} as well.
\end{rem}
Since we completely restrict our analysis to the case where the system uses common randomness, we can spare a few indices to distinguish the different sources of external randomness:
\begin{defn}[Secrecy capacities]
Given $(\mathfrak W,\mathfrak V)$, we define for every $G>0$ the secret common randomness assisted secrecy capacity as
\begin{align}
C_{\mathrm{key}}(\mathfrak W,\mathfrak V,G):=\sup\left\{R:\begin{array}{l}\mathrm{There\ exists\ secret\ common\ randomness\ assisted}\\
\mathrm{coding\ scheme}\ \mathfrak K\ \mathrm{at\ rate\ }R\ \mathrm{with\ }G_{\mathfrak K}= G\end{array}\right\}.
\end{align}
The uncorrelated coding secrecy capacity and the correlated random coding mean secrecy capacity are defined just as in \cite{wiese-noetzel-boche-I}:
\begin{align}
C_{\mathrm{S}}(\mathfrak W,\mathfrak V)&:=\sup\left\{R:\begin{array}{l}\mathrm{There\ exists\ a\ secure\ uncorrelated}\\
\mathrm{coding\ scheme\ operating\ at\ rate\ }R \end{array}\right\}\\
C^{\mathrm{mean}}_{\mathrm{S,ran}}(\mathfrak W,\mathfrak V)&:=\sup\left\{R:\begin{array}{l}\mathrm{There\ exists\ a\ secure\ common\ randomness}\\
\mathrm{assisted\ coding\ scheme\ satisfying\ the\ mean}\\
\mathrm{secrecy\ criterion\ operating\ at\ rate\ }R\end{array}\right\}.
\end{align}
\end{defn}
We refrain from defining the rate region for private and public messages in this work, and restrict ourselves to consider only the boundary of that region that arises from letting $R_\mathrm{pub}$ be arbitrarily small. This does for example allow us to transmit any finite number of messages, or numbers of messages that scale sub-exponentially in the number of channel uses.
\begin{defn}[Private/public secrecy capacity] The private/public secrecy capacity is given by
\begin{align}
C_\mathrm{pp}(\mathfrak W,\mathfrak V):=\sup\left\{R:\begin{array}{l}\mathrm{There\ exists\ a\ private/public\ coding\ scheme\ at}\\
\mathrm{rates}\ (R_\mathrm{pub},R_\mathrm{pri})\ \mathrm{such\ that}\ R=R_\mathrm{pri}\end{array}\right\}.
\end{align}
\end{defn}
The above definition explicitly allows for the super-activation strategy of \cite{bs} to be used, and shall be explained using this example first. Before we do so, let us give the formal definition of super-activation:
\begin{defn}[Super-activation]\label{def:super-activation}
Let $(\mathfrak W_1,\mathfrak V_1)$ and $(\mathfrak W_2,\mathfrak V_2)$ be AVWCs. Then $(\mathfrak W_1,\mathfrak V_1)$, $(\mathfrak W_2,\mathfrak V_2)$ are said to show super-activation if $C_\mathrm{S}(\mathfrak W_1,\mathfrak V_1)=C_\mathrm{S}(\mathfrak W_2,\mathfrak V_2)=0$ but $C_\mathrm{S}(\mathfrak W_1\otimes\mathfrak W_2,\mathfrak V_1\otimes\mathfrak V_2)>0$.
\end{defn}
Now set $\mathfrak W:=\mathfrak W_1\otimes\mathfrak W_2$ and $\mathfrak V:=\mathfrak V_1\otimes\mathfrak V_2$. In order to simplify the discussion, one may additionally set $\mathfrak V_2=\mathfrak W_2=(Id)$, where $Id\in C([2],[2])$ and assume that $\mathfrak W_1$ is symmetrizable but that $C_\mathrm{S,ran}^\mathrm{mean}(\mathfrak W_1,\mathfrak V_1)=\alpha>0$. It follows that $C_\mathrm{pp}(\mathfrak W_1,\mathfrak V_1)=C_\mathrm{pp}(\mathfrak W_2,\mathfrak V_2)=0$, because of symmetrizability and since decoding of the messages that are sent via $(\mathfrak W_2,\mathfrak V_2)$ is possible without any error both for Bob \emph{and} for Eve. These messages may therefore be treated as common randomness that is known by Eve. We know that already with the choice $L_n=n^2$ we have enough common randomness to remove any effect arising from symmetrizability of $\mathfrak W_1$. Since the code arising from the combination of sending and decoding public messages via $(Id,Id)$ and private messages via $(\mathfrak W_1,\mathfrak V_1)$ is a coding scheme that fits under Definition \ref{def:private/public_coding_scheme}, we get $C_\mathrm{pp}(\mathfrak W_1\otimes\mathfrak W_2,\mathfrak V_1\otimes\mathfrak V_2)\geq\alpha>0$.\\
That such a scheme does work as well when $C_\mathrm{S}$ is considered instead of $C_\mathrm{pp}$ can be understood as follows:\\
Let two AVWCs $(\mathfrak W_1,\mathfrak V_1)$ and $(\mathfrak W_2,\mathfrak V_2)$ be given. Let $\mathfrak W_1$ be symmetrizable, but such that
$C^{\mathrm{mean}}_{\mathrm{S,ran}}(\mathfrak W_1,\mathfrak V_1)=\alpha>0$. Since $\mathfrak W_1$ is symmetrizable
we have $C_{\mathrm{S}}(\mathfrak W_1,\mathfrak V_1)=0$. If no additional resources are available, the surplus $\alpha$ in the common-randomness assisted secrecy capacity cannot be put to use. Let now $C_{\mathrm{S}}(\mathfrak W_2,\mathfrak V_2)=0$ but $C_{\mathrm{S}}(\mathfrak W_2,\mathfrak T)=\beta>0$ ($\mathfrak T$ denotes the trash channel, so this just means that it is possible to reliably transmit messages over $\mathfrak W_2$). Then
\begin{align}
C_{\mathrm{S}}(\mathfrak W_1\otimes\mathfrak W_2,\mathfrak V_1\otimes\mathfrak V_2)\geq\alpha>0
\end{align}
and the reason for this effect is that (as before when we considered $C_\mathrm{pp}$) a small amount of messages can be sent over $\mathfrak W_2$ and is then used as common randomness, therefore increasing the rate of messages that can be sent reliably over $\mathfrak W_1$ from zero to $\alpha$. Of course, the messages sent over $\mathfrak W_2$ can be read by Eve. That this causes no problems with the security requirements can be seen by defining a toy-model where only two parallel channels with respective adversarially controlled channel states are considered. This is done as follows:\\
Let us define random variables $\mathfrak R_{s,\hat s}=(\mathfrak M,\hat{\mathfrak M},\mathfrak Z_{1,s},\hat{\mathfrak Z}_{2,\hat s})$ where
\begin{align}
\mathbb P(\mathfrak R=(m,\hat m,z,\hat z))=\frac{1}{M}\frac{1}{\hat M}w_{1,s}(z|m,\hat m)\hat w_{2,\hat s}(\hat z|\hat m)
\end{align}
and the channels $\{W_{1,s}\}_{s\in\mathcal S}$ and $\{\hat W_{\hat s}\}_{\hat s\in\hat{\mathcal S}}$ can be controlled by James separately. It is understood that $m$ are the messages, whereas $\hat m$ are the values of the shared randomness that is distributed between Alice and Bob by using $\{\hat W_{2,\hat s}\}_{\hat s\in\hat{\mathcal S}}$. We assume that for some small $\epsilon\geq0$ we have
\begin{align}\label{eqn:secrecy-assumption}
\max_{\hat s\in\hat{\mathcal S}}I(\mathfrak M;\hat{\mathfrak Z}_{2,\hat s}|\hat{\mathfrak M})\leq\epsilon.
\end{align}
Observe that $\hat{\mathfrak Z}_{2,\hat s}$ depends solely on $\hat{\mathfrak M}$ via the channel $\hat W_{2,\hat s}$ (this is where the fact that the two arbitrarily varying channels are used in parallel), so that the data processing inequality yields for every $s,\hat s$ that
\begin{align}
I(\mathfrak M;\mathfrak Z_{1,s},\hat{\mathfrak Z}_{2,\hat s})&\leq I(\mathfrak M;\mathfrak Z_{1,s},\hat{\mathfrak M}).
\end{align}
It is a consequence of the independence between $\mathfrak M$ and $\hat{\mathfrak M}$ that we can (for every $s$ and $\hat s$) then continue this chain of estimates as follows:
\begin{align}
I(\mathfrak M;\mathfrak Z_{1,s},\hat{\mathfrak Z}_{2,\hat s})&\leq I(\mathfrak M;\mathfrak Z_{1,s},\hat{\mathfrak M})\\
&=H(\mathfrak M)+H(\mathfrak Z_{1,s},\hat{\mathfrak M})-H(\mathfrak M,\mathfrak Z_{1,s},\hat{\mathfrak M})\\
&=H(\mathfrak M,\hat{\mathfrak M})+H(\mathfrak Z_{1,s},\hat{\mathfrak M})-H(\mathfrak M,\mathfrak Z_{1,s},\hat{\mathfrak M})-H(\hat{\mathfrak M})\\
&=H(\mathfrak M|\hat{\mathfrak M})+H(\mathfrak Z_{1,s},\hat{\mathfrak M})-H(\mathfrak M,\mathfrak Z_{1,s},\hat{\mathfrak M})\\
&=H(\mathfrak M|\hat{\mathfrak M})+H(\mathfrak Z_{1,s}|\hat{\mathfrak M})-H(\mathfrak M,\mathfrak Z_{1,s}|\hat{\mathfrak M})\\
&=I(\mathfrak M;\mathfrak Z_{1,s}|\hat{\mathfrak M})\\
&\leq\epsilon.
\end{align}
Thus it is clear that, in addition,
\begin{align}
\max_{s\in\mathcal S,\hat s\in\hat{\mathcal S}}I(\mathfrak M;\mathfrak Z_{1,s},\hat{\mathfrak Z}_{2,\hat s})&\leq\epsilon.\label{eqn:consequence-of-secrecy-assumption}
\end{align}
It is also evident that this argument ceases to hold true when the channels that are used for transmission of $M$ and of $\hat M$ ar not orthogonal anymore. Our sketch indicates why the protocol developed in \cite{bs} is able to meet the secrecy requirement in Definition \ref{def:uncorrelated-coding-scheme}.
\\\\
After we indicated why the super-activation protocol works we do now want to switch the topic and highlight a few connections to related problems and technical difficulties:

It is evident from the existing literature on AVCs \cite{ahlswede-cai}, arbitrarily varying classical-quantum channels \cite{bn-correlation} and on the quantification of shared randomness \cite{beigi,gacs-koerner,kang-ulukus,wyner-common-information,witsenhausen} that the latter is not an easy task. A brief overview concerning the connections between quantification of shared randomness and arbitrarily varying channels has been given in \cite{bn-correlation}. Our focus here is on systems that use only common randomness in various different ways.

In our previous work \cite{wiese-noetzel-boche-I} we developed a formula for $C_{\mathrm{S,ran}}^{\mathrm{mean}}$. The proof, extending the results established in \cite{bs} and \cite{bs-finite-coordination}, displays clearly that already amounts of common randomness which scale polynomially in the blocklength $n$ are sufficient for achieving the full random capacity. Moreover, an exact quantification of the amount of shared randomness is not necessary when speaking about correlated random coding mean secrecy capacity. Either no shared randomness is allowed in the sense that $\Gamma_n=1$ for all $n\in\mathbb N$ or else one allows arbitrarily large amounts of it but then only uses the above mentioned polynomial amount.

With the functions $G\mapsto C_{\mathrm{key}}(\mathfrak W,\mathfrak V,G)$ the story is a different one, as the following interesting behaviour occurs: They are well-defined for all $G>0$. However, when $G=0$ they are not unambiguously defined anymore, as it is clearly possible to take e.g. a sequence $(\Gamma_n)_{n\in\mathbb N}$ such that $\Gamma_n=n^2$ for each $n\in\mathbb N$. In that
case, $G=\lim_{n\to\infty}\frac{1}{n}\log\Gamma_n=0$, but the amount of randomness is sufficient in the sense that for every $\epsilon>0$ there exists a sequence $(\mathcal K_n)_{n\in\mathbb N}$ of codes which use only the common randomness $\Gamma_n$, operate at a rate
$R_\epsilon=C^{\mathrm{mean}}_{\mathrm{S,ran}}(\mathfrak W,\mathfrak V)-\epsilon$ and are both asymptotically secure and satisfy
$\lim_{n\to\infty}\mathrm{err}(\mathcal K_n)=0$. Thus, purely from the mathematical definition of $C_{\mathrm{key}}(\mathfrak W,\mathfrak V,G)$, one would be tempted to set $C_{\mathrm{key}}(\mathfrak W,\mathfrak V,0)=C^{\mathrm{mean}}_{\mathrm{S,ran}}(\mathfrak W,\mathfrak V)$.\\
However, from the operational point of view this is unsatisfying: imagine taking the statement ``no common randomness'' literally, and therefore setting $\Gamma_n=1$ for all $n\in\mathbb N$. Let $\mathfrak W$ be a symmetrizable AVC. In that case there is no chance to reliably transmit \emph{any} whatsoever small amount of messages with $\Gamma_n=1$ for all $n\in\mathbb N$ \cite{ericson}.\\
It is thus clear that $C_{\mathrm{S,ran}}^{\mathrm{mean}}(\mathfrak W,\mathfrak V)=\lim_{G\to0}C_{\mathrm{key}}(\mathfrak W,\mathfrak V,G)$ holds, but that it at least seems to be a difficulty to give a both operationally meaningful and mathematically satisfying definition of $C_{\mathrm{key}}(\mathfrak
W,\mathfrak V,0)$ (see e.g. \cite{wiese-diss} for a possible approach to such type of problem).
\\\\
A quantity which will be proved to be of importance during our proofs and when quantifying how close an AVC is to being symmetrizable is defined as follows:
We let $M_{\mathrm{fin}}$ denote the set of all finite
sets of elements of $C(\mathcal X,\mathcal Y)$.
\begin{defn}\label{def:definition-of-F} The function $F:M_{\mathrm{fin}}\to\mathbb R_+$ is defined via setting, for each $\mathfrak
W'=(W'(\cdot|\cdot,s))_{s\in\mathcal S}\in M_{\mathrm{fin}}$,
\begin{align}\label{eqn:definition-of-F}
F(\mathfrak W'):=\max_{U\in C(\mathcal X,\mathcal S)}\min_{x\neq x'}\left\|\sum_{s\in\mathcal S}u(s|x)W'(\delta_{x'}\otimes\delta_s)-\sum_{s\in\mathcal S}u(s|x')W'(\delta_x\otimes\delta_s)\right\|_1.
\end{align}
\end{defn}
This function obviously has the property that for every AWVC $\mathfrak W'$, the statement $F(\mathfrak W')=0$ is equivalent to $\mathfrak W'$
being symmetrizable.
\end{subsection}
\end{section}
\begin{section}{Main Results\label{sec:main-results}}
In this section we list our main results. We start with a coding theorem concerning the secret common randomness assisted secrecy capacity whose direct part
is based on our Lemma \ref{lemma:central-lemma-II} that we state directly afterwards. We continue with a second and even more delicate lemma,
which is an extension of \cite[Lemma 3 ]{csiszar-narayan} to AVWCs. This lemma (Lemma \ref{lemma:central-lemma}) is important: it provides a direct (coding) part for Theorem \ref{theorem:symmetrizability-properties-of-C1det}, which addresses the influence of the symmetrizability condition (\ref{eqn:definition-of-symmetrizability}) on the capacity $C_S$ and thereby relates it to $C_{\mathrm{S,ran}}^{\mathrm{mean}}$.\\
Our last result connects to the work \cite{bs}, which showed a very surprising effect that has so far not been observed for classical
information-carrying systems: super-activation. We give a precise characterization of the conditions which lead to super-activation in Theorem
\ref{thm:full-characteriation-of-super-activation}.
\begin{thm}[Coding Theorem for secret common randomness assisted secrecy capacity]\label{thm:coding-theorem}
Let $(\mathfrak W,\mathfrak V)$ be an AVWC. For every $n\in\mathbb N$, set $\mathcal U_n:=[|\mathcal X|^n]$. Define
\begin{align}
C^*(\mathfrak W,\mathfrak V):=\lim_{n\to\infty}\frac{1}{n}\max_{p\in\mathcal P(\mathcal U_n)}\max_{U\in C(\mathcal U_n,\mathcal
X^n)}\left(\min_{q\in\mathcal P(\mathcal S)}I(p;W_q^{\otimes n}\circ U)-\max_{s^n\in\mathcal S^n}I(p;V_{s^n}\circ U)\right).
\end{align}
It holds (with $\mathfrak T=(T)$ denoting the AVC consisting only of the memoryless channel that assigns the uniform output distribution to every input symbol),
\begin{align}
C_{\mathrm{key}}(\mathfrak W,\mathfrak V,G)&=\min\{C^*(\mathfrak W,\mathfrak V)+G,C^*(\mathfrak W,\mathfrak T)\}
\end{align}
\end{thm}
Of course, $C^*(\mathfrak W,\mathfrak T)$ is the capacity of the AVC $\mathfrak W$ under average error. This capacity has a single-letter
description. Since the first argument in above minimum is not single letter, there is room for speculation whether there is room for
improvement in this characterization or, if not, for which value of $G$ the description in terms of a single-letter quantity is possible and for which not. Apart from the complicated multi-letter form, an important take-away from the above formula is that the following is true:
\begin{cor}
For every $G>0$, the function $(\mathfrak W,\mathfrak V)\mapsto C_{\mathrm{key}}(\mathfrak W,\mathfrak V,G)$ is continuous.
\end{cor}
\begin{rem}
If $G=0$ in the sense that $\Gamma_n=0$ for all $n\in\mathbb N$, then for all AVWCs $(\mathfrak W,\mathfrak V)$ we know that $C_{\mathrm{key}}(\mathfrak W,\mathfrak V,G)$ equals $C_{\mathrm{S}}(\mathfrak W,\mathfrak V)$.
\end{rem}
We are getting closer to the technical core of our work now. The next Lemma is essential to proving the direct part of Theorem \ref{thm:coding-theorem}. It quantifies how many messages $L$ and how many different values $\Gamma$ for the common randomness are needed in order to make the output distributions at Eve's site independent from the chosen message $k$.
\begin{lem}\label{lemma:central-lemma-II}For every $\tau>0$ there exists a value $\nu(\tau)>0$ and an $N_0(\tau)$ such that
for all $n\geq N_0(\tau)$ and natural numbers $K,L,\Gamma$ and type $p\in\mathcal P_0^n(\mathcal X)$ there exist
codewords $(\mathbf x_{kl\gamma})_{k,l,\gamma=1}^{K,L,\Gamma}$ in $T_p\subset\mathcal X^n$ and decoding sets $D_{kl}^\gamma\subset\mathcal Y^n$ obeying $D_{kl}^\gamma\cap D_{k'l'}^\gamma=\emptyset$ if $(k,l)\neq(k',l')$, such that we have:\\
If $\tfrac{1}{n}\log(K\cdot L)\leq \min_{q\in\mathcal P(\mathcal S)}I(p;W_q)-\nu(\tau)$ and $\Gamma\geq2^{n\cdot5\cdot\nu(\tau)}$ then
\begin{align}
\min_{s^n}\sum_{\gamma=1}^\Gamma\frac{1}{\Gamma}\sum_{k,l=1}^{K,L}\frac{1}{K\cdot L}w_{s^n}(D_{kl}^\gamma|\mathbf x_{kl\gamma})\geq1-2^{-n\cdot\nu(\tau)}.
\end{align}
If $\tfrac{1}{n}\log(L\cdot\Gamma)\geq \max_q I(p;V_q)+\tau$ then
\begin{align}\label{eqn:exponential-bound-on-variational-distance-in-lemma-II}
\max_{s^n,k}\left\|\frac{1}{L\cdot\Gamma}\sum_{l,\gamma=1}^{L,\Gamma}v_{s^n}(\cdot|\mathbf x_{kl\gamma})-\mathbb
Ev_{s^n}(\cdot|X^n)\right\|_1\leq&2^{-n\cdot\nu(\tau)},
\end{align}
where $X^n$ is distributed according to $\mathbb P(X^n=x^n):=\frac{1}{|T_p|}\mathbbm1_{T_p}(x^n)$ and the dependence of $\nu$ on $\tau$ is such that $\lim_{\tau\to0}\nu(\tau)=0$.
\end{lem}
While Lemma \ref{lemma:central-lemma-II} delivers the correct interplay between and scaling of the size of the numbers of secret messages $K$, the number of additional messages $L$ that are just being sent in order to obfuscate Eve and the number of values for the (secret) common randomness $\Gamma$ that are being used up in the process, it is insufficient for dealing with the case when $\Gamma$ is set to one or is kept very small. For those cases where the secret or partially secret common randomness $\Gamma$ is set to one for every number of channel uses, we have to deal with the symmetrizability properties of the legal link $\mathfrak W$ from Alice to Bob. Initial statements in that case are as follows:
\begin{thm}[Symmetrizability properties of $C_{\mathrm{S}}$\label{theorem:symmetrizability-properties-of-C1det}]
Let $(\mathfrak W,\mathfrak V)$ be an AVWC.
\begin{enumerate}
\item If $\mathfrak W$ is symmetrizable, then $C_{\mathrm{S}}(\mathfrak W,\mathfrak V)=0$.
\item If $\mathfrak W$ is non-symmetrizable, then $C_{\mathrm{S}}(\mathfrak W,\mathfrak V)=C^{\mathrm{mean}}_{\mathrm{S,ran}}(\mathfrak
    W,\mathfrak V)$.
\end{enumerate}
\end{thm}
We now start to take on a slightly different point of view, under which the AVWC becomes an object that has some parameters which can be subject to changes. When considering practical deployment aspects, such a point of view is necessary as all information we may have gathered about the channel during for example a training phase may not be accurate enough to model the real-world behaviour. Thus one needs to understand whether a slight error in the parameters may lead to catastrophic events, and this is the content of our next theorem.
\begin{thm}[Stability of $C_{\mathrm{S}}$\label{theorem:stability-properties-of-C1det}] Let $(\mathfrak W,\mathfrak V)$ be an AVWC. If
$(\mathfrak W,\mathfrak V)$ satisfies $C_{\mathrm{S}}(\mathfrak W,\mathfrak V)>0$ then there is an $\epsilon>0$ such that for all
$(\mathfrak W',\mathfrak V')$ satisfying $d((\mathfrak W,\mathfrak V),(\mathfrak W',\mathfrak V'))\leq\epsilon$ we have
$C_{\mathrm{S}}(\mathfrak W',\mathfrak V')>0$.
\end{thm}
However, despite the reassuring statement of Theorem \ref{theorem:stability-properties-of-C1det}, care has to be taken at some points, which are characterized below.
\begin{thm}[Discontinuity properties of $C_{\mathrm{S}}$\label{theorem:discontinuity-properties-of-C1det}] Let $(\mathfrak W,\mathfrak
V)$ be an AVWC.
\begin{enumerate}
\item The function $C_{\mathrm{S}}$ is discontinuous at the point $(\mathfrak W,\mathfrak V)$ if and only if the following hold:
    First, $C^{\mathrm{mean}}_{\mathrm{S,ran}}(\mathfrak W,\mathfrak V)>0$ and second $F(\mathfrak W)=0$ but for all $\epsilon>0$ there is $\mathfrak
    W_\epsilon$ such that $d(\mathfrak W,\mathfrak W_\epsilon)<\epsilon$ and $F(\mathfrak W_\epsilon)>0$.
\item If $C_{\mathrm{S}}$ is discontinuous in the point $(\mathfrak W,\mathfrak V)$ then it is discontinuous for all $\hat{\mathfrak
    V}$ for which $C^{\mathrm{mean}}_{\mathrm{S,ran}}(\mathfrak{W},\hat{\mathfrak V})>0$.
\end{enumerate}
Note that $F(\mathfrak W)=0$ is equivalent to $\mathfrak W$ being symmetrizable - a property which is defined in the introduction in equation (\ref{eqn:definition-of-symmetrizability}). The function $F$ itself is the content of Definition \ref{def:definition-of-F}.
\end{thm}
The take-away from above Theorem is two-fold: First, it delivers a criterion for the finding of a point of discontinuity that only requires the validation that $C^\mathrm{mean}_\mathrm{S,ran}$ (a continuous function) is nonzero in a specific point and the running of a convex optimization problem (calculation of $F$ in that point). Second, it becomes clear that any discontinuity of the capacity $C_\mathrm{S}$ arises solely from effects that stem from the ``legal'' link $\mathfrak W$ - changing $\mathfrak V$ has no effect on discontinuity.
\begin{cor}
For every $\mathfrak W$, the function $\mathfrak V\mapsto C_{\mathrm{S}}(\mathfrak W,\mathfrak V)$ is continuous.
\end{cor}
Note that discontinuity is caused both by the legal link $\mathfrak W$ (see statement 1) and the link $\mathfrak V$ to Eve (statement 2), but depends on $\mathfrak V$ only insofar as the capacity $C^{\mathrm{mean}}_{\mathrm{S,ran}}(\mathfrak{W},\hat{\mathfrak V})$ has to stay above zero in order for a discontinuity to occur.\\
Theorem \ref{theorem:discontinuity-properties-of-C1det} also delivers an efficient way for calculating whether $C_\mathrm{S}$ is discontinuous in a specific point or not: One only needs to give a good-enough approximation of the continuous function $C^\mathrm{mean}_\mathrm{S,ran}$ and then run a convex optimization in order to calculate $F(\mathfrak W)$. Regarding future research, it may therefore be of interest to quantify the degree of continuity of the capacity of arbitrarily varying channels in those regions where it is continuous.
\begin{rem}
It is necessary to request the existence of the $\mathfrak W_\epsilon$ in the first statement of Theorem
\ref{theorem:discontinuity-properties-of-C1det}, and an easy example why this is so is the following:\\
Define $W_{i,\epsilon}\in C(\{1,2\},\{1,2,3\})$ for $i=1,2$ and $\epsilon\in[0,1/2]$ by
\begin{align}
W_{1,\epsilon}:=\left(\begin{array}{ll} 0&1-\epsilon\\ \epsilon&0\\ 1-\epsilon&\epsilon\end{array}\right),\qquad
W_{2,\epsilon}:=\left(\begin{array}{ll}1-\epsilon&0\\ \epsilon&1-\epsilon\\ 0&\epsilon\end{array}\right).
\end{align}
For every $\epsilon\in[0,1/2]$, these AVCs are symmetrizable with $u(1|1)=\epsilon/(1-\epsilon)$ and $u(1|2)=(1-2\cdot\epsilon)/(1-\epsilon)$.
The reason for this is that for every $\epsilon\in[0,1/2]$ the convex sets $\conv(\{W_{1,\epsilon}(\delta_1),W_{2,\epsilon}(\delta_1)\})$ and
$\conv(\{W_{1,\epsilon}(\delta_2),W_{2,\epsilon}(\delta_2)\})$ have non-empty intersections. It is also geometrically clear that for any
$\epsilon\in(0,1/2)$, there will be a small vicinity of AVCs which share this property. Thus, around such a $\mathfrak W_\epsilon$, all other AVCs
are symmetrizable as well and for every $\mathfrak V$ we therefore have both $C_{\mathrm{S}}(\mathfrak W_\epsilon,\mathfrak V)=0$ and
$C_{\mathrm{S}}(\mathfrak W',\mathfrak V)=0$ whenever $d(\mathfrak W_\epsilon,\mathfrak W')$ is small enough.\\
It is additionally clear from \cite{bbt-avc} that $C^{\mathrm{mean}}_{\mathrm{S,ran}}(\mathfrak W_0,\mathfrak T)>0$ and that it is therefore (since
$C^{\mathrm{mean}}_{\mathrm{S,ran}}$ is continuous by the results of \cite{wiese-noetzel-boche-I}) possible to choose $\mathfrak V$ and $\delta>0$ such
that $C^{\mathrm{mean}}_{\mathrm{S,ran}}(\mathfrak W_0,\mathfrak V)>0$, $C_{\mathrm{S}}(\mathfrak W_0,\mathfrak V)=0$, and
$C_{\mathrm{S}}(\mathfrak W',\mathfrak V)=0$ whenever $d(\mathfrak W_\delta,\mathfrak W')$ is small enough.\\
It is easy to see that the AVC $\mathfrak W_0$ does not share this property: Although $C^{\mathrm{mean}}_{\mathrm{S,ran}}(\mathfrak W_0,\mathfrak T)>0$ and
$C^{\mathrm{mean}}_{\mathrm{S,ran}}(\mathfrak W_0,\mathfrak T)>0$, it is easy to find explicit examples of AVCs $\mathfrak W'$ which are arbitrarily close to
$\mathfrak W_0$ but are non-symmetrizable.
\end{rem}
Of course, every whatsoever nice characterization of a set of interesting objects is pretty useless if the set turns out to be empty.
Fortunately, it has been proven in \cite{boche-schaefer-poor} that the function mapping an AVC $\mathfrak W$ to its capacity has discontinuity points by explicit example.\\ Such an example is also given by $(\mathfrak W_0,\mathfrak T)$ with $\mathfrak
W_0$ taken from above.
\begin{rem}
The capacity $C^{\mathrm{mean}}_{\mathrm{S,ran}}(\mathfrak W,\mathfrak V)$ was quantified in \cite{wiese-noetzel-boche-I}. It satisfies
\begin{align}
C^{\mathrm{mean}}_{\mathrm{S,ran}}(\mathfrak W,\mathfrak V)=\lim_{G\to0}C_{\mathrm{key}}(\mathfrak W,\mathfrak V,G)=C^*(\mathfrak W,\mathfrak V).
\end{align}
\end{rem}
The proofs of Theorems \ref{thm:coding-theorem} and Theorem \ref{theorem:symmetrizability-properties-of-C1det} are carried out by providing coding strategies.  The proof of the direct part of Theorem \ref{theorem:symmetrizability-properties-of-C1det} extends the techniques from \cite{csiszar-narayan} by adding constraints on the random code that lead to it having additional security features. These features are quantified in the following Lemma:
\begin{lem}\label{lemma:central-lemma}
For any $\tau>0$ and $\beta>0$, there exists a value $\nu(\tau)>0$ and an $N_0(\tau)$ such that for all $n\geq N_0(\tau)$, natural numbers
$K,L,\Gamma$ satisfying $K\cdot L\geq2^{n\cdot\tau}$ and type $p\in\mathcal P_0^n(\mathcal X)$ satisfying
$\min_{x:p(x)>0}p(x)\geq\beta$, there exist codewords $(\mathbf x_{kl\gamma})_{k,l,\gamma=1}^{K,L,\Gamma}$ in $T_p\subset\mathcal X^n$, and a $c'>0$ such that if $\Gamma^{-1}>\exp(-2^{n\cdot c'})$ and upon setting $R=\frac{1}{n}\log(K\cdot L)$ we have
\begin{align}
\max_{\gamma,x^n,s^n}|\{(k,l):(x^n,\mathbf x_{kl\gamma},s^n)\in T_{\bar N(\cdot|x^n,\mathbf x_{kl\gamma},s^n)}\}|\leq2^{n(|R-I(\mathbf
x_{kl\gamma};x^n,s^n)|^++\tau)}\label{eqn:property-1-of-central-lemma}\\
\max_{\gamma,s^n}|\{(k,l):I(\mathbf x_{kl\gamma};s^n)\geq\tau\}|\leq K\cdot L\cdot2^{-n\cdot\tau}\label{eqn:property-2-of-central-lemma}\\
\max_{\gamma,s^n} \left|\left\{(k,l,\gamma):\begin{array}{l}\mathrm{There\ is}\ (k',l',\gamma')\neq (k,l,\gamma)\ \mathrm{such\ that}\\
I(\mathbf x_{kl\gamma};\mathbf x_{k'l'\gamma'},s^n)-|R-I(\mathbf x_{kl\gamma};s^n)|^+>\tau
\end{array}\right\}\right|\leq K\cdot L\cdot 2^{-n\cdot\tau/2}\label{eqn:property-4-of-central-lemma}\\
\frac{\log L\cdot\Gamma}{n}\geq \max_{q\in\mathcal P(\mathcal S)}I(p;V_q)+\tau\ \Rightarrow\ \max_{s^n,k}\left\|\frac{1}{L\cdot\Gamma}\sum_{l,\gamma=1}^{L,\Gamma}V_{s^n}(\cdot|\mathbf x_{kl\gamma})-\mathbb
EV_{s^n}(\cdot|X^n)\right\|_1\leq2^{-n\cdot\nu(\tau)}\label{eqn:property-5-of-central-lemma}
\end{align}
where $X^n$ is distributed\ according\ to\ $\mathbb P(X^n=x^n):=\frac{1}{|T_p|}\mathbbm1_{T_p}(x^n)$ and the dependence of $\nu$ on $\tau$ is
such that $\lim_{\tau\to0}\nu(\tau)=0$.
\end{lem}
Our intention was be to apply this Lemma to AVWCs for which the link between Alice and Bob is not symmetrizable. While Lemma \ref{lemma:central-lemma} contains the possibility to use shared randomness $\Gamma$, this is not necessary in the
application intended by us in this work (we use it only with $\Gamma$ set to one). The main reason for keeping $\Gamma$ as a variable in our
proof this is that it allows us to deliver a unified treatment of the whole topic, increases the generality of the Lemma and does not require much additional work.
\begin{rem}
The properties $(\ref{eqn:property-1-of-central-lemma})$, $(\ref{eqn:property-2-of-central-lemma})$ and $(\ref{eqn:property-4-of-central-lemma})$ of the code are identical to those stated in \cite[Lemma 3]{csiszar-narayan}. This Lemma again is the main ingredient to the proof of \cite{csiszar-narayan} that non-symmetrizability (symmetrizability is defined in (\ref{eqn:definition-of-symmetrizability})) is sufficient for message transmission over AVCs if the average error criterion and non-randomized codes are used.
Our strategy thus is to use the properties $(\ref{eqn:property-1-of-central-lemma})$, $(\ref{eqn:property-2-of-central-lemma})$ and
$(\ref{eqn:property-4-of-central-lemma})$ in Lemma \ref{lemma:central-lemma} in order to ensure successful message transmission over the legal link, if $\mathfrak W$ is non-symmetrizable.\\
The main tool used by Csiszar and Narayan for proving properties $(\ref{eqn:property-1-of-central-lemma})$,
$(\ref{eqn:property-2-of-central-lemma})$ and $(\ref{eqn:property-4-of-central-lemma})$ of Lemma \ref{lemma:central-lemma} in their work
\cite{csiszar-narayan} was large deviation theory, and this is where we can make the connection to our work and prove the additional properties
via application of the Chernoff-bound.\\
Roughly speaking, this method of proof amounts to adding some additional requirements in a situation where any exponential number of additional
requirements can be satisfied simultaneously.
\end{rem}
When utilizing Lemma \ref{lemma:central-lemma} (with $\Gamma=1$) in the proof of Theorem \ref{theorem:symmetrizability-properties-of-C1det} one sees that while reliable transmission is achieved via fulfillment of conditions (\ref{eqn:property-1-of-central-lemma}), (\ref{eqn:property-2-of-central-lemma}) and (\ref{eqn:property-4-of-central-lemma}) in Lemma \ref{lemma:central-lemma} if and only if the legal link $\mathfrak W$ is non-symmetrizable, the security of the communication can always be achieved by making $L$ large enough. This implies that there are generic communication systems (AVWCs with a symmetrizable legal link) for which it is much easier to design codes that convey little information to Eve than codes which ensure robust communication.

In order to derive from Lemma \ref{lemma:central-lemma-II} the connection between symmetrizability and the capacity $C_S$ (which is the content of Theorem \ref{theorem:symmetrizability-properties-of-C1det}) it is necessary to prove not only achievability of quantities like e.g. $\min_qI(p;W_q)-\max_sI(p;V_s)$ but also of quantities like $\min_qI(p';W_{q}^n\circ U)-\max_{s^n}I(p';V_{s^n}\circ U)$ that involve multiple channel uses and pre-coding that is defined via the optimization problem \eqref{eqn:capacity-formula}. Such a process of adding pre-coding may unfortunately cause the AVWC arising from the concatenation of pre-coding and the original AVWC to be symmetrizable. This highly interesting interplay of pre-coding and symmetrizability is quantified in the next Lemma and the following example.
\begin{lem}\label{lem:extension-of-symmetrizability} Let $\mathfrak W$ be an arbitrarily varying channel with input alphabet $\mathcal A$, output alphabet $\mathcal B$ and state set $\mathcal S$. Let $T\in C(\mathcal A',\mathcal A)$ be a channel. Let $\mathfrak W'$ be the arbitrarily varying channel with input alphabet $\mathcal A'$, output alphabet $\mathcal B$ and state set $\mathcal R$ defined by $w'(b|a',s):=\sum_{a\in\mathcal A}w(b|a,s)t(a|a')$ (or, equivalently, via setting $W'_s:=W_s\circ T$ for all $s\in\mathcal S$).\\
If $\mathfrak W$ is symmetrizable then $\mathfrak W'$ is symmetrizable as well.
\end{lem}
That, even for channels $T$ whose associated matrix $(t(a|a')_{a'\in\mathcal A',a\in\mathcal A}$ has full range, the reverse implication ``$\mathfrak W'$ is symmetrizable $\Rightarrow$ $\mathfrak W$ is symmetrizable'' does not hold came as a surprise and is proven here by explicit example:
\begin{example}\label{example:invertible-pre-coding-induces-symmetrizability}
Define an AVC $\mathfrak W\subset C(\{x_1,x_2\},\{1,2,3\})$ by setting
\begin{align}
w(\cdot|s_1,x_1)&:=\delta_1,\\
w(\cdot|s_2,x_1)&:=\delta_2,\\
w(\cdot|s_1,x_2)&:=0.6\delta_1+0.2\delta_2+0.2\delta_3,\\
w(\cdot|s_2,x_2)&:=0.1\delta_1+0.3\delta_2+0.6\delta_3,
\end{align}
where $\delta_i(j)=1$ if and only if $i=j$ holds for $i,j\in[3]$. Then $W$ is non-symmetrizable: The equation
\begin{align}
\lambda w(\cdot|s_1,x_1)+(1-\lambda)w(\cdot|s_2,x_1)=\mu w(\cdot|s_1,x_2)+(1-\mu)w(\cdot|s_2,x_2)
\end{align}
cannot have a solution with $\lambda,\mu\in[0,1]$ because $\delta_3$ appears only on the right hand side and with strictly positive weights.

However, if we add pre-coding by a binary-symmetric channel $N_p$ with parameter $p\in[0,1]$ we obtain the new AVC $\mathfrak W'$ defined via $W_s':=W_s\circ N_p$ or, more concretely, by
\begin{align}
w'(\cdot|s_1,x_1)&=p\delta_1+p'(0.2\delta_1+0.6\delta_2+0.2\delta_3)\\
w'(\cdot|s_2,x_1)&=p\delta_2+p'(0.1\delta_1+0.3\delta_2+0.6\delta_3)\\
w'(\cdot|s_1,x_2)&=p'\delta_1+p(0.6\delta_1+0.2\delta_2+0.2\delta_3)\\
w'(\cdot|s_2,x_2)&=p'\delta_2+p(0.1\delta_1+0.3\delta_2+0.6\delta_3)
\end{align}
where $p':=1-p$. We set $p=0.4$. The equation
\begin{align}
\lambda w'(\cdot|s_1,x_1)+(1-\lambda)w'(\cdot|s_2,x_1)=\mu w'(\cdot|s_1,x_2)+(1-\mu)w'(\cdot|s_2,x_2)
\end{align}
can be written out explicitly into three equations for the two parameters $\mu,\lambda$. The solution is given by
\begin{align}
\lambda=31/37,\qquad \mu=75/148.
\end{align}
This shows that $\mathfrak W'$ is symmetrizable. The situation is depicted as follows:
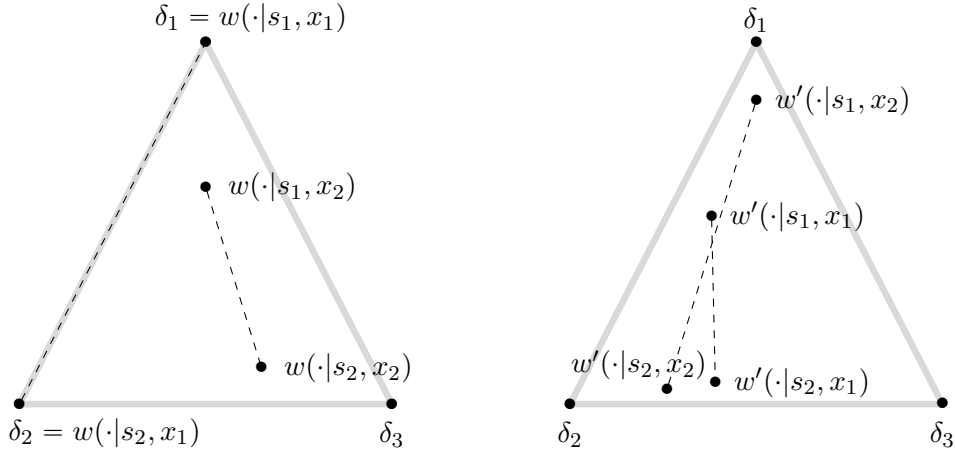
\begin{figure}[hhh]
\begin{center}
\begin{tikzpicture}[scale=3.5]
\definecolor{light-gray}{gray}{0.85}
\node (D1) at (-0.000160307, 0.966059) {$\bullet$};
\node (D2) at (-0.707336, -0.4082) {$\bullet$};
\node (D3) at (0.706949, -0.408155) {$\bullet$};
\node (W11) at (D1) {};
\node (W21) at (D2) {};
\node (W12) at (-0.000173516, 0.416364) {$\bullet$};
\node (W22) at (0.211953, -0.270747) {$\bullet$};
\draw[shorten <= -8pt, shorten >= -8pt,light-gray,line width=2.5pt] (D1) -- (D2);
\draw[shorten <= -8pt, shorten >= -8pt,light-gray,line width=2.5pt] (D2) -- (D3);
\draw[shorten <= -8pt, shorten >= -8pt,light-gray,line width=2.5pt] (D3) -- (D1);
\node (w12) at (W12) {\qquad\qquad\qquad$w(\cdot|s_1,x_2)$};
\node (w22) at (W22) {\qquad\qquad\qquad$w(\cdot|s_2,x_2)$};
\node (D1) at (-0.000160307, 0.966059) {$\bullet$};
\node (D2) at (-0.707336, -0.4082) {$\bullet$};
\node (D3) at (0.706949, -0.408155) {$\bullet$};
\node[above right = 0pt and -45pt] (d1) at (D1) {\qquad $\delta_1=w(\cdot|s_1,x_1)$};
\node[above right=-20pt and -8pt] (d2) at (D2) {$\delta_2=w(\cdot|s_2,x_1)$};
\node[above = -20pt] (d3) at (D3) {$\delta_3$};
\draw[dashed, shorten <=-5pt, shorten >=-5pt] (W12) -- (W22);
\draw[dashed, shorten <=-5pt, shorten >=-5pt] (W11) -- (W21);
\end{tikzpicture}
\qquad\qquad
\begin{tikzpicture}[scale=3.5]\definecolor{light-gray}{gray}{0.85}
\node (D1) at (-0.000160307, 0.966059) {$\bullet$};
\node (D2) at (-0.707336, -0.4082) {$\bullet$};
\node (D3) at (0.706949, -0.408155) {$\bullet$};
\node (W11) at (-0.16989, 0.30642) {$\bullet$};
\node (W21) at (-0.155763, -0.325728) {$\bullet$};
\node (W12) at (-0.000165591, 0.746181) {$\bullet$};
\node (W22) at (-0.339621, -0.353219) {$\bullet$};
\draw[shorten <= -8pt, shorten >= -8pt,light-gray,line width=2.5pt] (D1) -- (D2);
\draw[shorten <= -8pt, shorten >= -8pt,light-gray,line width=2.5pt] (D2) -- (D3);
\draw[shorten <= -8pt, shorten >= -8pt,light-gray,line width=2.5pt] (D3) -- (D1);
\node[above right = 0pt and -40pt] (w22) at (W22) {$w'(\cdot|s_2,x_2)$};
\node[right = 3pt] (w11) at (W11) {$w'(\cdot|s_1,x_1)$};
\node (w12) at (W12) {\qquad\qquad\qquad$w'(\cdot|s_1,x_2)$};
\node (w21) at (W21) {\qquad\qquad\qquad$w'(\cdot|s_2,x_1)$};
\node (D1) at (-0.000160307, 0.966059) {$\bullet$};
\node (D2) at (-0.707336, -0.4082) {$\bullet$};
\node (D3) at (0.706949, -0.408155) {$\bullet$};
\node[above=0pt] (d1) at (D1) {$\delta_1$};
\node[above=-20pt] (d2) at (D2) {$\delta_2$};
\node[above=-20pt] (d3) at (D3) {$\delta_3$};
\draw[dashed, shorten <=-5pt, shorten >=-5pt] (W12) -- (W22);
\draw[dashed, shorten <=-5pt, shorten >=-5pt] (W11) -- (W21);
\end{tikzpicture}
\caption{Light gray lines are the vertices of the probability simplex $\mathcal P(\{1,2,3\})$. The sets $\conv(\{w(\cdot|s_1,x_i),w(\cdot|s_2,x_i)\})$ where $i=1,2$ are displayed as dashed lines. The intersection of the dashed lines on the right shows that $\mathfrak W'$ is symmetrizable.}
\end{center}
\end{figure}
\end{example}
In order to derive the statement of Theorem \ref{theorem:symmetrizability-properties-of-C1det} from Lemma \ref{lemma:central-lemma} we can therefore not use a simple blocking strategy. Rather, we will present two methods of proof. The first employs a reasoning along the lines of equations \eqref{eqn:secrecy-assumption} until \eqref{eqn:consequence-of-secrecy-assumption}. This approach is based on the concept of using a few non-secret bits in order to guarantee secrecy for the actual data. While this is highly interesting from a practical point of view, it does not utilize the full strength of Lemma \ref{lemma:central-lemma}. This proof uses a set of public messages that can be read by Eve but not by James, secrecy is only obtained for the (exponentially larger) set of private messages.\\
Our second proof of Theorem \ref{theorem:symmetrizability-properties-of-C1det} is based on lifting the optimal pre-codings for $n$ channel uses to $n+1$ channel uses by using no pre-coding on the $(n+1)$th channel use. This type of pre-coding preserves non-symmetrizability. The second proof makes almost full use of the statements of Lemma \ref{lemma:central-lemma}, as we still set $\Gamma=1$. No public messages are used in the construction of the code.

It remains an interesting open question whether, for $n$ channel uses, the optimal channel $U$ arising from the $n$-th term of the optimization problem \eqref{eqn:capacity-formula} does in fact symmetrize $(W_{s^n})_{s^n\in\mathcal S^n}$ or not.

Our next result is potentially the most interesting in this work, since it sheds additional light on a rather new phenomenon: the
super-activation of ``the'' secrecy capacity of AVWCs.
\begin{thm}[Characterization of super-activation of $C_{\mathrm{S}}$ via properties of
$C^{\mathrm{mean}}_{\mathrm{S,ran}}$\label{thm:full-characteriation-of-super-activation}]
Let $(\mathfrak W_i,\mathfrak V_i)_{i=1,2}$ be AVWCs.
\begin{enumerate}
\item If $C_{\mathrm{S}}(\mathfrak W_1,\mathfrak V_1)=C_{\mathrm{S}}(\mathfrak W_2,\mathfrak V_2)=0$, then the estimate
\begin{align}
C_{\mathrm{S}}(\mathfrak W_1\otimes\mathfrak W_2,\mathfrak V_1\otimes\mathfrak V_2)>0
\end{align}
is true if and only if $\mathfrak W_1\otimes\mathfrak W_2$ is not symmetrizable and $C^{\mathrm{mean}}_{\mathrm{S,ran}}(\mathfrak W_1\otimes\mathfrak
W_2,\mathfrak V_1\otimes\mathfrak V_2)>0$.\\
If $(\mathfrak W_i,\mathfrak V_i)_{i=1,2}$ can be super-activated it holds
\begin{align}
C_{\mathrm{S}}(\mathfrak W_1\otimes\mathfrak W_2,\mathfrak V_1\otimes\mathfrak V_2)=C^{\mathrm{mean}}_{\mathrm{S,ran}}(\mathfrak
W_1\otimes\mathfrak W_2,\mathfrak V_1\otimes\mathfrak V_2).
\end{align}
\item There exist AVWCs which exhibit the above behaviour.
\item If $C^{\mathrm{mean}}_{\mathrm{S,ran}}$ shows super-activation for $(\mathfrak W_1,\mathfrak V_1)$ and $(\mathfrak W_2,\mathfrak V_2)$, then
    $C_{\mathrm{S}}$ shows super-activation for $(\mathfrak W_1,\mathfrak V_1)$ and $(\mathfrak W_2,\mathfrak V_2)$ if and only if at
    least one of $\mathfrak W_1$ or $\mathfrak W_2$ is non-symmetrizable.
\item If $C^{\mathrm{mean}}_{\mathrm{S,ran}}$ shows no super-activation for $(\mathfrak W_1,\mathfrak V_1)$ and $(\mathfrak W_2,\mathfrak V_2)$ then
    super-activation of $C_{\mathrm{S}}$ can only happen if $\mathfrak W_1$ is non-symmetrizable and $\mathfrak W_2$ is symmetrizable
    and $C^{\mathrm{mean}}_{\mathrm{S,ran}}(\mathfrak W_1,\mathfrak V_1)=0$ and $C^{\mathrm{mean}}_{\mathrm{S,ran}}(\mathfrak W_2,\mathfrak V_2)>0$. The statement is
    independent of the specific labelling.
\end{enumerate}
\end{thm}
\begin{rem}\label{rem:super-activation-explanation}
Of course for $\mathfrak W_1\otimes\mathfrak W_2$ to be non-symmetrizable, it has to be that at least one out of $\mathfrak W_1$, $\mathfrak
W_2$ is non-symmetrizable.\\
While Theorem \ref{thm:full-characteriation-of-super-activation} offers a complete characterization, it does not give any explicit examples - fortunately this has already been done in \cite{bs}, where two AVWCs were used as follows: The first legal link is modeled by an AVC $\mathfrak W_1=(W_{1,1},W_{1,2})$ with input system for Alice being $\{1,2\}$ and output at Bob's site being $\{1,2,3\}$. The transition probabilities were given by
\begin{align}W_{1,1}=\left(\begin{array}{lll}1&0&0\\0&0&1\end{array}\right)^\top,\qquad W_{1,2}=\left(\begin{array}{lll}0&0&1\\0&1&0\end{array}\right)^\top\end{align}
(note that assume that the columns of a matrix representing a channel sum up to one, not the rows!) and the first link to the eavesdropper by $\mathfrak V_1=(V_1)$ (no influence from the jammer on that link). For the purpose of this example, it would even be sufficient to let $\mathfrak V_1=\mathfrak T$. This channel has the property that $\mathfrak W_1$ is symmetrizable. The second link was chosen to consist of two binary symmetric channels $W_2,V_2$ where $W_2$ was a degraded version of $V_2$, but both had nonzero capacity. Thus, $C_{\mathrm{S}}(\mathfrak W_2,\mathfrak V_2)=0$ but nontheless it is possible to transmit (non-secret) messages via $\mathfrak W_2$. This example fits into the third class of pairs of AVWCs described in the above Theorem \ref{thm:full-characteriation-of-super-activation}.\\
While this explicit example is very interesting, our analysis provides a more systematic analysis.\\
Note that all our arguments only apply to the strong secrecy criterion. The weak secrecy criterion can be handled differently, and will be the scope of future work.
\end{rem}
As a last point in this section, we would like to discuss connections between $C_\mathrm{pp}$ and $C_\mathrm{S}$. At first, let us observe a similarity: The former shows super-activation if and only if the latter shows super-activation. To see this, we argue as follows: By definition, the class of codes which transmit public and private messages as defined in Definition \ref{def:private/public_coding_scheme} includes that according to Definition \ref{def:uncorrelated-coding-scheme} where no public information is transmitted. Therefore it holds that $C_\mathrm{pp}(\mathfrak W,\mathfrak V)\geq C_\mathrm{S}(\mathfrak W,\mathfrak V)$ for all AVWCs $(\mathfrak W,\mathfrak V)$. Further, the definition of private/public codes according to Definition \ref{def:private/public_coding_scheme} is more narrow than the one of a common randomness assisted code according to Definition \ref{def:CR-assisted-code}, so that every private/public code is at the same time also a common randomness assisted code. Especially, the public messages may be treated as if they were common randomness if $L=\Gamma$.\\
Therefore, $C_\mathrm{pp}(\mathfrak W,\mathfrak V)>0$ implies that $C^\mathrm{mean}_\mathrm{S,ran}(\mathfrak W,\mathfrak V)>0$ for all $(\mathfrak W,\mathfrak V)$. We conclude from Theorem \ref{theorem:symmetrizability-properties-of-C1det} that $C_\mathrm{pp}(\mathfrak W,\mathfrak V)>0$ implies $C_\mathrm{S}(\mathfrak W,\mathfrak V)>0$ for all $(\mathfrak W,\mathfrak V)$. This leads us to conclude that
\begin{align}\label{eqn:iff-for-Cpp-and-CS}
\forall\ (\mathfrak W,\mathfrak V):\qquad C_\mathrm{pp}(\mathfrak W,\mathfrak V)>0\qquad\Leftrightarrow\qquad C_\mathrm{S}(\mathfrak W,\mathfrak V)>0.
\end{align}
Let now $C_\mathrm{pp}$ show super-actication on $((\mathfrak W_1,\mathfrak V_1),(\mathfrak W_2,\mathfrak V_2))$. Then it follows from the statement in equation (\ref{eqn:iff-for-Cpp-and-CS}) that both $C_\mathrm{S}(\mathfrak W_1,\mathfrak V_1)=C_\mathrm{S}(\mathfrak W_2,\mathfrak V_2)=0$ and $C_\mathrm{S}(\mathfrak W,\mathfrak V)>0$. Therefore, super-activation of $C_\mathrm{pp}$ implies super-activation of $C_\mathrm{S}$.\\
In the reverse direction, let $C_\mathrm{S}$ show super-activation on the pair $((\mathfrak W_1,\mathfrak V_1),(\mathfrak W_2,\mathfrak V_2))$. From the statement in equation (\ref{eqn:definition-of-symmetrizability}) we immediately see that $C_\mathrm{pp}$ shows super-activation as well.\\
Concerning differences, we note that a question we have to leave open is whether there could exist AVWCs $\mathfrak W,\mathfrak V$ such that $C_\mathrm{pp}(\mathfrak W,\mathfrak V)>C^\mathrm{mean}_\mathrm{S,ran}(\mathfrak W,\mathfrak V)$ holds.\\
This question is of huge practical importance, as it allows the quantification of the interplay between private and public communication in interfering networks when i.i.d. assumptions are not met, as is often the case.
\end{section}
\begin{section}{Proofs\label{sec:proofs}}
\begin{subsection}{Technical definitions and facts\label{subsec:technical}}
An important part of our results builds on the mathematical structure that was developed in \cite{csiszar-narayan}. The structure of the codes
developed there builds on randomly sampling codewords which are all taken from one and the same set $T_p$. In our previous paper we used an
approach that was built on sampling codewords according to some pruned distribution $p'$ defined by $p'(x^n):=\frac{1}{p^{\otimes
n}(T_{p,\delta}^n)}\mathbbm1_{T_{p,\delta}^n}\cdot p^{\otimes n}(x^n)$ for some $p\in\mathcal P(\mathcal X)$. The small deviation of $p'$ from $p^{\otimes n}$ brings with it some benefits concerning asymptotic estimates. Since this work uses the outcomes of the earlier work \cite{wiese-noetzel-boche-I}, it would be desirable to use exactly the same technical approach.\\
However, due to the intended connection to \cite{csiszar-narayan}, we cannot use $p'$ in this work. Instead, we decided to use the same distribution as the one which was used in \cite{csiszar-narayan} which is in some sense further away from $p^{\otimes n}$. While this ensures seamless connectivity to \cite{csiszar-narayan}, it also made us deviate (compared to for example our previous work \cite{wiese-noetzel-boche-I}) from standard formulations in some other points, namely: We use a different notion of conditional typicality than before, and we define typical sets using the relative entropy rather than the one-norm.\\
This deviation is motivated by the fact that, for any finite alphabet $\mathcal A$ and $p\in\mathcal A$ as well as type $\bar N\in\mathcal P_0^n(\mathcal A)$, we have $p^{\otimes n}(T_{\bar N})=\mathrm{poly}(n)2^{-n\cdot D(p\|\bar N)}$ for some polynomial $\mathrm{poly}$ in $n$. Therefore, defining typicality with respect to relative entropy gives the best control on the asymptotic behaviour of typical sets. All methods that use other distance measures for the definition of typicality need to relate these other measures to the relative entropy.\\
That the use of relative entropy is also elegant as compared to other methods can be seen as follows: Looking at \cite[Definition 2.9]{csiszar-koerner} (which deals with typicality in the presence of channels and inputs to those channels) one sees an additional advantage of using relative entropy over using one-norm: defining typicality with respect to variational distance requires one to add additional assumptions which are not necessary when relative entropy is used, as the latter quantity can become infinite.\\
More precisely, let us assume we are given a channel $W\in C(\mathcal A,\mathcal B)$ such and $(a^n,b^n)\in\mathcal A^n\times\mathcal B^n$ such that for one specific choice of $a,b$ we have $N(a,b|a^n,b^n)>0$ but $w(b|a)=0$. Then $b^n$ is not a typical output of the channel $w^{\otimes n}$ given that its input was $a^n$, since the probability that it is received when $a^n$ as sent is zero:
\begin{align}
0&\leq w^{\otimes n}(b^n|a^n)\\
&=\prod_{i=1}^nw(b_i|a_i)\\
&\leq\prod_{i:a_i=a,b_i=b}w(b_i|a_i)\\
&=w(b|a)^{N(a,b|a^n,b^n)}\\
&=0^{N(a,b|a^n,b^n)}\\
&=0.
\end{align}
Excluding non-typical sequences is crucial for the derivation of lower bounds on cardinality of the conditionally typical set, for example. Thus, the above sequence $b^n$ is excluded from the $w$-typical set given $a^n$ explicitly in \cite[Definition 2.9]{csiszar-koerner}.\\
A notion using relative entropy captures this perfectly as well, but without necessitating the explicit exclusion: Let us assume that $b^n$ is said to be $w$- typical given $a^n$ iff $\Omega(a^n,b^n):=D(\bar N(\cdot|a^n,b^n)\|p_{AB})$ satisfies $\Omega(a^n,b^n)\leq\delta$ for some $\delta>0$, where $p_{AB}(a,b):=\bar N(a)w(b|a)$. Then let $a^n$ be given and $b^n$ be such that there exists $a,b$ such that $N(a,b|a^n,b^n)>0$ but $w(b|a)=0$. It follows $p_{AB}(a,b)=0$ and therefore $D(\bar N(\cdot|a^n,b^n)\|p_{AB})=\infty$, so that $\Theta(x^n,y^n)=\infty$ and hence $b^n$ is not $w$-typical given $a^n$.\\
A brief look at robust typicality as defined in \cite{orlitsky-roche} shows that this quantity is also only related to relative entropy via inequalities.\\
Therefore, our definition achieves two goals: It connects in the most direct way to the relevant probability estimates and can be written down with minimal effort.\\
Thus, the sets which we will be using frequently in the following are, for arbitrary finite sets $\mathcal A,\mathcal B,\mathcal C$, every
$p\in\mathcal P(\mathcal A)$, $\tilde V\in C(\mathcal A\times\mathcal B,\mathcal C)$ and $\delta>0$ defined as follows: for a given $(a^n,b^n)\in\mathcal A^n\times\mathcal B^n$ we define $p_{ABC}\in\mathcal P(\mathcal A\times\mathcal B\times\mathcal C)$ via $p_{ABC}(a,b,c):=\bar N(a^n,b^n)\tilde v(c|a,b)$ and
\begin{align}
T^n_{p,\delta}&:=\{a^n\in\mathcal A^n:D(\bar N(\cdot|a^n)\|p)\leq\delta\},\\
T_{\tilde V,\delta}(a^n,b^n)&:=\{c^n:D(\bar N(\cdot|a^n,b^n,c^n)\|p_{ABC})\leq\delta\}.
\end{align}
These definitions are only valid for $\delta>0$. Each $T_{V,\delta}(s^n,x^n)$ obeys the estimate
\begin{align}
\tilde v^{\otimes n}(T_{\tilde V,\delta}(a^n)|a^n)\geq1-2^{-n\cdot\delta/2},
\end{align}
for all $n\in\mathbb N$ such that $|\mathcal A\times\mathcal B|\tfrac{1}{n}\log(2n)\leq\delta$. We set, for every $p\in\mathcal P(\mathcal X)$,
\begin{align}
E(p):=\max_{q\in\mathcal P(\mathcal S)}I(p;V_q)\qquad\mathrm{and}\qquad B(p):=\min_{q\in\mathcal P(\mathcal S)} I(p;W_q).
\end{align}
For the technical part of our proofs, the most important tool will be the Chernoff-Hoeffding bound:
\begin{lem}\label{lem:Chernoff}
Let $b$ be a positive number. Let $Z_1,\ldots,Z_L$ be i.i.d.\ random variables with values in $[0,b]$ and expectation $\mathbb EZ_l=\nu$, and
let $0<\varepsilon<\frac{1}{2}$. Then
\begin{align}
  \mathbb P\left\{\frac{1}{L}\sum_{l=1}^LZ_l\notin[(1\pm\varepsilon)\nu]\right\}\leq 2\exp\left(-L\cdot\frac{\varepsilon^2\cdot\nu}{3\cdot
  b}\right),
\end{align}
where $[(1\pm\varepsilon)\nu]$ denotes the interval $[(1-\varepsilon)\nu,(1+\varepsilon)\nu]$.
\end{lem}
The proof can be found in \cite[Theorem 1.1]{DP} and in \cite{AW}.
\end{subsection}
\begin{subsection}{Proof of the converse part of Theorem \ref{thm:coding-theorem} (coding theorem for
$C_{\mathrm{key}}$)\label{subsec:proofs-of-the-converse-parts}}
Main ingredients to this proof are Fano's inequality, data processing and almost-convexity of the entropy.
\begin{proof}[{\bf Proof of converse for secret common randomness assisted secrecy capacity}]
Let a sequence $\mathcal K=(\mathcal K_n)_{n=1}^\infty$ of common randomness-assisted codes be given such that for all $n\in\mathbb N$ we have
\begin{align}
\min_{s^n\in\mathcal S^n}\frac{1}{\Gamma_n\cdot
K_n}\sum_{\gamma,k=1}^{\Gamma_n,K_n}e^\gamma(x^n|k)w_{s^n}(D^\gamma_{k}|x^n)\geq1-\epsilon_n,\\
\max_{s^n\in\mathcal S^n}I(\mathfrak K_n;\mathfrak Z_{s^n})\leq\epsilon_n,
\end{align}
and of course $\limsup_{n\to\infty}\epsilon_n=0$. Set $R:=\liminf_{n\to\infty}\frac{1}{n}\log K_n$, and $G:=\lim_{n\to\infty}\frac{1}{n}\log \Gamma_n$. In addition to the random variable defined in Definition \ref{def:CR-assisted-code}, consider $(\mathfrak K_n,\mathfrak K'_{q,n},\mathfrak\Gamma_n)$ distributed as
\begin{align}
\mathbb P((\mathfrak K_n,\mathfrak Y_{q}^n,\mathfrak K'_{q,n},\mathfrak\Gamma_n)=(k,k',\gamma))=\sum_{s^n\in\mathcal S^n}q^{\otimes n}(s^n)\mathbb P(\mathfrak K_n,\mathfrak Y_{q,n},\mathfrak K'_n,\mathfrak\Gamma_n).
\end{align}
Then for all $n\in\mathbb N$, $q\in\mathcal P(\mathcal S)$ and $s^n\in\mathcal S^n$ Fano's inequality implies
\begin{align}
(1-\epsilon_n)\log K_n&\leq I(\mathfrak K_n;\mathfrak K'_{q,n}|\mathfrak\Gamma_n)-I(\mathfrak K_n;\mathfrak Z_{s^n})+1+\epsilon_n.
\end{align}
We can apply the data processing inequality to get
\begin{align}
(1-\epsilon_n)\log K_n&\leq I(\mathfrak K_n;\mathfrak Y^n_q|\mathfrak \Gamma_n)-I(\mathfrak K_n;\mathfrak Z_{s^n})+1+\epsilon_n,
\end{align}
and from e.g. Lemma 3.4 in \cite{csiszar-koerner} and independence of the random variables $\mathfrak K_n$ and $\mathfrak G_n$ it follows that the asymptotic scaling of the rate $\liminf_{n\to\infty}\frac{1}{n}\log K_n$ can be upper bounded through the following inequality:
\begin{align}
(1-\epsilon_n)\log K_n&\leq I(\mathfrak K_n;\mathfrak Y_q^n)-I(\mathfrak K_n;\mathfrak
Z_{s^n})+H(\mathfrak\Gamma_n)+1+\epsilon_n.
\end{align}
Since this estimate is valid for all $q\in\mathcal P(\mathcal S)$ and $s^n\in\mathcal S^n$ we get
\begin{align}
\log K_n&\leq \frac{1}{1-\epsilon_n}\left(\min_{q\in\mathcal P(\mathcal S)}I(\mathfrak K_n;\mathfrak Y_q^{ n})-\max_{s^n\in\mathcal S^n}I(\mathfrak K_n;\mathfrak
Z_{s^n})\right)+\frac{1+\epsilon_n}{1-\epsilon_n}+\frac{\log\Gamma_n}{1-\epsilon_n}.
\end{align}
Define the distribution $p\in\mathcal P([K_n])$ and the channel $U\in C([K_n],\mathcal X^n)$ by
\begin{align}
p(k)=\frac{1}{K_n},\qquad U(x^n|k):=\sum_{\gamma=1}^{\Gamma_n}\frac{1}{\Gamma}e^\gamma(x^n|k)\qquad(k\in[K_n],\ x^n\in\mathcal X^n).
\end{align}
Then we arrive at
\begin{align}
\log K_n&\leq\frac{1}{1-\epsilon_n}\left(\min_{q\in\mathcal P(\mathcal S)}I(p;W_q^{\otimes n}\circ U)-\max_{s^n\in\mathcal S^n}I(p; V_{s^n}\circ
U)\right)+\frac{1+\epsilon_n}{1-\epsilon_n}+\frac{\log\Gamma_n}{1-\epsilon_n}.
\end{align}
Of course, we can obtain a more relaxed upper bound by optimizing over all $p\in\mathcal P([K_n])$ and $U\in C([K_n],\mathcal X^n)$. We then
obtain (since $K_n\leq|\mathcal X^n|$ for every reliably working code and, therefore, $\mathcal P([K_n])\subset\mathcal P([|\mathcal X^n|])$ under the standard embedding $[K_n]\subset[|\mathcal X|^n]$) by further increasing the size of the input alphabet from $K_n$ to $|\mathcal X|^n$ with $\mathcal U_n:=[|\mathcal X|^n]$ that
\begin{align}
R&\leq\lim_{n\to\infty}\frac{1}{n}\max_{p\in\mathcal U_n}\max_{U\in C(\mathcal U_n,\mathcal X^n)}\left(\min_{q\in\mathcal
P(\mathcal S^n)}I(p;W_q^{\otimes n}\circ U)-\max_{s^n\in\mathcal S^n}I(p;V_{s^n}\circ U)\right)+G.
\end{align}
As it has been proven in \cite{wiese-noetzel-boche-I} that the capacity $C^{\mathrm{mean}}_{\mathrm{S,ran}}$ equals the leftmost part in the above sum we have proven the desired result.\\
Another obvious bound on the capacity arises by ignoring all security issues: since $\mathcal K$ ensures an asymptotically perfect
transmission, we have
\begin{align}
\lim_{n\to\infty}\frac{1}{n}\log K_n\leq\max_{p\in\mathcal P(\mathcal X)}\min_{q\in\mathcal P(\mathcal S)}I(p;W_q).
\end{align}
This establishes the converse part of the coding theorem.
\end{proof}
\end{subsection}
\begin{subsection}{Proof of the direct part of Theorem \ref{thm:coding-theorem} (coding theorem for $C_{\mathrm{key}}$)}
Let $G>0$ be given. Define $p:=\arg\max_{p\in\mathcal P(\mathcal X)}(B(p)-E(p))$. Set $G':=\max\{E(p),G\}$. Intuitively speaking, this is the amount of common randomness which can be put to use in the obfuscation of Eve. Choose a $\tau>0$ such that $\nu(\tau)$ from Lemma \ref{lemma:central-lemma-II} satisfies $\nu(\tau)<G'$. Let $n\in\mathbb N$ be so that for all $n\geq N$ there is $p_n\in\mathcal P_0^n(\mathcal X)$ such that $|B(p_n)-B(p)|\leq\max\{\tau,\nu(\tau)\}$ and $|E(p_n)-E(p)|\leq\max\{\tau,\nu(\tau)\}$. This can be achieved by approximating $p$ through types $p_n$ via Lemma \ref{lemma:types-are-dense} and since both $B$ and $E$ are continuous functions. Take three sequences $(K_n)_{n=1}^\infty$, $(L_n)_{n=1}^\infty$, $(\Gamma_n)_{n=1}^\infty$ of natural numbers. Without loss of generality, we can ensure that $(\Gamma_n)_{n\in\mathbb N}$ satisfies both
$\Gamma_n\leq2^{n\cdot G'}$ for all $n\in\mathbb N$ and $\lim_{n\to\infty}\frac{1}{n}\log\Gamma_n=G'$. Let now $n\in\mathbb N$ satisfying $n\geq N$ be fixed but large enough such that in addition
\begin{align}
E(p)-G'+4\tau\geq\frac{1}{n}\log(L_n)\geq E(p)-G'+2\tau,\\
B(p)-E(p)+G'-4(\tau+\nu(\tau))\geq\frac{1}{n}\log(K_n)\geq B(p)-E(p)+G'-2(\tau+\nu(\tau))
\end{align}
be satisfied, for all large enough $n\in\mathbb N$. This implies both
\begin{align}
\frac{1}{n}\log(K_n\cdot L_n)&\leq B(p)-E(p)+G'-4(\tau+\nu(\tau))+E(p)-G'+4\tau\\
&=B(p)-4\nu(\tau)\\
&\leq B(p_n)-\nu(\tau)
\end{align}
and
\begin{align}
\frac{1}{n}\log(L_n\cdot\Gamma_n)\geq E(p)+2\tau\geq E(p_n)+\tau.
\end{align}
Asymptotically, we also have this yields
\begin{align}
\liminf_{n\to\infty}\frac{1}{n}\log(K_n)\geq B(p)-E(p)+G-4\cdot(\tau+\nu(\tau)).
\end{align}
At the same time, the prerequisites of Lemma \ref{lemma:central-lemma-II} are met such that a reliable sequence of codes exists which is also secure with respect to $\|\cdot\|_1$: For all large enough $n\in\mathbb N$ we have
\begin{align}
\min_{s^n}\sum_{\gamma=1}^\Gamma\frac{1}{\Gamma}\sum_{k,l=1}^{K,L}\frac{1}{K\cdot L}w_{s^n}(D_{kl}^\gamma|\mathbf x_{kl\gamma})\geq1-2^{-n\cdot
\nu(\tau)},\\
\max_{s^n,k}\left\|\frac{1}{L\cdot\Gamma}\sum_{l,\gamma=1}^{L,\Gamma}v_{s^n}(\cdot|\mathbf x_{kl\gamma})-\mathbb
Ev_{s^n}(\cdot|X^n)\right\|_1\leq&2^{-n\cdot\nu(\tau)}.
\end{align}
It can already be seen that this yields reliable communication at any rate which is strictly below $B(p)-E(p)+G$ - we proved the achievability of rates close enough to $B(p)-E(p)+G$, but it is clear that time sharing between a trivial strategy where only one codeword is being transmitted (which is then automatically perfectly secure) and the strategy which was proven to work in the above will show achievability of all other rates $R\in[0,|B(p)-E(p)+G|^+]$. That we also get secure communication can be seen as follows: From \cite[Lemma 2.7]{csiszar-koerner} we know that our exponential bound (\ref{eqn:exponential-bound-on-variational-distance-in-lemma-II}) asymptotically leads to fulfillment of the strong secrecy criterion.\\
We have thus proven that, for each $\tau'>0$, the number
\begin{align}
\max_{p\in\mathcal P(\mathcal X)}\left(\min_{q\in\mathcal P(\mathcal S)}I(p;W_q)-\max_{q\in\mathcal P(\mathcal S)}I(p;V_q)\right)+G-\tau'
\end{align}
is an achievable rate. We now proceed by adding channels $U$ at the sender and using blocks of the original channels together: Since we now know that, for every $r\in\mathbb N$, $G>0$ and $\delta>0$, $p\in\mathcal P(\mathcal U_r)$ where $U_r:=[|\mathcal X|^{r}]$ and $U\in C(\mathcal U_r,\mathcal X^r)$ there exist sequences $\mathcal K=(\mathcal K_m)_{m=1}^\infty$ such that for every $s^{r\cdot m}\in(\mathcal S^r)^m=\mathcal S^{r\cdot m}$ we have
\begin{align}
\frac{1}{K_m}\frac{1}{\Gamma_m}\sum_{k=1}^{K_m}\sum_{\gamma=1}^{\Gamma_m}\sum_{x^{r\cdot m}}e(x^{r\cdot m}|k,\gamma)w_{s^{m\cdot r}}(D_{k}^\gamma|x^{r\cdot m})\geq1-\epsilon_m,
\end{align}
where $\mathbf x_{k,\gamma}\in U_r^m$ are codewords (each $x_{k,\gamma,i}$ is an element of $\mathcal U_r$) for $(W_{s^r}\circ U)_{s^r\in\mathcal S^r}$, and the stochastic encoder is $e(x^{r\cdot m}|k,\gamma)=\prod_{i=1}^m u(x_{ij}|x_{k,\gamma,i})$ for $x^{r\cdot m}$ and it holds that
\begin{align}
\liminf_{m\to\infty}\frac{1}{m}\log K_m\geq \min_{q\in\mathcal P(\mathcal S^r)}I(p;W_q\circ U)-\max_{s^r\in\mathcal S^r}I(p; V_{s^r}\circ
U)+r\cdot G-\delta.
\end{align}
We can define values $t_n\in\{0,\ldots,r-1\}$ by requiring $n=m\cdot r+t_n$ for them to hold for some suitably chosen $m=m(n)\in\mathbb N$.
This quantity satisfies $-1+n/r\leq m(n)\leq n/r$. For every $n\in\mathbb N$ we then define new decoding sets by
\begin{align}
\hat D_{k}^\gamma:=D_{k}^\gamma\times\mathcal Y^{t_n}
\end{align}
and new codewords by setting for some arbitrary but fixed $x^{t_n}$
\begin{align}
\hat x_{k\gamma}:=(x_{k\gamma},x^{t_n}).
\end{align}
From the choice of codewords and the decoding rule it is clear that this code is asymptotically reliable. The asymptotic number of codewords
(mind that $\hat K_n=K_{m(n)}$) calculated and normalized with respect to $n$, is
\begin{align}
\liminf_{n\to\infty}\frac{1}{n}\log\hat K_n&=\liminf_{n\to\infty}\frac{1}{m(n)\cdot r+t_n}\log K_{m(n)}\\
&\geq\liminf_{n\to\infty}\frac{1}{r}\cdot\frac{1}{m(n)+1}\log K_{m(n)}\\
&=\liminf_{n\to\infty}\frac{1}{r}\cdot\frac{1}{m(n)}\cdot\frac{m(n)}{m(n)+1}\cdot \log K_{m(n)}\\
&=\frac{1}{r}\liminf_{n\to\infty}\frac{1}{m(n)}\cdot \log K_{m(n)}\\
&=\frac{1}{r}\left(\min_{q\in\mathcal P(\mathcal S^r)}I(p;W_q\circ U)-\max_{s^r\in\mathcal S^r}I(p;V_{s^r}\circ U)+r\cdot G-\delta\right).
\end{align}
To see that every number $C^*(\mathfrak W,\mathfrak V)-\epsilon$ is an achievable rate, take $r$, $U$ and $p$ such that
\begin{align}
C^*(\mathfrak W,\mathfrak V)-\epsilon/2\leq\frac{1}{r}\left(\min_{q\in\mathcal P(\mathcal S^r)}I(p;W_q\circ U)-\max_{s^r\in\mathcal
S^r}I(p;V_{s^r}\circ U)\right).
\end{align}
This is possible since in \cite{wiese-noetzel-boche-I} it was (in addition to the equality $C^\mathrm{mean}_\mathrm{S,ran}(\cdot,\cdot)=C^*(\cdot,\cdot)$) proven that
\begin{align}
C^*(\mathfrak W,\mathfrak V)=\lim_{r\to\infty}\frac{1}{r}\max_{p\in\mathcal P(\mathcal U_n)}\max_{U_n\in C(\mathcal U,\mathcal X^n)}\left(\min_{q\in\mathcal P(\mathcal S^r)}I(p;W_q\circ U)-\max_{s^r\in\mathcal
S^r}I(p;V_{s^r}\circ U)\right).
\end{align}
We set $\delta=r\cdot\epsilon/4$. Then from our preceding arguments it becomes clear that there is a sequence $\hat{\mathcal K}$ of
asymptotically reliable codes at an asymptotic rate
\begin{align}
\liminf_{n\to\infty}\frac{1}{n}\log\hat K_n&\geq\frac{1}{r}\left(\min_{q\in\mathcal P(\mathcal S^r)}I(p;W_q^m\circ U)-\max_{s^r\in\mathcal
S^r}I(p;V_{s^r}\circ U)+r\cdot G-r\cdot\epsilon/4\right)\\
&\geq\frac{1}{r}\left(\min_{q\in\mathcal P(\mathcal S^r)}I(p;W_q^m\circ U)-\max_{s^r\in\mathcal S^r}I(p;V_{s^r}\circ U)+r\cdot
G\right)-\epsilon/4\\
&\geq C^*(\mathfrak W,\mathfrak V)+G-\epsilon/2-\epsilon/4\\
&\geq C^*(\mathfrak W,\mathfrak V)+G-\epsilon.
\end{align}
This proves the direct part of the coding theorem.
\end{subsection}
\begin{subsection}{An intermediate result}
We now have to prove the core results from which all the other statements can be deduced. The idea of proof will be to make a random selection
of the codewords $\mathbf x_{kl\gamma}$ where $k$ are the messages, $l$ are non-secret messages which are only being sent in order to obfuscate the received signal at Eve, and $\gamma$ are the values of the common randomness. When applying the results to AVWCs, the decoder is the one defined in
\cite{csiszar-narayan} whenever we study $C_\mathrm{S}$ and is defined here according to our needs for the study of $C_{\mathrm{key}}$.\\
We define events $E_1,\ldots,E_5$ which describe certain desirable properties of our codewords, in dependence of $(\mathfrak W,\mathfrak V)$
and the numbers $K,L,\Gamma$ of available indices $k,l,\gamma$. We then use Chernoff bounds. This guarantees that the random selection of
codewords has each single property we would like them to have with probability lower bounded by $1-\exp(-2^{nc})$ for some positive constant
$c>0$ and all large enough $n$ under some conditions on $\Gamma,L$ and $K$ which of course depend on $(\mathfrak W,\mathfrak V)$ as well.
Application of a union bound then reveals the existence of one particular choice of codewords that has all the desired properties
simultaneously.\\
Using exactly this method of proof, Csiszar and Narayan \cite[Lemma 3]{csiszar-narayan} proved properties
$(\ref{eqn:property-1-of-central-lemma}), (\ref{eqn:property-2-of-central-lemma})$ and $(\ref{eqn:property-4-of-central-lemma})$ of Lemma
\ref{lemma:central-lemma}. Thus what remains for us is to provide proof that the remaining event $(\ref{eqn:property-5-of-central-lemma})$ has high probability.\\
In \cite{csiszar-narayan}, large deviation results for dependent random variables were employed, but the underlying probability employed in
codeword selection was the same as the one used by us, so that our findings connect seamlessly.
\\\\
We become a bit more concrete now. Let $p\in\mathcal P_0^n(\mathcal A),\ q\in\mathcal P_0^n(\mathcal S)$. Throughout, we will attempt to twist
and tweak asymptotic quantities such that they are calculated with respect to the random variables $(S,X,Z)$ defined via $\mathbb
P((S,X,Z)=(s,x,z)):=p(x)q(s)v(z|x,s)$. Since the distribution of $(S,X,Z)$ is so important, we label it by $p_{SXZ}$. The variable $p$ will remain fixed, and $q$ will always denote a type corresponding to one of the
choices of James.\\
The proof will require us to draw codewords at random. As stated already, we adapt this procedure to the one chosen in \cite{csiszar-narayan}.
This is done as follows: We define the random variables $X_{kl\gamma}$ ($1\leq k\leq K$, $1\leq l\leq L$, $1\leq\gamma\leq\Gamma$) by $\mathbb
P(X_{kl\gamma}=x^n):=\frac{1}{|T_p|}\mathbbm1_{T_p}(x^n)$ for all $k\in[K],\ l\in[L]\ \gamma\in[\Gamma]$ and $x^n\in\mathcal X^n$, where
$K,L,\Gamma$ are natural numbers. We write $\mathbf x_{kl\gamma}$ for the realizations of the variable $X_{kl\gamma}$, instead of
$x^n_{kl\gamma}$. The random variable $\mathbf X:=(\mathbf X_{kl\gamma})_{k,l,\gamma=1}^{K,L,\Gamma}$ is distributed such that each $\mathbf X_{kl\gamma}$ is
independent of $\mathbf X_{k'l'\gamma'}$ if $(k,l,\gamma)\neq(k',l',\gamma')$. The realizations of $(\mathbf X_{kl\gamma})_{k,l,\gamma=1}^{K,L,\Gamma}$ are
written $\mathbf x$. We use the projections $\pi_{kl\gamma}$ defined by $\pi_{kl\gamma}(\mathbf x):=\mathbf x_{kl\gamma}$. Further projections
as e.g. $\mathbf x_\gamma:=\pi_\gamma(\mathbf x):=(\mathbf x_{kl\gamma})_{k,l=1}^{K,L}$ are defined wherever there is a need.\\
In order to enhance readability, we will not only omit the superscript $n$ in our codewords, but from time to time we will also write
statements like $\forall s^n$, property $P$ holds. Then, it is understood that $P$ holds for all $s^n\in\mathcal S^n$.\\
When calculating expectations of any of the $X_{kl\gamma}$ we need no reference to $k,l,\gamma$ due to independence of our random variables. We
therefore add another random variable, $X^n$, distributed as $\mathbb P(X^n=x^n)=\frac{1}{|T_p|}\mathbbm1_{T_p}(x^n)$ as well.\\
A first and crucial step for all that is to come in the proofs of the technical Lemmas \ref{lemma:central-lemma-II} and
\ref{lemma:central-lemma} is to fix some $\delta>0$ and $p\in\mathcal P(\mathcal X)$ and define, for all $s^n\in\mathcal S^n$ and
$z^n\in\mathcal Z^n$, the functions $\Theta_{s^n,z^n}:\mathcal X^n\to[0,b]$ (where $b:=2^{-n(H(Z|X,S)-f_2(n,\delta))}$ for some function $f_2$, as we will see soon) by
\begin{align}\label{eqn:definition-of-theta}
M(s^n,z^n)&:=\{x^n\in T_p:D(\bar N(\cdot|s^n,x^n,z^n)\|p_{SXZ})\leq\delta\}\\
\Theta_{s^n,z^n}(x^n)&:=v^{\otimes n}(z^n|s^n,x^n)\mathbbm1_{M(s^n,z^n)}.
\end{align}
In order to enhance readability, the dependence of both $M$ and $\Theta$ on $\delta$ is suppressed here and in the following. All our proofs rely on a common strategy, which only deviates in one point: The codes which ensure reliable transmission. For non-symmetrizable
AVWCs we rely on the work \cite{csiszar-narayan} and use the codes which are defined therein. This will be sufficient to obtain all the results that we claimed for the uncorrelated coding secrecy capacity.\\
The coding theorem for secret common randomness assisted secrecy capacity needs an additional definition of codes. This definition is as follows:\\
For every $n\in\mathbb N$, set $\Xi_n:=\mathcal P_0^n(\mathcal S)$. For every $x^n$, define (not necessarily
disjoint) ``decoding'' sets by
\begin{align}\label{eqn:definition-of-very-optimistic-decoding-set}
\hat D_{x^n}:=\bigcup_{\xi\in\Xi_n}T_{W_\xi,\delta}(x^n)
\end{align}
and for a collection $\mathbf x_\gamma:=(x_{kl\gamma})_{k,l=1}^{K,L}$ of codewords with fixed value of $\gamma$ set
\begin{align}\label{eqn:definition-of-realistic-decoding-set}
D(\mathbf x_\gamma)_{kl}:=\hat D_{x_{kl\gamma}}\bigcap\left(\bigcup_{k'\neq k}\bigcup_{l'\neq l}\hat D_{x_{k'l'\gamma}}\right)^\complement.
\end{align}
This defines the code $\mathcal K_n$. This definition allows the decoder to decode the randomization index $l$ as well, an approach which works for AVWCs and compound (wiretap) channels with convex
state sets via the minimax theorem. Note that this code will only ensure reliable transmission if $\Gamma$ is sufficiently large.\\
In order to deliver a joint treatment of the subject it makes sense to define the following events, where we implicitly assume a functional
dependence $\delta=\delta(\tau)$ that will be specified more exactly later during our proofs. The sets $E_3,\ldots,E_5$ depend only on $\tau$, whereas $E_1$ depends also on $\mathfrak W$ and $E_2$ on $\mathfrak V$.
\begin{align}
E_1&:=\left\{\mathbf x|\forall s^n,z^n,k:\ \frac{1}{L\cdot\Gamma}\sum_{l,\gamma=1}^L\Theta_{s^n,z^n}(\mathbf x_{kl\gamma})\in[(1\pm2^{-n\cdot\tau/4})\mathbb E\Theta_{s^n,z^n}]\right\}\\
E_2&:=\left\{\mathbf x|\min_{s^n}\frac{1}{\Gamma}\sum_{\gamma=1}^\Gamma d_{s^n}(\mathcal K_\gamma)\geq1-2\cdot2^{-n\delta/4}\right\}.\\
E_3&:=\left\{\mathbf x|\max_{\gamma,x^n,s^n}|\{(k,l):(x^n,\mathbf x_{kl\gamma},s^n)\in T_{\bar N(\cdot|x^n,\mathbf
x_{kl\gamma},s^n)}\}|\leq2^{n(|R-I(\mathbf x_{kl\gamma};x^n,s^n)|^++\tau)}\right\}\\
E_4&:=\left\{\mathbf x|\max_{\gamma,s^n}|\{(k,l):I(\mathbf x_{kl\gamma};s^n)>\tau\}|\leq K\cdot L\cdot2^{-n\cdot\tau}\right\}\\
E_5&:=\left\{\mathbf x|\max_{\gamma,s^n} \left|\left\{(k,l,\gamma):\begin{array}{l}\mathrm{There\ is}\ (k',l',\gamma')\neq (k,l,\gamma)\
\mathrm{such\ that}\\
I(\mathbf x_{kl\gamma};\mathbf x_{k'l'\gamma'},s^n)-|R-I(\mathbf x_{kl\gamma};s^n)|^+>\tau
\end{array}\right\}\right|\leq K\cdot L\cdot 2^{-n\cdot\tau/2}\right\}
\end{align}
The average success probability $d_{s^n}(\mathcal K_\gamma)$ was defined in Definition \ref{def:CR-assisted-code}. The events $E_3,E_4,E_5$ are proven to have high probability in \cite{csiszar-narayan} (actually, their proof is valid for $|\Gamma|=1$ but can be extended to arbitrary $|\Gamma|$ by simple union bounds, which leads to the following statement:
\begin{lem}[Cf. \cite{csiszar-narayan}]\label{lemma:csiszar-narayan-result}
There is $c'>0$ such that, if $\mathfrak W$ is non-symmetrizable, we have that
\begin{align}
\mathbb P(E_3\cap E_4\cap E_5)\geq1-\Gamma\cdot\exp(-2^{n\cdot c'})
\end{align}
\end{lem}
The bound in Lemma \ref{lemma:csiszar-narayan-result} is trivial whenever $\Gamma>\exp(2^{n\cdot c'})$. In the applications intended here, the maximal scaling of $\Gamma$ with $n$ will be exponential, so that nontrivial bounds arise.\\
Our main effort in the following will be to show that a similar bound is true for $\mathbb P(E_1)$ and $\mathbb P(E_2)$ under
the right conditions on $K,L$ and $\Gamma$. With respect to these conditions, any of the intersections $E_i\cap\ldots\cap E_j$ will then have very high probability as well.\\
For the proofs of both Lemma \ref{lemma:central-lemma-II} and \ref{lemma:central-lemma} it will be of importance to control the amount of
information which leaks out to Eve. This will require us to prove that a careful random choice of codewords will be provably secure, and this
is the main content of the following Lemma (which contains statements concerning the message transmission capabilities of the common randomness assisted
codes defined in (\ref{eqn:definition-of-very-optimistic-decoding-set}) and (\ref{eqn:definition-of-realistic-decoding-set}) as well).
\begin{lem}\label{lemma:intermediate-secrecy-result-I} Let $K,L,\Gamma\in\mathbb N$. Let the random variable $\mathbf X$ be as described above.
Then for every $\tau>0$ and $\beta>0$ there is a $\delta>0$ and and $N\in\mathbb N$ such that for all $n\geq N$ and types $p\in\mathcal
P_0^n(\mathcal X)$, the following statements are true:
\begin{enumerate}
\item If $\tfrac{1}{n}\log(L\cdot\Gamma)\geq E(p)+\tau$ and $\min_{x:p(x)>0}p(x)\geq\beta$, then
    $\begin{aligned}\label{eqn:applicable-version-of-chernoff-bound}
\mathbb P(E_1)\geq1-2\cdot|\mathcal S\times\mathcal X\times\mathcal Z|^n\cdot\exp(-2^{n\cdot\tau/6})
\end{aligned}$.
\item If $\frac{1}{n}\log(K\cdot L)\leq B(p)-\delta-2\cdot f_1(\sqrt{2\cdot\delta})$ then $\begin{aligned}\mathbb
    P(E_2)\geq1-\exp(n\cdot\log(|S|)-\Gamma\cdot2^{-n\delta})\end{aligned}$.
\item For every $\beta>0$, $|\mathcal X|$, $|\mathcal S|$ and $|\mathcal Z|$, a functional dependence between $\delta$ and $\tau$ can be
    chosen such that $\lim_{\tau\to0}\delta(\tau)=0$.
\end{enumerate}
The number $N$ depends on $|\mathcal X|$, $|\mathcal S|$, $|\mathcal Z|$ as well as on $p$ (via the quantity $\beta:=\min_{x\in\mathcal
X:p(x)>0}p(x)$) and on $\delta$.
\end{lem}
\begin{proof}
Some of the statements we wish to prove here are not about the full random variable $\mathbf X=(X_{kl\gamma})_{k,l,\gamma=1}^{K,L,\Gamma}$ but
only about exponentially many parts of it. We do therefore feel the need to write a few lines concerning our strategy of proof. We adopt the
usual point of view that $\mathbf X$ somehow generates matrices of codewords. In the special case treated here it will be convenient to think
of realizations of $\mathbf X$ as a list of $\Gamma$ matrices, all of which describe a code-book and each of these code-books uses the index
$l$ solely for making Eve obfuscated, while $k$ is used to transmit messages. The fact that $\gamma$ is known to both the sender and the
receiver lets the receiver adapt his decoder appropriately, while Eve only sees the average over all code-books. The effective randomness used
for obfuscation of Eve is therefore $L\cdot\Gamma$.\\
Before making this more precise, we need additional notation:\\
As stated already, the projections $\pi_{kl\gamma}:(\mathcal X^n)^{KL\Gamma}\to\mathcal X^n$ project onto the copy of $X^n$ corresponding to
$k,l,\gamma$, such that $\pi_{kl\gamma}(\mathbf X)=X_{kl\gamma}$. Accordingly, $\pi_k$ are the projections mapping $\mathbf X$ to $\mathbf
X_k:=(X_{kl\gamma})_{l,\gamma=1}^{L,\Gamma}$.\\
The trick will be to first understand how to embed statements concerning only certain projections of $\mathbf X$ into the whole random
selection process. The idea is to proceed as follows:\\
Take any set of functions $g_1,\ldots,g_M:\mathcal X^n\to[0,b']$. Then for all $k\in[K]$,
\begin{align}\label{eqn:lhs-and-rhs-probabilistic-statements}
\mathbb P(\frac{1}{\Gamma\cdot L}\sum_{l,\gamma=1}^{L,\Gamma}g_m(\pi_{kl\gamma}(\mathbf X))\notin[(1\pm\epsilon)\mathbb E g_m])&=\mathbb
P(\frac{1}{L\cdot\Gamma}\sum_{l,\gamma=1}^{L,\Gamma}g_m(\pi_{l\gamma}(\mathbf X_k))\notin[(1\pm\epsilon)\mathbb E g_m]),
\end{align}
where the left hand side is a probabilistic statement about $\mathbf X=(X_{kl\gamma})_{k,l,\gamma=1}^{K,L,\Gamma}$ and the right hand side is a
statement about the random variables $\mathbf X_k=(X_{kl\gamma})_{l,\gamma=1}^{L,\Gamma}$. Thus by the usual Chernoff bound Lemma \ref{lem:Chernoff} we have
\begin{align}\label{eqn:the-usual-union-bound-trick}
\mathbb P\left(\exists\ m,k:\ \frac{1}{L\cdot\Gamma}\sum_{l,\gamma=1}^{L,\Gamma}g_m(\pi_{l\gamma}(\mathbf X_k))\notin[(1\pm\epsilon)\mathbb E g_m]\right)\leq 2\cdot
M\cdot K\cdot \exp\left(-\frac{L\cdot\Gamma\cdot\epsilon^2\cdot\min_m\mathbb E g_m}{3\cdot b'}\right).
\end{align}
Another crucial connection in what is to follow is that for all $z^n,x^n$ and $s^n$ we have (using the abbreviation $N(\cdot):= N(\cdot|s^n,x^n,z^n)$ and $r(z|x,s):= N(s,x,z)/N(s,x|s^n,x^n)$):
\begin{align}\label{eqn:asymptotics-of-v-otimes-n(z^n|s^n,x^n)}
v^{\otimes n}(z^n|s^n,x^n)&=2^{n\cdot\sum_{s,x,z}\bar N(s,x,z)\log v(z|s,x)}\\
&=2^{n\cdot(\sum_{s,x,z}\bar N(s,x,z)(\log\frac{v(z|s,x)p(x)q(s)}{\bar N(s,x,z)}+\log\frac{\bar N(s,x,z)}{p(x)q(s)}))}\\
&=2^{n\cdot(-D(\bar N(\cdot|s^n,x^n,z^n)\|p_{XSZ})+\sum_{s,x,z}\bar N(s,x,z)\log(\frac{\bar N(s,x|s^n,x^n)\cdot r(z|x,s)}{p(x)\cdot q(s)})}\\
&=2^{n\cdot(-D(\bar N(\cdot|s^n,x^n,z^n)\|p_{XSZ})+D(\bar N(\cdot|s^n,x^n)\|p\otimes q)-H(\hat Z|\hat S,\hat X))},
\end{align}
where $\hat S\hat X\hat Z$ is distributed according to $\bar N$ (note that without loss of generality we may assume that $p,q>0$ here and in the following lines, since otherwise we could simply erase a symbol from the alphabet $\mathcal X$ or $\mathcal S$).
\\\\
{\bf Proof of property 1 of Lemma \ref{lemma:intermediate-secrecy-result-I}:} Let $n\in\mathbb N$. Replace $M$ with $\mathcal S^n\times\mathcal Z^n$ and the functions $g_m$ with the $\Theta_{s^n,z^n}$'s. We let $\delta>0$ be arbitrary for the moment. Using equation (\ref{eqn:asymptotics-of-v-otimes-n(z^n|s^n,x^n)}) and the fact that the relative entropy is never negative it can be seen that each $\Theta_{s^n,z^n}$ obeys
\begin{align}
\Theta_{s^n,z^n}(x^n)&=2^{n\cdot(-D(\bar N(\cdot|s^n,x^n,z^n)\|p_{XSZ})+D(\bar N(\cdot|s^n,x^n)\|p\otimes q)-H(\hat Z|XS))}\mathbbm1_{M(s^n,z^n)}(x^n)\\
&\leq2^{-n\cdot(H(\hat Z|\hat S\hat X)-D(\bar N(\cdot|s^n,x^n)\|p\otimes q))}.
\end{align}
This bound does obviously still depend on $x^n$. But if $x^n\in M(s^n,z^n)$ then the distribution of $\hat S\hat X\hat Z$ has the following important feature: by Pinsker's inequality, we have
\begin{align}\label{eqn:application-of-pinsker}
\|\bar N-p_{SXZ}\|_1\leq\sqrt{2\delta}.
\end{align}
Setting $f_2(\delta):=2\cdot f_1(\sqrt{2\delta})+\delta$, an application of Lemma \ref{lemma:continuity-of-conditional-entropy-with-respect-to-averaged-norm} from the appendix together with monotonicity of the relative entropy then yields
\begin{align}
\forall\ x^n\in\mathcal X^n:\qquad\Theta_{s^n,z^n}(x^n)&\leq2^{-n\cdot(H(Z|XS)-f_2(\delta))}.
\end{align}
Here, $f_1$ is defined setting $\mathcal A=\mathcal S\times\mathcal X\times\mathcal Z$. This justifies our choice of $b$. Note that the definition of $\Theta$ together with the monotonicity of $D(\cdot\|\cdot)$ ensures
that the empirical distribution $\bar N(\cdot|x^n,s^n)$ is almost product ($\bar N(\cdot,\cdot|x^n,s^n)\approx p(\cdot)\cdot\bar N(\cdot|s^n)$) and that this property was vital in the derivation of the results contained in \cite{csiszar-narayan}, whereas it may not be strictly necessary here (but does lead to a valid strategy of proof, nonetheless).\\
In order to apply the Chernoff bound we also need to calculate the expectation of each $\Theta_{s^n,z^n}$, and for that matter it will be
important to obtain a tight enough lower bound on $|M(s^n,z^n)|$: According to Lemma \ref{lemma:cardinality-bound} from the appendix (set $\mathcal A=\mathcal X$ and $\mathcal B=\mathcal S\times\mathcal Z$ there) we have
\begin{align}
|M(s^n,z^n)|\geq2^{n(H(\hat X|\hat S\hat Z)-f_C(n))}.
\end{align}
We are now almost ready to give a lower bound on the expectation of $\Theta_{s^n,z^n}$. Be aware that $s^n$ of type $q$ and $z^n$ remain fixed quantities for the moment. From monotonicity of the relative entropy and Pinsker's inequality applied together with Lemma \ref{lemma:continuity-of-conditional-entropy-with-respect-to-averaged-norm} it follows that we can estimate
\begin{align}
x^n\in M(s^n,z^n)\qquad\Rightarrow\qquad v^{\otimes n}(z^n|s^n,x^n)\geq2^{-n(H(Z|X,S)+2\delta+f_1(\sqrt{2\delta}))}.
\end{align}
It then follows that, if $M(s^n,z^n)\neq\emptyset$, we have the estimate
\begin{align}
\mathbb E\Theta_{s^n,z^n}&=\frac{1}{|T_p|}\sum_{x^n\in M(s^n,z^n)}v^{\otimes
n}(z^n|x^n,s^n)\\
&\geq2^{-n(H(Z|X,S)+2\delta+f_1(\sqrt{2\delta}))}\cdot2^{n(H(X)-f_C(n))}|M(s^n,z^n)|.
\end{align}
Estimate (\ref{eqn:application-of-pinsker}) together with the continuity of entropy yields (see \cite[Lemma 2.7]{csiszar-koerner})
\begin{align}
M(s^n,z^n)\neq\emptyset\qquad\Rightarrow\qquad|M(s^n,z^n)|\geq2^{n(H(X|S,Z)+f_C(n)+f_1(\sqrt{2\delta}))}.
\end{align}
We define $m:\mathcal S^n\times\mathcal Z^n\to\{0,1\}$ by $m(s^n,z^n)=1$ if $M(s^n,z^n)\neq\emptyset$ and $m(s^n,z^n)=0$ else. It then follows that for all large enough $n\in\mathbb N$
\begin{align}
\mathbb E\Theta_{s^n,z^n}&\geq m(s^n,z^n)\cdot2^{-n(H(Z|S)-f_3(n,\delta))},
\end{align}
where $f_3(\delta):=4(\delta+f_1(\sqrt{2\delta})$. For our random variable $\mathbf X$ this can be used as follows: via the Chernoff bound,
\begin{align}
\mathbb P(&\exists k,\ s^n,\ z^n:\ \frac{1}{L\cdot\Gamma}\sum_{l,\gamma=1}^{L,\Gamma}\Theta_{s^n,z^n}(\pi_{kl\gamma}(\mathbf
X))\notin[(1\pm\epsilon)\mathbb E\Theta_{s^n,z^n}])\label{eqn:lower-bound-on-prob-1}\\
&\leq2\cdot|\mathcal S\times\mathcal X\times\mathcal Z|^n\cdot \exp\left(-\epsilon^2\cdot L\cdot\Gamma\cdot\frac{\min_{s^n,z^n}\mathbb E\Theta_{s^n,z^n}}{3\cdot
b}\right)\\
&=c(n)\cdot \exp\left(-\epsilon^2\cdot \Gamma\cdot L\cdot\frac{\min_{s^n,z^n}\mathbb E\Theta_{s^n,z^n}}{3\cdot b}\right),
\end{align}
on account of the same argument that we used in equations (\ref{eqn:lhs-and-rhs-probabilistic-statements}) and
(\ref{eqn:the-usual-union-bound-trick}) and with the obvious definition of $c(n)$. Now we have to plug in the asymptotic behaviour of $L\cdot\Gamma$,
$\epsilon$ and $b$. If $m(s^n,z^n)=0$ then the statement is trivial. We set $f(\delta):=f_2(\delta)+f_3(\delta)$, $E(p):=\max_qI(p;V_q)$ and let $\frac{1}{n}\log L\cdot\Gamma\geq E(p)+\tau$ for some
$\tau>0$. Note that, no matter what the distribution of $S$ (which depends on the choice $s^n$ of James!), we have $E(p)-I(X;Z|S)\geq0$. Therefore,
\begin{align}
\frac{\epsilon^2}{3}\cdot L\cdot\Gamma\cdot\frac{\mathbb
E\Theta_{s^n,z^n}}{b}&\geq m(s^n,z^n)\frac{\epsilon^2}{3}\cdot2^{n(E(p)+\tau-H(Z|S)+H(Z|X,S)-f_2(\delta)-f_3(\delta))}\\
&=m(s^n,z^n)\frac{\epsilon^2}{3}\cdot2^{n(E(p)+\tau-I(Z;X|S)-f(\delta))}\\
&\geq m(s^n,z^n)\frac{\epsilon^2}{3}\cdot2^{n(E(p)+\tau-E(p)-f(\delta))}\\
&=m(s^n,z^n)\frac{\epsilon^2}{3}\cdot2^{n(\tau-f(\delta))}.
\end{align}
Upon choosing $\epsilon=2^{-n\cdot\alpha}$ we get a doubly exponential decay of the probability in equation (\ref{eqn:lower-bound-on-prob-1})
if $0>\tau-2\alpha-f(\delta)$, and since $\lim_{\delta\to0}f(\delta)=0$ there is a combination of $\delta>0$, $\tau>0$
such that for $\alpha=\tau/6$ and all large enough $n\in\mathbb N$ we have
\begin{align}\label{eqn:applicable-version-of-chernoff-bound-proven}
\mathbb P\left(\exists\ k,s^n,z^n:\ \frac{1}{L\cdot\Gamma}\sum_{l,\gamma=1}^{L,\Gamma}\Theta_{s^n,z^n}(\mathbf
x_{kl\gamma})\notin[(1\pm2^{-n\tau/6})\mathbb E\Theta_{s^n,z^n}]\right)\leq c(n)\cdot\exp(-2^{n\cdot\tau/6}).
\end{align}
It is clear that this defines a dependence $\delta=\delta(\tau)$ and that $\lim_{\tau\to0}\delta(\tau)=0$ and $\delta(\tau)>0$ for all (small enough) $\tau$. A specific choice that we will use here is $\delta(\tau)=\tau$.
\\\\
{\bf Proof of statement 2 of Lemma \ref{lemma:intermediate-secrecy-result-I}:} We will need Ahlswede's robustification technique.
\begin{lem}[\cite{ahlswede-coloring, ahlswede-gelfand-pinsker}]\label{lemma:ahlswede-robustification}
  If a function $f:\mathcal S^n\rightarrow [0,1]$ satisfies
\begin{equation}\label{eq:condRT}
  \sum_{s^n\in \mathcal S^n}f(s^n)q(s_1)\cdots q(s_n)\geq 1-\varepsilon
\end{equation}
for all $q\in\mathcal P_0^n(\mathcal S)$ and some $\varepsilon\in[0,1]$, then
\begin{equation}\label{eq:conclRT}
  \frac{1}{n!}\sum_{\pi\in\Pi_n}f(\pi(s^n))\geq 1-3\cdot(n+1)^{\lvert\mathcal S\rvert}\cdot\varepsilon.
\end{equation}
\end{lem}
We will in the following make use of the codes $\mathcal K_\gamma$ which defined the set $E_2$.\\
We would like to use the Chernoff bound for the variable $\Gamma$, so we have to control the expectation for each fixed $\gamma$. Note that the
construction of codes is such that it is independent from $\gamma$, so this will not turn into a hopeless case if we draw an independent number
$\Gamma$ of realizations of above codes. We go as follows: First associate to any given choice $\mathbf x_\gamma=(\mathbf
x_{kl\gamma})_{k,l=1}^{K,L}$ of codewords the corresponding code $\mathcal K(\mathbf x_\gamma)$ as defined in equations
(\ref{eqn:definition-of-very-optimistic-decoding-set}) and (\ref{eqn:definition-of-realistic-decoding-set}). Then, for every $s^n$ and
$\gamma\in[\Gamma]$, define the success probability of that code via
\begin{align}
d_{s^n}(\mathbf x_\gamma):=\sum_{k,l=1}^{K,L}\frac{1}{K\cdot L}w_{s^n}(D(\mathbf x_\gamma)_{kl}|\mathbf x_{kl\gamma}).
\end{align}
We then have for each fixed $\gamma$
\begin{align}
\mathbb E d_{s^n}(\mathbf X_\gamma)&=\mathbb E\frac{1}{K\cdot L}\sum_{k,l=1}^{K,L}w_{s^n}(D(\mathbf X_\gamma)_{kl}|\mathbf X_{kl\gamma})\\
&\geq\mathbb E\frac{1}{K\cdot L}\sum_{k,l=1}^{K,L}\left(w_{s^n}(\hat D_{\mathbf X_{kl\gamma}}|\mathbf X_{kl\gamma})-w_{s^n}(\bigcup_{k'\neq
k}\bigcup_{l'\neq l}\hat D_{\mathbf X_{k'l'\gamma}}|\mathbf X_{kl\gamma})\right)\\
&\geq\sum_{x^n\in T_p}\frac{1}{|T_p|}w_{s^n}(\hat D_{x^n}|x^n)-K\cdot L\cdot \sum_{x^n,\hat x^n\in T_p}\frac{1}{|T_p|^2}w_{s^n}(\hat
D_{x^n}|\hat x^n).
\end{align}
Now observe that $\pi(T_p)=T_p$ for every $\pi\in S_n$ and that, for all $\pi\in S_n$, $x^n,y^n$ and $s^n$ we have
$w_{s^n}(\pi(y^n)|\pi(x^n))=w_{\pi^{-1}(s^n)}(y^n|x^n)$. In addition to that, $\hat D_{\pi(x^n)}=\pi(\hat D_{x^n})$, so that we can write
\begin{align}
\mathbb E d_{s^n}(\mathbf X_\gamma)&\geq\frac{1}{n!}\sum_{\pi\in S_n}w_{s^n}(\hat D_{\pi(x^n)}|\pi(x^n))-K\cdot L\cdot \sum_{x^n,\hat x^n\in
T_p}\frac{1}{|T_p|^2}w_{s^n}(\hat D_{x^n}|\hat x^n).
\end{align}
By Lemma 2.3 and equation (2.1) in \cite{csiszar-koerner}, the density $\frac{1}{|T_p|}\mathbbm1_{T_p}$ satisfies
\begin{align}\label{eqn:upper-bound-on-type-distribution}
\frac{1}{|T_p|}\mathbbm1_{T_p}&\leq(n+1)^{|\mathcal X|}2^{-n\cdot H(p)}\mathbbm1_{T_p}\\
&=(n+1)^{|\mathcal X|}p^{\otimes n}\mathbbm1_{T_p}\\
&\leq(n+1)^{|\mathcal X|}p^{\otimes n}.
\end{align}
Setting $\mathrm{pl}(n):=(n+1)^{|\mathcal X|}$, we use this to further develop our bound as follows:
\begin{align}
\mathbb E d_{s^n}(\mathbf X_\gamma)&\geq\frac{1}{n!}\sum_{\pi\in S_n}w_{s^n}(\hat D_{\pi(x^n)}|\pi(x^n))-K\cdot L\cdot \mathrm{pl}(n)\cdot
\sum_{x^n\in T_p}\frac{1}{|T_p|}w_p^{\otimes n}(\hat D_{x^n}|s^n)\\
&=\frac{1}{n!}\sum_{\pi\in S_n}w_{s^n}(\hat D_{\pi(x^n)}|\pi(x^n))-K\cdot L\cdot \mathrm{pl}(n)\cdot \sum_{\pi\in S_n}\frac{1}{n!}w_p^{\otimes
n}(\hat D_{x^n}|\pi(s^n)),
\end{align}
where $x^n\in T_p$ is arbitrary and $w_p(y|s)=\sum_{x\in\mathcal X}p(x)w(y|s,x)$ according to our definition in equation \eqref{eqn:def-of-W(p)}. By carrying out the same estimate as in equation (\ref{eqn:upper-bound-on-type-distribution}) for the
distribution $\frac{1}{|T_q|}\mathbbm1_{T_q}$ induced by the type $q$ of $s^n$ and setting $\mathrm{pl_2}(n):=(n+1)^{2\cdot\max\{|\mathcal
X|,|\mathcal S|\}}$ we get (note here that $w_{p\otimes q}(y):=\sum_{s,x}q(s)\cdot p(x)\cdot w(y|s,x)$ defines, according to our convention, a probability distribution on $\mathcal P(\mathcal Y)$ which is identical to $W(p\otimes q)$)
\begin{align}
\mathbb E d_{s^n}(\mathbf X_\gamma)&\geq\frac{1}{n!}\sum_{\pi\in S_n}w_{s^n}(\hat D_{\pi(x^n)}|\pi(x^n))-K\cdot L\cdot \mathrm{pl}_2(n)\cdot
w_{p\otimes q}^{\otimes n}(\hat D_{x^n})\\
&=\frac{1}{n!}\sum_{\pi\in S_n}w_{\pi(s^n)}(\hat D_{x^n}|x^n)-K\cdot L\cdot \mathrm{pl}_2(n)\cdot w_{p\otimes q}^{\otimes n}(\hat D_{x^n})\\
&\geq\frac{1}{n!}\sum_{\pi\in S_n}w_{\pi(s^n)}(\hat D_{x^n}|x^n)-K\cdot L\cdot \mathrm{pl}_2(n)\cdot\max_{\xi\in\Xi_n} w_{p\otimes q}^{\otimes
n}(T_{W_\xi,\delta}(x^n)).
\end{align}
It is now the time to apply Ahlswede's robustification technique. For the fixed but arbitrary $x^n\in T_p$ define $f$ by fixing all its values $f(s^n)$ via $f(s^n):=w_{s^n}(\hat D_{x^n}|x^n)$. Then by Lemma \ref{lemma:ahlswede-robustification} we get
\begin{align}
\mathbb E d_{s^n}(\mathbf X_\gamma)&\geq1-(n+1)^{|\mathcal S|}\max_{\xi\in\Xi_n}w_\xi^{\otimes n}(\hat D_{x^n}^\complement|x^n)-K\cdot L\cdot
\mathrm{pl}_2(n)\cdot\max_{\xi\in\Xi_n} w_{p\otimes q}^{\otimes n}(T_{W_\xi,\delta}(x^n))\\
&\geq1-\mathrm{pl}_2(n)\left(\max_{\xi\in\Xi_n} W_\xi^{\otimes n}(T_{W_\xi,\delta}(x^n)^\complement|x^n)+\right.\\
&\left.\qquad\qquad\qquad\qquad\qquad\qquad +K\cdot L\cdot\max_{\xi\in\Xi_n} w_{p\otimes q}^{\otimes n}(T_{W_\xi,\delta}(x^n))\right)\\
&\geq1-\mathrm{pl}_2(n)\left(2^{-n\cdot\delta/2}+K\cdot L\cdot\max_{\xi\in\Xi_n} \prod_{x\in\mathcal X}w_{p\otimes q}^{\otimes n\cdot
p(x)}(T_{W_\xi(x),\delta})\right).
\end{align}
The last term in above estimate deserves special attention. Following the lines of proof of Lemma 3 in \cite{bbs-secrecy} (which was originally proven in
\cite{wyrembelski-bjelakovic-oechtering-boche}) we see that
\begin{align}
D(\bar N(\cdot|y^n)\|W_\xi(p))&=D(\sum_{x}p(x)\bar N_x(\cdot|y^n)\|W_\xi(p))\\
&\leq\sum_{x}p(x)D(N_x(\cdot|x^n,y^n)\|W_\xi(p))\\
&= D(\bar N(\cdot|x^n,y^n)\|W_\xi(p)\otimes p)\\
&\leq \delta.
\end{align}
It follows that for each $\xi\in\Xi_n$ we have by Lemma \ref{lemma:continuity-of-conditional-entropy-with-respect-to-averaged-norm} that
\begin{align}
W_{p\otimes q}^{\otimes n}(T_{W_\xi,\delta}(x^n))&\leq|T_{W_\xi,\delta}(x^n)|\max_{y^n\in T_{W_\xi,\delta}(x^n)}w_{p\otimes q}^{\otimes
n}(y^n)\\
&\leq|T_{W_\xi,\delta}(x^n)|\max_{y^n\in T_{W_\xi,\delta}(x^n)}2^{-n(D(\bar N(\cdot|y^n)\|W(p\otimes q))+H(\bar N(\cdot|y^n)))}\\
&\leq|T_{W_\xi,\delta}(x^n)|2^{-n(H(W_\xi(p))-f_1(\sqrt{2\cdot\delta}))}.
\end{align}
We further estimate that for the distribution $p_{XY,\xi}\in\mathcal P(\mathcal X\times\mathcal Y)$ defined via $p_{XY}(x,y):=p(x)w_\xi(y|x)$ we have
\begin{align}
|T_{W_\xi,\delta}(x^n)|&\leq\max_{y^n:D(\bar N(\cdot|x^n,y^n)\|p_{XY,\xi})\leq\delta}|\{\hat y^n:N(\cdot|\hat y^n,x^n)=N(\cdot|y^n,x^n)\}|\\
&\leq\max_{y^n:D(\bar N(\cdot|x^n,y^n)\|p_{XY,\xi})\leq\delta}2^{n\cdot H(\hat Y|\hat X)}\\
&\leq2^{n\cdot\sum_xp(x)H(W_\xi(\delta_x))+f_1(\sqrt{2\cdot\delta})},
\end{align}
by Lemma \ref{lemma:cardinality-bound} and Lemma \ref{lemma:continuity-of-conditional-entropy-with-respect-to-averaged-norm}.
We can now re-insert this estimate into our original problem and obtain
\begin{align}
\mathbb Ed_{s^n}(\mathbf X_\gamma)&\geq1-\mathrm{pl}_2(n)\left(2^{-n\cdot\delta}+K\cdot L\cdot2^{-n(\min_\xi
I(p;W_\xi)-2f_1(\sqrt{2\cdot\delta}))}\right)\\
&\geq1-\mathrm{pl}_2(n)\left(2^{-n\cdot\delta/2}+K\cdot L\cdot 2^{-n(\min_q I(p;W_q)-2\cdot f_1(\sqrt{2\cdot\delta}))}\right)\\
&\geq1-\mathrm{pl}_2(n)\left(2^{-n\cdot\delta/2}+ 2^{-n\cdot\delta/2}\right)\\
&\geq1-2^{-n\cdot\delta/4}
\end{align}
for all large enough $n\in\mathbb N$, since
\begin{align}
K\cdot L\leq  2^{n(\min_q I(p;W_q)-\delta-2\cdot f_1(\sqrt{2\cdot\delta}))}\leq 2^{n(\min_\xi I(p;W_\xi)-\delta-2\cdot f_1(\sqrt{2\cdot\delta}))}
\end{align}
by assumption and since $\Xi_n\subset\mathcal P(\mathcal S)$. Observe that this lower bound is entirely independent from the choice of
$s^n\in\mathcal S^n$. It now follows from the Chernoff bound Lemma \ref{lem:Chernoff} that
\begin{align}
\mathbb P(\forall s^n:\ \frac{1}{\Gamma}\sum_{\gamma=1}^\Gamma d_{s^n}(\mathcal K_\gamma)\leq(1-\epsilon)\mathbb E d_{s^n}(\mathcal
K))&\leq|\mathcal S|^n\cdot\exp(-\Gamma\cdot\epsilon^2\cdot\mathbb Ed_{s^n}(\mathcal K)/3)\\
&\leq \exp(n\cdot\log(|S|)-\Gamma\cdot\epsilon^2\cdot(1-2^{-n\cdot\delta/4})/3).\label{eqn:prob-existence-of-random-code}
\end{align}
Choose $\epsilon=2^{-n\cdot\delta/4}$ to obtain the statement.
\\\\
{\bf Proof of statement 3 in Lemma \ref{lemma:intermediate-secrecy-result-I}:} The proof of this statement follows from the proof of statement 1 where the functional dependence
$\tau\mapsto\delta(\tau)$ is specified.
\end{proof}
\end{subsection}
\begin{subsection}{Proof of Lemma \ref{lemma:central-lemma-II}\label{subsec:proof-of-central-lemma-II}}
\begin{proof}[Proof of Lemma \ref{lemma:central-lemma-II}]
We know from Lemma \ref{lemma:intermediate-secrecy-result-I} that (if $\frac{1}{n}\log(K\cdot L)\leq B(p)-\delta-2\cdot f_1(\sqrt{2\cdot\delta})$ for some $\delta>0$ and $n$ is large enough)
\begin{align}
\mathbb P(E_2)\geq1-\exp(n\cdot\log(|\mathcal S|)-\Gamma\cdot\epsilon^2\cdot(\frac{1}{3}-2^{-n\cdot\delta/2})).
\end{align}
Stepping away from the goal of proving Lemma \ref{lemma:central-lemma-II} we see that there are two possible routes which diverge from here.
One is to make $\Gamma$ as small as possible, the other will be to exploit large numbers $\Gamma$. We will soon go on with the second approach
and thereby prove Lemma \ref{lemma:central-lemma-II}, but first let us assume that we want $\Gamma$ to be as small as possible (in an
asymptotic sense of course). How can we achieve this? We take any sequence $(\epsilon_n)_{n\in\mathbb N}$ of numbers $\epsilon_n\in[0,1]$ which
converges to zero. Depending on such a choice, we set $\Gamma_n=3\cdot\log(|\mathcal S|^2)\frac{n}{\epsilon_n^2}(1-2^{-n\delta})$. It follows for the average success probability $d_{s^n}(\mathcal K_\gamma)$ as defined in Definition \ref{def:CR-assisted-code} that
\begin{align}
\mathbb P(\ \forall s^n:\ \frac{1}{\Gamma}\sum_{\gamma=1}^\Gamma d_{s^n}(\mathcal K_\gamma)\leq(1-\epsilon)\mathbb E d_{s^n}(\mathcal K)\ )<1,
\end{align}
proving the existence of a sequence of codes for which
\begin{align}
\min_{s^n\in\mathcal S^n}\frac{1}{\Gamma_n\cdot K\cdot L}\sum_{\gamma,k,l=1}^{\Gamma_n,K,L}w_{s^n}(D_{k\gamma}|x_{k\gamma})\geq1-\epsilon_n
\end{align}
(whenever $\Gamma_n$ scales asymptotically as $\Gamma_n\approx\frac{n}{\epsilon_n^2}$). If $\epsilon_n=n^{-\nu}$ for some small number $\nu>0$
for example we get $\Gamma_n\approx\frac{n}{n^{-2\nu}}=n^{1+2\nu}$. This type of asymptotic scaling of common randomness has been observed
several times now in the literature, and obviously raises the question whether $\Gamma_n=\mathrm{const}\cdot n$ would be sufficient to
guarantee asymptotically optimal performance, for some sufficiently large number $\mathrm{const}$ depending only on $|\mathcal S|$, for
example.\\
We can now proceed our proof of Lemma \ref{lemma:central-lemma-II} by using equation (\ref{eqn:prob-existence-of-random-code}) together with
Lemma \ref{lemma:intermediate-secrecy-result-I} and a union bound: Let $\beta>0$ and $\tau>0$. From now on until the end of this proof, let $\delta=\delta(\tau)$. Let
\begin{align}
\frac{1}{n}\log(\Gamma_n\cdot L_n)\geq E(p)+\tau,\qquad B(p)-\delta-2\cdot f_1(\sqrt{2\cdot\delta})\geq\frac{1}{n}\log(K_n\cdot L_n).
\end{align}
It then follows that for all large enough $n$ it holds that
\begin{align}
\mathbb P(E_1\cap E_2)>0.
\end{align}
Thus, there is a realization $\mathbf x$ of $\mathbf X$ such that for this particular realization we have
\begin{align}
\forall s^n,z^n,k:\ \frac{1}{L\cdot\Gamma}\sum_{l,\gamma=1}^{L,\Gamma}\Theta_{s^n,z^n}(\mathbf x_{kl\gamma})\in[(1\pm2^{-n\tau/4})\mathbb
E\Theta_{s^n,z^n}]\label{eqn:approximate-security-is-given!}\\
\min_{s^n\in\mathcal S^n}\frac{1}{\Gamma}\sum_{\gamma=1}^\Gamma d_{s^n}(\mathcal K_\gamma)\geq1-2\cdot 2^{-n\delta/2}\label{eqn:code-is-good!}
\end{align}
Further, for every $k\in[K_n]$ we have (setting $\Delta(s^n,z^n,x^n):=\Theta_{s^n,z^n}(x^n)$ for all $s^n,z^n$ and $x^n$)
\begin{align}
&\left\|\frac{1}{L\cdot\Gamma}\sum_{l,\gamma=1}^{L,\Gamma}v_{s^n}(\cdot|\mathbf x_{kl\gamma})-\mathbb Ev_{s^n}\right\|_1\\
&\leq\left\|\frac{1}{L\cdot\Gamma}\sum_{l,\gamma=1}^{L,\Gamma}(v_{s^n}(\cdot|\mathbf x_{kl\gamma})-\Delta(s^n,\cdot,\mathbf
x_{kl\gamma})\right\|_1
+\left\|\frac{1}{L\cdot\Gamma}\sum_{l,\gamma=1}^{L,\Gamma}\Delta(s^n,\cdot,\mathbf x_{kl\gamma})-\mathbb
E\Delta(s^n,\cdot,X^n)\right\|_1\\
&\qquad+\left\|\mathbb E(v_{s^n}(\cdot|X^n)-\mathbb E\Delta(s^n,\cdot,X^n)\right\|_1\\
&\leq\frac{1}{L\cdot\Gamma}\sum_{l,\gamma=1}^{L,\Gamma}\left\|v_{s^n}(\cdot|\mathbf x_{kl\gamma})-\Delta(s^n,\cdot,\mathbf
x_{kl\gamma})\right\|_1+2^{-n\cdot\tau/4}+\mathbb E\left\|v_{s^n}(\cdot|X^n)-\Delta(s^n,\cdot,X^n)\right\|_1
\end{align}
where the first inequality is due to the triangle inequality of $\|\cdot\|_1$ and the second one due to the specific probabilistic choice of
$\mathbf x$, especially the validity of (\ref{eqn:approximate-security-is-given!}). We now use the definition of $\Theta_{s^n,z^n}$ in order to
derive bounds on the remaining quantities: for every $x^n\in T_p$ we have
\begin{align}
\|v_{s^n}(\cdot|x^n)-\Delta(s^n,\cdot,x^n)\|_1&=\sum_{z^n:D(\bar N(\cdot|s^n,x^n,z^n\|p_{SXZ})>\delta}v^{\otimes n}(z^n|s^n,x^n)\\
&=v^{\otimes n}(T_{V,\delta}(s^n,x^n)^\complement|s^n,x^n)\\
&\leq2^{-n\cdot\delta/2},
\end{align}
for all large enough $n$. Thus
\begin{align}
\frac{1}{L\cdot\Gamma}\sum_{l,\gamma=1}^{L,\Gamma}\|v_{s^n}(\cdot|\mathbf x_{kl\gamma})-\Delta(s^n,\cdot,\mathbf
x_{kl\gamma})\|_1+\mathbb E\|v_{s^n}(\cdot|X^n)-\Delta(s^n,\cdot,X^n)\|_1\leq2\cdot2^{-n\delta/2}
\end{align}
for all large enough $n\in\mathbb N$ so that we ultimately get (uniformly in $k\in[K]$) the bound
\begin{align}
\frac{1}{L\cdot\Gamma}\sum_{l,\gamma=1}^{L,\Gamma}\|v_{s^n}(\cdot|\mathbf x_{kl\gamma})-\Delta(s^n,\cdot,\mathbf
x_{kl\gamma})\|_1&\leq2\cdot2^{-n\delta/2}+2^{-n\tau/4}\leq2^{-n\cdot\nu(\tau)},
\end{align}
for all large enough $n$ and setting $\nu(\tau):=\min\{\delta(\tau),\tau\}/5$ (note that $\nu(\tau)=\tau/5$ is a valid choice).
\end{proof}
\end{subsection}
\begin{subsection}{Proof of Theorems \ref{theorem:symmetrizability-properties-of-C1det}, \ref{theorem:stability-properties-of-C1det} and
\ref{theorem:discontinuity-properties-of-C1det} \label{subsec:proof-of-stability} (properties of $C_{\mathrm{S}}$)}
\begin{proof}[Proof of Theorem \ref{theorem:symmetrizability-properties-of-C1det}]
We give the proof of the properties of $C_{\mathrm{S}}$ in the same order as they were stated in the theorem:
\\\\
{\bf1.} This is clear from \cite{ericson} where it was proven that symmetrizability makes it impossible to reach reliable transmission of
messages.
\\\\
{\bf2.} The strategy of proof is to use Lemma \ref{lemma:central-lemma} with $\Gamma=1$. The reason for this is that, by assumption, $\mathfrak W$ is non-symmetrizable. Now, we know from Example \ref{example:invertible-pre-coding-induces-symmetrizability} that this does not imply that every $\mathfrak W\circ U$ is non-symmetrizable as well. More precisely, to a given $r\in\mathbb N$ there may exist an alphabet $\mathcal U_n$, a $p\in\mathcal P(\mathcal U_n)$ and a channel $U_n\in C(\mathcal U,\mathcal X^n)$ such that
\begin{align}
\min_{q\in\mathcal P(\mathcal S^r)}&I(p;W_q\circ U_r)-\max_{s^r\in\mathcal S^r}I(p; V_{s^r}\circ U_r)\\
&=\max_{p'\in\mathcal P(\mathcal U_r)}\max_{U_r'\in C(\mathcal U_r,\mathcal X^r}\min_{q\in\mathcal P(\mathcal S^r)}I(p';W_q\circ U_r')-\max_{s^r\in\mathcal S^r}I(p'; V_{s^r}\circ U_r')\\
&\geq C^\mathrm{mean}_\mathrm{S,ran}(\mathfrak W,\mathfrak V)-\epsilon
\end{align}
but, additionally, $(W_{s^r}\circ\mathcal U_r)_{s^r\in\mathcal S^r}$ is symmetrizable. We provide here two approaches to deal with this problem: First, we will use the fact that $\mathfrak W$ is non-symmetrizable for transmission of a small number of messages that can be read by Eve but, since backwards communication from Eve to James is forbidden, are sufficient to counter any of the allowed jamming strategies.\\
Second, we will consider a variant of the optimization problem \eqref{eqn:capacity-formula} where optimization of $U'_r$ is restricted to maps of the form $U'_r=Id\otimes U_{2,\ldots,r}''$ and we will prove that these restricted maps are asymptotically as good as those that are derived from the original problem when it comes to calculating capacity. However, these maps have the additional property that they cannot turn a non-symmetrizable AVC into a symmetrizable one.\\
Now let $r\in\mathbb N$ be arbitrary but fixed and $p,U_r$ as above. Let $k,l\in\mathbb N$ be such that $n=k+l$ and $l=\lfloor\lambda\cdot n\rfloor$, where $\lambda\in(0,1)$ is arbitrary but fixed for the moment. Then from \cite[Lemma 5]{csiszar-narayan}, if $\hat K$ satisfies the assumptions of Lemma \ref{lemma:central-lemma} with $L$ set to one based on the properties (\ref{eqn:property-1-of-central-lemma}), (\ref{eqn:property-2-of-central-lemma}) and (\ref{eqn:property-4-of-central-lemma}) of the lemma.\\
So, on the grounds of \ref{lemma:central-lemma} and of the results proven in \cite{csiszar-narayan}, we see that for every $m'\in\mathbb N$, $r\in\mathbb N$ and $\delta>0$, $p\in\mathcal P_0^{m'}(\mathcal U_r)$ (where $\mathcal U_r=[|\mathcal X|^{r}]$ ) and $U\in C(\mathcal U_r,\mathcal X^r)$ there exists a code $\mathcal K=(\mathcal K_m)_{m=1}^\infty$ such that for every $s^{r\cdot m}\in(\mathcal S^r)^m=\mathcal S^{r\cdot m}$ we have
\begin{align}
\frac{1}{K_k'}\sum_{a=1}^{K_k'}\sum_{x^{k}}w_{s^{k}}(D_{a}'|\mathbf x_a')\geq1-\epsilon_k,
\end{align}
where $\{\epsilon_k\}_{k\in\mathbb N}\subset[0,1]$, $\lim_{k\to\infty}\epsilon_k=0$ and it may be assumed that $K_k'=l^3$. In addition to that we know from \cite{wiese-noetzel-boche-I} that there exist codes for $(\mathfrak W,\mathfrak V)$ such that
\begin{align}
\min_{s^l\in\mathcal S l}\frac{1}{\Gamma_l}\frac{1}{K_l''}\sum_{a,b=1}^{\Gamma_l,K_l''}\sum_{x^l\in\mathcal X^l}u_l(x^l|a,b)w_{s^{l}}(D_{a,b}''|\mathbf x_{ab}'')\geq1-\delta_l,
\end{align}
where $\{\delta_l\}_{l\in\mathbb N}\subset[0,1]$, $\lim_{l\to\infty}\epsilon_l=0$, $\Gamma_l=l^3$, $U_l\in C([\Gamma_l]\times[K_l''],\mathcal X^l)$ is stochastic pre-coding and $D_{a,b}\cap D_{a,b'}=\emptyset$ whenever $b\neq b'$ ($a\in[\Gamma_l]$ is used as common randomness in \cite{wiese-noetzel-boche-I}, whereas here we will substitute the messages that were sent on the first $k$ channel uses for it. Note that the messages on the first $k$ channel uses are not secure against Eve). In addition to that it holds
\begin{align}
\lim_{l\to\infty}\frac{1}{l}\log K_l''=C^\mathrm{mean}_\mathrm{S,ran}(\mathfrak W,\mathfrak V)-\nu
\end{align}
for some arbitrarily small $\nu>0$ and
\begin{align}\label{eqn:secrecy-from-wnb-I}
\lim_{l\to\infty}\frac{1}{l}\max_{\gamma\in[\Gamma_l]}\max_{s^l\in\mathcal S^l}I(\mathfrak K''_l;\mathfrak Z_{s^l}|\mathfrak \Gamma_l=a)=0.
\end{align}
The mutual information is evaluated on the random variables defined via
\begin{align}
\mathbb P_{s^l}((\mathfrak K''_l,\mathfrak Z_{s^l},\mathfrak \Gamma_l)=(b,z^l,a))&:=\frac{1}{\Gamma_l}\frac{1}{K_l''}\sum_{x^l\in\mathcal X^l}u_l(x^l|a,b)v(z^l|s^l,x^l).
\end{align}
We concatenate the two codes by defining new stochastic encodings $E_n\in C([K''_l],\mathcal X^n)$ via
\begin{align}
e_n((x^k,x^l)|b):=\sum_{a=1}^{\Gamma_l}\delta_{\mathbf x_a}(x^k)u_l(x^l|a,b)
\end{align}
and new decoding sets via
\begin{align}
D_{b}:=\cup_aD_a'\times D_{a,b}''\subset\mathcal X^n.
\end{align}
It holds $D_b\cap D_{b'}=\cup_{a,a'}(D_{a}\times D_{a,b}\cap D_{a'}\times D_{a',b'})=\emptyset$. We set $K_n:=K''_l$, $\alpha_n:=\epsilon_k$ and $\beta_n:=\delta_l$ for the $l$ satisfying $l=\lfloor\lambda\cdot n\rfloor$ and the $k$ satisfying $k=n-l$. Then $\lim_{n\to\infty}\alpha_n=\lim_{n\to\infty}\beta_n=0$. As a consequence of the Innerproduct Lemma in \cite{ahlswede-elimination} we know that for every $s^n=(s^k,s^l)$ we have
\begin{align}
\frac{1}{K_n}\sum_{b=1}^{K_n}\sum_{x^n\in\mathcal X^n}e_n(x^n|b)w(D_b|s^n,x^n)&\geq\frac{1}{K_l}\sum_{a,b=1}^{\Gamma_k,K_l'}\sum_{x^l\in\mathcal X^n}u(x^l|a,b)w(D_a'|s^k,x^k)w(D_{a,b}''|s^l,x^l)\\
&\geq1-2\max\{\alpha_n,\beta_n\}.
\end{align}
That the messages $b\in[K_n]$ are also asymptotically secure in the sense that
\begin{align}
\lim_{n\to\infty}\frac{1}{n}\max_{s^n\in\mathcal s^n}I(\mathfrak K_n;\mathfrak Z_{s^n})&\leq\lim_{l\to\infty}\frac{\lambda}{l}\max_{s^l\in\mathcal S^l}I(\mathfrak K_l'';\mathfrak Z_{s^l}|\mathfrak\Gamma_l)\\
&=0
\end{align}
follows from independence of the distributions of the messages $b$ and the values $a$ of the common randomness as described in the inequalities from \eqref{eqn:secrecy-assumption} to \eqref{eqn:consequence-of-secrecy-assumption}. Especially inequality \eqref{eqn:secrecy-assumption} is valid since as a consequence of \eqref{eqn:secrecy-from-wnb-I}. The rate of the code is calculated as
\begin{align}
\lim_{n\to\infty}\frac{1}{n}\log K_n=\lambda\left(C^\mathrm{mean}_\mathrm{S,ran}(\mathfrak W,\mathfrak V)-\nu\right).
\end{align}
Since $\nu$ can be arbitrarily close to $0$ and $\lambda$ can be chosen arbitrarily close to $1$ we have proven the desired result.

We now explain the second approach to proving statement $2.$ in Theorem \ref{theorem:symmetrizability-properties-of-C1det}. Here we aim to utilize the full power of Lemma \ref{lemma:central-lemma} with $\Gamma=1$. Our starting point are the distributions $p$ and the channels $U$ arising from the optimization \eqref{eqn:capacity-formula} for fixed $r\in\mathbb N$. Note that, without loss of generality, $\mathcal U_r=\mathcal X^r$ for every $r\in\mathbb N$ in \eqref{eqn:capacity-formula}. Set, for every $r\in\mathbb N$,
\begin{align}\label{eqn:Cr}
C_r:=\max_{p\in\mathcal P(\mathcal X^r)}\max_{U_r\in C(\mathcal X^r,\mathcal X^r)}\min_{q\in\mathcal P(\mathcal S^r)}I(p;W_q\circ U_r)-\max_{s^r\in\mathcal S^r}I(p; V_{s^r}\circ U_r).
\end{align}
Let $r\in\mathbb N$ be arbitrary but fixed. For an arbitrary $\epsilon\geq0$, let $p$ and $U_r$ be such that
\begin{align}
C_r-\epsilon&=\min_{q\in\mathcal P(\mathcal S^r)}I(p;W_q\circ U_r)-\max_{s^r\in\mathcal S^r}I(p; V_{s^r}\circ U_r).
\end{align}
Now define $\tilde U_{r+1}$ by $\tilde u_{r+1}((x_1,\ldots,x_{r+1})|(x,u)):=\sum_{x'\in\mathcal X}u_r((x',x_2,\ldots,x_{r+1})|u)\delta_x(x_1)$ for all $x,x_1,\ldots,x_{r+1}\in\mathcal X$ and $u\in\mathcal U_r=\mathcal X^r$. Then it holds that
\begin{align}
C_{r+1}&\geq\min_{q\in\mathcal P(\mathcal S^{r+1})}I(p\otimes\pi;W_q\circ U_{r+1})-\max_{s^{r+1}\in\mathcal S^{r+1}}I(p\otimes\pi; V_{s^{r+1}}\circ U_{r+1})\\
&\geq\min_{q\in\mathcal P(\mathcal S^r)}I(p;W_q\circ U_r)-\max_{s^r\in\mathcal S^r}I(p; V_{s^r}\circ U_r)-\log|\mathcal X|\\
&=C_r-\epsilon-\log|\mathcal X|,
\end{align}
where $\pi\in\mathcal P(\mathcal X)$ is defined by $\pi(x):=|\mathcal X|^{-1}$ for all $x\in\mathcal X$. This latter estimate is due to the equality $I(p\otimes\pi; V_{s^{r+1}}\circ U_{r+1})=I(p;V_{s^r}\circ U)+I(\pi;V_{s_{r+1}})$, the data processing inequality and the fact that for arbitrary channels $S\in C(\mathcal A\times\mathcal B,\mathcal C)$ and $T\in C(\mathcal A'\times\mathcal B',\mathcal C')$, as well as distributions $q\in\mathcal S(\mathcal B\times\mathcal B')$ with respective marginal distributions $q_B\in\mathcal P(\mathcal B)$ and $q_{B'}\in\mathcal P(\mathcal B')$ and $p\in\mathcal S(\mathcal A\times\mathcal A')$ with respective marginal distributions $p_A\in\mathcal P(\mathcal A)$ and $q_{A'}\in\mathcal P(\mathcal A')$ we have
\begin{align}
\forall\ (a,b,c)\in\mathcal A\times\mathcal B\times\mathcal C:\qquad \sum_{a',c'}\sum_{b,b'}s(c|a,b)t(c'|a',b')p(a,a')q(b,b')&=\sum_{b}q_{B}p_{A}t(c|a,b).
\end{align}
Since $\mathfrak W$ is non-symmetrizable we know that $\mathfrak W^{\otimes r}\circ \tilde U_{r}$ is non-symmetrizable for every $r\geq2$. The reason for that is explained as follows: Let again $S,T$ be channels as above. Assume that $S$ is symmetrizable but $T$ is not. Then $S\otimes T$ is non-symmetrizable. This can be seen by assuming the existence of a symmetrising map $Q\in C(\mathcal A\times\mathcal A',\mathcal B\times\mathcal B')$. The statement
\begin{align}
\forall (a_1,a_2,a'_1,a'_2)&\in\mathcal A^2\times{\mathcal A'}^2:\nonumber\\
&\sum_{b,b'}s(\cdot|a_1,b)t(\cdot|a_1',b')q(b,b'|a_2,a_2')=\sum_{b,b'}s(\cdot|a_2,b)t(\cdot|a_2',b')q(b,b'|a_1,a_1')
\end{align}
would obviously imply for any fixed choice of $(a_1,a_2)$ the statement
\begin{align}
\forall (a_1',a_2')\in\mathcal A\times\mathcal A':\qquad\sum_{b'}t(\cdot|a_1',b')q_{B'}(b'|a_2,a_2')=\sum_{b'}t(\cdot|a_2',b')q_{B'}(b'|a_1,a_1'),
\end{align}
where $q_{B'}(b'|a_1,a_1'):=\sum_{b}q(b,b'|a_1,a_1')$. This would be in contradiction to non-symmetrizability of $T$. Since $\tilde U_r=U_{r-1}\otimes Id$ we can thus conclude that $\mathfrak W^{\otimes r}\circ \tilde U_{r}$ is non-symmetrizable. We now proceed with the proof of Theorem \ref{theorem:symmetrizability-properties-of-C1det}.\\
With this approach we have evaded the problem that $\mathfrak W^{\otimes r}\circ U_r$ may well be symmetrizable (see our Example \ref{example:invertible-pre-coding-induces-symmetrizability}).\\
By \cite[Lemma 4]{csiszar-narayan} non-symmetrizability of $\mathfrak W^{\otimes r}\circ \tilde U_{r}$ implies that it is possible to define a decoder according to \cite[Definition 3]{csiszar-narayan}, with $N=K\cdot L$ and $[N]$ replaced by $[K]\times[L]$. Since only the number of codewords and their type ever enters the proof it makes no difference whether we enumerate them by one index taken from $[N]$ or by two indices taken from $[K]\times[L]$. This decoder is proven to work reliably in \cite[Lemma 5]{csiszar-narayan} (even with an exponentially fast decrease of average error), if $N=K\cdot L$ satisfies the assumptions of Lemma \ref{lemma:central-lemma} based on the properties (\ref{eqn:property-1-of-central-lemma}), (\ref{eqn:property-2-of-central-lemma}) and (\ref{eqn:property-4-of-central-lemma}) of the lemma.\\
So, on the grounds of Lemma \ref{lemma:central-lemma} and of the results proven in \cite{csiszar-narayan}, we see that for every $m\in\mathbb N$, $r\in\mathbb N\backslash\{1\}$ and $\delta>0$, $p\in\mathcal P_0^{m}(\mathcal X^r)$ and $U\in C(\mathcal X^{r-1},\mathcal X^{r-1})$ there exists a code $\mathcal K=(\mathcal K_m)_{m=1}^\infty$ such that for every $s^{r\cdot m}\in(\mathcal S^{r})^m=\mathcal S^{r\cdot m}$ we have
\begin{align}
\frac{1}{K_m}\frac{1}{L_m}\sum_{k,l=1}^{K_m,L_m}\sum_{x^{m\cdot r}}w_{s^{m\cdot r}}(D_{kl}|x^{m\cdot r})u^{\otimes m}(x^{m\cdot
r}|u_{kl})\geq1-\epsilon_m,
\end{align}
where $\{\epsilon_m\}_{m\in\mathbb N}\subset[0,1]$, $\lim_{m\to\infty}\epsilon_m=0$ and it holds that
\begin{align}
\liminf_{m\to\infty}\frac{1}{m}\log(K_m\cdot L_m)\geq \min_{q\in\mathcal P(\mathcal S^r)}I(p;W_q^m\circ \tilde U_r)-\delta
\end{align}
(the code we use here is defined by using the codewords $\mathbf x_{kl\gamma}$ together with the decoder from \cite[Definition 3]{csiszar-narayan} defined for the AVC $\mathfrak W^{\otimes r}\circ \tilde U_r:=(W_{s^r}\circ (U_{r-1}\otimes Id))_{s^r\in\mathcal S^r}$) and
\begin{align}
\max_{s^r\in\mathcal S^r}I(p;V_{s^r}\circ \tilde U_r)+2\delta\geq\liminf_{m\to\infty}\frac{1}{m}\log L_m\geq \max_{s^r\in\mathcal S^r}I(p;V_{s^r}\circ
\tilde U_r)+\delta,
\end{align}
implying that for a sequence $(p_m)_{m\in\mathbb N}$ of choices for $p_m$ converging to some $p$ having a decomposition $p=p'\otimes\pi$ for $p'\in\mathcal P(\mathcal X^{r-1})$ being an optimal choice in the sense of \eqref{eqn:Cr} we get
\begin{align}
\liminf_{m\to\infty}\frac{1}{m}\log K_m&\geq\min_{q\in\mathcal P(\mathcal S^r)}I(p;W_q\circ \tilde U_r)-\max_{s^r\in\mathcal S^r}I(p; V_{s^r}\circ
\tilde U_r)-3\delta\\
&\geq C_{r-1}-\log|\mathcal X|-3\delta.
\end{align}
Also, it is clear from the last part of Lemma \ref{lemma:central-lemma} (equation (\ref{eqn:property-5-of-central-lemma})) together with
\cite[Lemma 20]{wiese-noetzel-boche-I} that the codes employed here are asymptotically secure in the strong sense:
\begin{align}
\limsup_{m\to\infty}\max_{s^{r\cdot m}}I(\mathfrak K_m;\mathfrak Z_{s^{r\cdot m}})=0.
\end{align}
We now wish to apply the code for the extended channel $(\mathfrak W^{\otimes r}\circ \tilde U_{r},\mathfrak V^{\otimes r})$ to the original channel $(\mathfrak W,\mathfrak V)$. Define values $t_n\in\{0,\ldots,r-1\}$ by requiring $n=m\cdot r+t_n$ for them to hold for some suitably chosen $m=m(n)\in\mathbb N$. This
quantity satisfies $-1+n/r\leq m(n)\leq n/r$. For every $n\in\mathbb N$ we then define new decoding sets by
\begin{align}
\hat D_{kl}:=D_{kl}\times\mathcal Y^{t_n}
\end{align}
and new randomized encodings by setting for some arbitrary but fixed $x^{t_n}$
\begin{align}
E(\hat x^n|k):=\sum_{l=1}^L\frac{1}{L}u^{\otimes n}(x^{m\cdot r}|u_{kl})\cdot\delta_{x^{t_n}}(\hat x^{t_n}).
\end{align}
From the choice of codewords and the decoding rule it is clear that this code is asymptotically reliable. The asymptotic number of codewords
(mind that $\hat K_n=K_{m(n)}$) calculated and normalized with respect to $n$, is
\begin{align}
\liminf_{n\to\infty}\frac{1}{n}\log\hat K_n&=\liminf_{n\to\infty}\frac{1}{m(n)\cdot r+t_n}K_{m(n)}\\
&\geq\liminf_{n\to\infty}\frac{1}{r}\cdot\frac{1}{m(n)+1}K_{m(n)}\\
&=\liminf_{n\to\infty}\frac{1}{r}\cdot\frac{1}{m(n)}\cdot\frac{m(n)}{m(n)+1}\cdot K_{m(n)}\\
&=\frac{1}{r}\liminf_{n\to\infty}\frac{1}{m(n)}\cdot K_{m(n)}\\
&=\frac{1}{r}\left(C_{r-1}-3\delta\right)\\
&=\frac{1}{r-1}\cdot\frac{r-1}{r}\left(C_{r-1}-\log|\mathcal X|-3\delta\right).
\end{align}
In addition to that, the code is secure: For each $n\in\mathbb N$, the distribution of the input codewords and Eve's
outputs is
\begin{align}
\mathbb P(\mathfrak K_n=k,&\mathfrak Z_{s^n}=z^n)\nonumber\\
&=\sum_{l=1}^L\frac{1}{L}\sum_{x^{r\cdot m}}\sum_{x^{t_n}}u^{\otimes m}(x^{r\cdot
m}|u_{kl})v^{\otimes r\cdot m}(z^{r\cdot m}|x^{r\cdot m},s^{r\cdot m})v^{\otimes t_n}(z^{t_n}|x^{t_n},s^{t_n})\\
&=\mathbb P(\mathfrak K_n=k,\mathfrak Z_{s^{r\cdot m}}=z^{r\cdot m})\cdot v^{\otimes t_n}(z^{t_n}|x^{t_n},s^{t_n}).
\end{align}
This demonstrates that (uniformly in $s^n\in\mathcal S^n$ and since $\mathfrak K_n=\mathfrak K_m$ holds) we have
\begin{align}
I(\mathfrak K_n;\mathfrak Z_{s^n})=I(\mathfrak K_n;\mathfrak Z_{s^{r\cdot m}})+0=I(\mathfrak K_m;\mathfrak Z_{s^{r\cdot m}}).
\end{align}
Since the right hand side of above equation goes to zero for $n$ going to infinity and since $\lim_{r\to\infty}\frac{r-1}{r}=1$ we see that the capacity $C_\mathrm{S}$ is lower bounded by $\lim_{r\to\infty}\frac{1}{r}C_r$. It is not an immediate consequence that this implies we can reach the capacity $C^\mathrm{mean}_\mathrm{S,ran}(\mathfrak W,\mathfrak V)=C^*(\mathfrak W,\mathfrak V)$. Fortunately it has been proven in \cite{wiese-noetzel-boche-I} that
\begin{align}
C^*(\mathfrak W,\mathfrak V)=\lim_{r\to\infty}\frac{1}{r}\max_{p\in\mathcal P(\mathcal U_n)}\max_{U_n\in C(\mathcal U,\mathcal X^n)}\left(\min_{q\in\mathcal P(\mathcal S^r)}I(p;W_q\circ U)-\max_{s^r\in\mathcal
S^r}I(p;V_{s^r}\circ U)\right)
\end{align}
holds. Thus $\lim_{r\to\infty}\frac{1}{r}C_r=C^*(\mathfrak W,\mathfrak V)$. This finally implies the desired result.
\end{proof}
\begin{proof}[Proof of Theorem \ref{theorem:stability-properties-of-C1det}] If $C_{\mathrm{S}}(\mathfrak W,\mathfrak V)=0$, there is
nothing to prove. Assume that $C_{\mathrm{S}}(\mathfrak W,\mathfrak V)>0$. It is evident that, in that case, $\mathfrak W$ is not
symmetrizable. The function $F$ defined in Definition \ref{def:definition-of-F} is continuous with respect to the Hausdorff distance (proving this statement is in complete analogy as the
corresponding part in the proof of Theorem 5 in \cite{bn-positivity}). Thus, if $F(\mathfrak W)>0$, then there is an $\epsilon>0$ such that for
all $\mathfrak W'$ satisfying $d(\mathfrak W,\mathfrak W')<\epsilon$ we know that $F(\mathfrak W')>0$ as well. Thus, every of these $\mathfrak W'$ is non-symmetrizable.\\
For some suitably chosen $\epsilon'<\epsilon$ we additionally know from Theorem 9 in \cite{wiese-noetzel-boche-I} that
$C^{\mathrm{mean}}_{\mathrm{S,ran}}(\mathfrak W',\mathfrak V)>0$ for all those $\mathfrak W'$ for which $d(\mathfrak W,\mathfrak W')<\epsilon'$. But since
Theorem \ref{thm:coding-theorem} shows that $F(\mathfrak W')>0\ \Rightarrow C_{\mathrm{S}}(\mathfrak W',\mathfrak
V)=C^{\mathrm{mean}}_{\mathrm{S,ran}}(\mathfrak W',\mathfrak V)$ this implies that
\begin{align}
C_{\mathrm{S}}(\mathfrak W,\mathfrak V)>0\ \forall\ \mathfrak W':\ d(\mathfrak W,\mathfrak W')<\epsilon'.
\end{align}
Since from Theorem \ref{theorem:stability-properties-of-C1det} we know that positivity of $C_{\mathrm{S}}(\mathfrak W',\mathfrak
V)$ ensures that it equals $C^{\mathrm{mean}}_{\mathrm{S,ran}}(\mathfrak W',\mathfrak V)$, and since the latter is continuous, we are done.
\end{proof}
\begin{proof}[Proof of Theorem \ref{theorem:discontinuity-properties-of-C1det}] Again, we prove everything in the same order as it is listed in
the theorem.
\\\\
{\bf1.} Let $C_{\mathrm{S}}$ be discontinuous in the point $(\mathfrak W,\mathfrak V)$. By Theorem
\ref{theorem:stability-properties-of-C1det} we know that this can only be the case if $C_{\mathrm{S}}(\mathfrak W,\mathfrak V)=0$. If in addition
we have $C^{\mathrm{mean}}_{\mathrm{S,ran}}(\mathfrak W,\mathfrak V)=0$ then we have, since $C^{\mathrm{mean}}_{\mathrm{S,ran}}$ is continuous, that for every
$\epsilon>0$ there is $\delta>0$ such that for all $(\mathfrak W_\delta,\mathfrak V)$ satisfying $d(\mathfrak W_\delta,\mathfrak W)<\delta$ we
have $C^{\mathrm{mean}}_{\mathrm{S,ran}}(\mathfrak W_\delta,\mathfrak V)\leq\epsilon$. Since since $C^{\mathrm{mean}}_{\mathrm{S,ran}}\geq C_{\mathrm{S}}$ this
would imply that $C_{\mathrm{S}}$ is continuous as well, in contradiction to the assumption. Thus $C^{\mathrm{mean}}_{\mathrm{S,ran}}(\mathfrak
W,\mathfrak V)>0$. Of course this immediately implies that $\mathfrak W$ has to be symmetrizable, by property 2. This is, in turn, equivalent
to $F(\mathfrak W)=0$. The definition of $F$ can be picked up from equation (\ref{eqn:definition-of-F}), its connection to symmetrizability is obvious from the definition. The notion of symmetrizability is explained in the introduction in equation (\ref{eqn:definition-of-symmetrizability}). Clearly, if for all $\epsilon>0$ and $\mathfrak W'$ satisfying $d(\mathfrak W,\mathfrak W')<\epsilon$ we would have
$F(\mathfrak W')=0$, then $C_{\mathrm{S}}(\mathfrak W',\mathfrak V')$ would be zero in a whole vicinity of $(\mathfrak W,\mathfrak V)$.
Thus for all $\epsilon>0$ there has to be at least one $\mathfrak W_\epsilon$ such that $d(\mathfrak W,\mathfrak W_\epsilon)<\epsilon$ but
$F(\mathfrak W_\epsilon)>0$.\\
The reverse direction is basically established by using all our arguments backwards: For all $\epsilon>0$, let there be at least one $\mathfrak W_\epsilon$ such that $d(\mathfrak W,\mathfrak W_\epsilon)<\epsilon$ but
$F(\mathfrak W_\epsilon)>0$. Let in addition to that $F(\mathfrak W)=0$ but $C_\mathrm{S,ran}^{\mathrm{mean}}(\mathfrak W,\mathfrak V)>0$. Since $C_\mathrm{S,ran}^{\mathrm{mean}}$ is continuous, there is a $\delta>0$ such that $C_\mathrm{S,ran}^{\mathrm{mean}}(\mathfrak W',\mathfrak V')>(1/2)\cdot C_\mathrm{S,ran}^{\mathrm{mean}}(\mathfrak W,\mathfrak V)=:\alpha$ whenever $d((\mathfrak W,\mathfrak V),(\mathfrak W',\mathfrak V'))<\delta$.\\
For every $\epsilon' \leq(1/2)\min\{\epsilon,\delta\}$ we can therefore deduce the following: It holds that $C_{\mathrm S}(\mathfrak W_{\epsilon'},\mathfrak V)=C_\mathrm{S,ran}^{\mathrm{mean}}(\mathfrak W_{eps'},\mathfrak V)\geq\alpha>0$ (since $F(\mathfrak W_{\epsilon'})>0$), but $C_\mathrm{S}(\mathfrak W_0,\mathfrak V)=0$. Thus $C_\mathrm{S}$ is discontinuous in the point $(\mathfrak W,\mathfrak V)$.
\\\\
{\bf2.} Let $C_{\mathrm{S}}$ be discontinuous in the point $(\mathfrak W,\mathfrak V)$. By property 4 this implies that for all
$\epsilon>0$ there is $\mathfrak W_\epsilon$ such that $d(\mathfrak W,\mathfrak W_\epsilon)<\epsilon$ but $F(\mathfrak W_\epsilon)>0$. If
$\hat{\mathfrak{V}}$ is such that $C^{\mathrm{mean}}_{\mathrm{S,ran}}(\mathfrak W,\hat{\mathfrak{V}})>0$ then the pair $(\mathfrak W,\hat{\mathfrak{V}})$
fulfills all the points in the second of the two equivalent formulations in statement 4, and this implies that $C_{\mathrm{S}}$ is
discontinuous in the point $(\mathfrak W,\hat{\mathfrak{V}})$.
\end{proof}
\end{subsection}
\begin{subsection}{Proof of Lemma \ref{lemma:central-lemma}}
\begin{proof}[Proof of Lemma \ref{lemma:central-lemma}] The proof is in many ways similar to the one for Lemma \ref{lemma:central-lemma-II}. As
we know already that for some $c'>0$ and all large enough $n\in\mathbb N$
\begin{align}
\mathbb P(E_3\cap E_4\cap E_5)\geq1-\Gamma\cdot\exp(2^{-n\cdot c'})
\end{align}
holds from \cite{csiszar-narayan}, there is not much left to prove, as only $\mathbb P(E_1)$ needs to be controlled in
order to get statement $(\ref{eqn:property-5-of-central-lemma})$ of Lemma \ref{lemma:central-lemma}.
We know from Lemma \ref{lemma:intermediate-secrecy-result-I} that both
\begin{align}
\mathbb P(E_1^\complement)\leq 2\cdot|\mathcal X\times\mathcal S\times\mathcal Z|^n\cdot \exp(-2^{n\cdot\tau/2}),
\end{align}
if we choose $\delta=\delta(\tau)$. Keeping in mind that we already know from \cite{csiszar-narayan} that $\mathbb P(E_3\cap E_4\cap
E_5)\geq1-\Gamma\cdot \exp(2^{n\cdot c'})$ we can combine all the previous to get the statement
\begin{align}
\mathbb P(E_3\cap\ldots\cap E_1)\geq1-(2+\Gamma)\cdot \exp(-2^{n\cdot c''}),
\end{align}
for some $c''>0$ and for all large enough $n$. If $\Gamma$ scales at most exponentially there will thus exist $N_0\in\mathbb N$ such that for
all $n\geq N_0$ there exists a choice $\mathbf x=(\mathbf x_{kl\gamma})_{k,l,\gamma=1}^{K,L,\Gamma}$ satisfying all conditions in Lemma
\ref{lemma:central-lemma} and, in addition, the estimate
\begin{align}
\forall\ s^n,z^n,k:\ \frac{1}{L\cdot\Gamma}\sum_{l,\gamma=1}^{L,\Gamma}\Theta_{s^n,z^n}(\mathbf x_{kl\gamma})\notin[(1\pm2^{-n\tau/4})\mathbb
E\Theta_{s^n,z^n}].
\end{align}
That this leads to secure transmission is proven exactly as in the proof of Lemma \ref{lemma:central-lemma-II}. The Lemma is thus proven.
\end{proof}
\end{subsection}
\begin{subsection}{Proof of Theorem \ref{thm:full-characteriation-of-super-activation} (super-activation
results)\label{subsec:proof-of-the-super-activation-result}}
We will divide this proof into three parts, each corresponding to its counterpart in Theorem
\ref{thm:full-characteriation-of-super-activation}.
\begin{proof}
{\bf1.} Let us start with the ``only if'' statement. Clearly, if $\mathfrak W_1\otimes\mathfrak W_2$ is symmetrizable then
$C_{\mathrm{S}}(\mathfrak W_1\otimes\mathfrak W_2,\mathfrak V_1\otimes\mathfrak V_2)=0$. So, this part of the statement is proven.\\
If, on the other hand, $\mathfrak W_1\otimes\mathfrak W_2$ is not symmetrizable and $C^{\mathrm{mean}}_{\mathrm{S,ran}}(\mathfrak W_1\otimes\mathfrak
W_2,\mathfrak V_1\otimes\mathfrak V_2)>0$ then on account of Theorem \ref{theorem:symmetrizability-properties-of-C1det}, statement 1, we know that
$C_{\mathrm{S}}(\mathfrak W_1\otimes\mathfrak W_2,\mathfrak V_1\otimes\mathfrak V_2)>0$.\\
This proves the first part of the Theorem.
\\\\
{\bf2.} In \cite{bs}, Section VI, an explicit example of a pair $(\mathfrak W_i,\mathfrak V_i)_{i=1,2}$ has been given with the property
that $\mathfrak W_1$ is symmetrizable, but $\mathfrak W_2$ is not. By elementary calculus, this implies that $\mathfrak W_1\otimes\mathfrak
W_2$ is non-symmetrizable.\\
Since this holds, our Theorem \ref{theorem:symmetrizability-properties-of-C1det}, statement 1, shows that the uncorrelated capacity of $(\mathfrak
W_1\otimes\mathfrak W_2,\mathfrak V_1\otimes\mathfrak V_2)$ equals its randomness-assisted capacity.\\
In \cite{bs} it was further shown that $C^{\mathrm{mean}}_{\mathrm{S,ran}}(\mathfrak W_1,\mathfrak V_1)>0$ and $C_{\mathrm{S}}(\mathfrak
W_i,\mathfrak V_i)=0$ ($i=1,2$).
\\\\
{\bf3.} By assumption, $C^{\mathrm{mean}}_{\mathrm{S,ran}}(\mathfrak W_i,\mathfrak V_i)=0$ ($i=1,2$) but $C^{\mathrm{mean}}_{\mathrm{S,ran}}(\mathfrak
W_1\otimes\mathfrak V_1,\mathfrak W_2\otimes\mathfrak V_2)>0$. The former implies $C_{\mathrm{S}}(\mathfrak W_i,\mathfrak V_i)=0$
($i=1,2$). If $\mathfrak W_1$ and $\mathfrak W_2$ were symmetrizable then clearly $\mathfrak W_1\otimes\mathfrak W_2$ would be symmetrizable
and by \cite{ericson} the message transmission capacity of $\mathfrak W_1\otimes\mathfrak W_2$ would be zero, implying
$C_{\mathrm{S}}(\mathfrak W_1\otimes\mathfrak W_2,\mathfrak V_1\otimes\mathfrak V_2)=0$. If on the other hand either $\mathfrak W_1$ or
$\mathfrak W_2$ are not symmetrizable then $\mathfrak W_1\otimes\mathfrak W_2$ is not symmetrizable and this implies
\begin{align}
C_{\mathrm{S}}(\mathfrak W_1\otimes\mathfrak W_2,\mathfrak V_1\otimes\mathfrak V_2)=C^{\mathrm{mean}}_{\mathrm{S,ran}}(\mathfrak
W_1\otimes\mathfrak W_2,\mathfrak V_1\otimes\mathfrak V_2)>0,
\end{align}
where the equality is due to Theorem \ref{theorem:symmetrizability-properties-of-C1det}, part 1, and the lower bound is true by assumption.
\\\\
{\bf4.} We do again rely on Theorem \ref{theorem:symmetrizability-properties-of-C1det}. Let both $\mathfrak W_1$ and $\mathfrak W_2$ be symmetrizable. Then $\mathfrak W_1\otimes\mathfrak W_2$ is symmetrizable. Since by assumption $C^{\mathrm{mean}}_{\mathrm{S,ran}}$ shows no super-activation on the pair $(\mathfrak W_i,\mathfrak V_i)$ ($i=1,2$) it follows that $C_{\mathrm{S}}$ cannot show super-activation as well. Thus at least one of the two AVCs has to be non-symmetrizable. Let without loss of generality this channel be $\mathfrak W_1$.\\
If in addition $\mathfrak W_2$ would be non-symmetrizable, then $C_{\mathrm{S}}(\mathfrak W_i,\mathfrak V_i)=C^{\mathrm{mean}}_{\mathrm{S,ran}}(\mathfrak W_i,\mathfrak V_i)$ would hold for $i=1,2$ and since $\mathfrak W_1\otimes\mathfrak W_2$ would be symmetrizable as well, we would additionally have $C_{\mathrm{S}}(\mathfrak W_1\otimes\mathfrak W_2,\mathfrak V_1\otimes\mathfrak V_2)=C^{\mathrm{mean}}_{\mathrm{S,ran}}(\mathfrak W_1\otimes\mathfrak W_2,\mathfrak V_1\otimes\mathfrak V_2)$. But since $C^{\mathrm{mean}}_{\mathrm{S,ran}}$ shows no super-activation on the pair $(\mathfrak W_i,\mathfrak V_i)$ ($i=1,2$) this cannot be. Thus again without loss of generality we have $\mathfrak W_2$ is symmetrizable.\\
Since we are talking about super-activation of $C_{\mathrm{S}}$, it has to be that $C_{\mathrm{S}}(\mathfrak W_i,\mathfrak V_i)=0$ holds for $i=1,2$. But since $\mathfrak W_1$ is non-symmetrizable this requires that $C^{\mathrm{mean}}_{\mathrm{S,ran}}(\mathfrak W_1,\mathfrak V_1)=0$ holds. If in addition we would have $C^{\mathrm{mean}}_{\mathrm{S,ran}}(\mathfrak W_2,\mathfrak V_2)=0$ would hold than $C_{\mathrm{S}}$ could not be super-activated since $C^{\mathrm{mean}}_{\mathrm{S,ran}}$ cannot be super-activated by assumption. Thus $C^{\mathrm{mean}}_{\mathrm{S,ran}}(\mathfrak W_2,\mathfrak V_2)>0$.
\end{proof}
\end{subsection}
\begin{subsection}{Proof of Lemma \ref{lem:extension-of-symmetrizability}\label{subsec:proof-of-extension-lemma}}
We now prove Lemma \ref{lem:extension-of-symmetrizability}: First and without loss of generality, we have $\mathcal A\subset\mathcal A'$. Let $\mathfrak U$ be symmetrizable. Let $Q\in C(\mathcal A,\mathcal R)$ be the symmetrizing channel, meaning that for all $a,a'\in\mathcal A$ the equality
\begin{align}
\left(U\circ(Id\otimes Q)\right)(a,a')=\left(U\circ(Id\otimes Q)\right)(a',a)
\end{align}
holds true. It follows that for all $a,a'\in\mathcal A'$ it holds that
\begin{align}
\left(U\circ(T\otimes QT)\right)(a,a')&=\sum_{a'',a'''\in\mathcal A}\sum_{r\in\mathcal R}u(\cdot|a'',r)t(a''|a)q(r|a''')t(a'''|a')\\
&=\sum_{a'',a'''\in\mathcal A}\sum_{r\in\mathcal R}u(\cdot|a''',r)t(a''|a)q(r|a'')t(a'''|a')\\
&=\left(U\circ(T\otimes QT)\right)(a',a).
\end{align}
Thus, $\mathfrak U'$ is symmetrizable.
\end{subsection}
\end{section}
\begin{section}{Appendix (auxiliary results and proofs)}
\begin{lem}[Cf. \cite{bbt-compound}\label{lemma:types-are-dense}]
Let $p\in\mathcal P(\mathcal X)$. For every $n\geq|\mathcal X|^2$, there is $p'\in\mathcal P_0^n(\mathcal X)$ such that
\begin{align}
\|p-p'\|_1\leq\frac{2|\mathcal X|}{n}
\end{align}
and $p(x)=0$ implies $p'(x)=0$ for all $x\in\mathcal X$.
\end{lem}
\begin{proof}[Proof of Lemma \ref{lemma:types-are-dense}]
Let $n\in\mathbb N$ be arbitrary. Set $\mathcal X':=\{x\in\mathcal X:p(x)>0\}$. From the next lines it will follow that, without loss of
generality, we may assume $\mathcal X=\mathcal X'$. For sake of simplicity, assume again without loss of generality that $\mathcal
X=\{1,\ldots,|\mathcal X|\}$ and that $p(|\mathcal X|)\geq 1/|\mathcal X|$. Choose $p'(i)$, for $i=1,\ldots,|\mathcal X|-1$, such that
$|p'(i)-p(i)|\leq\frac{1}{n}$. Clearly, this is possible. Then necessarily $p'(|\mathcal X|)=1-\sum_{i=1}^{|\mathcal X|-1}p'(i)$ and
\begin{align}
\|p-p'\|_1&\leq\sum_{i=1}^{|\mathcal X|-1}\frac{1}{n}+|p'(|\mathcal X|)-p(|\mathcal X|)|\\
&=\frac{|\mathcal X|-1}{n}+|\sum_{i=1}^{|\mathcal X|-1}p(i)-p'(i)|\\
&\leq\frac{|\mathcal X|-1}{n}+\sum_{i=1}^{|\mathcal X|-1}|p(i)-p'(i)|\\
&\leq\frac{2|\mathcal X|}{n}.
\end{align}
Of course, while all the $p'(i)\geq0$ by construction if $i<|\mathcal X|$, this does not hold for $p'(|\mathcal X|)$. This is where we need the
additional condition that $n\geq|\mathcal X|^2$:
\begin{align}
p'(|\mathcal X|)&=1-\sum_{i=1}^{|\mathcal X|-1}p'(i)\\
&\geq1-\sum_{i=1}^{|\mathcal X|-1}p(i)-\frac{|\mathcal X|-1}{n}\\
&\geq p(|\mathcal X|)-\frac{|\mathcal X|}{n}\\
&\geq\frac{1}{|\mathcal X|}-\frac{|\mathcal X|}{n}\\
&\geq0.
\end{align}
\end{proof}
\begin{lem}[C.f. \cite{csiszar-types} ]\label{lemma:cardinality-bound}
Let $\hat a^n\in\mathcal A^n$ and $\hat b^n\in\mathcal B^n$. There exists a function $f_C:\mathbb N\to\mathbb R_+$ such that with $\hat A\hat B$ being distributed as $\mathbb P((\hat A,\hat B)=(a,b))=\tfrac{1}{n}N(a,b|\hat a^n,\hat b^n)$ we have
\begin{align}
|\{a^n:N(\cdot|\hat a^n,\hat b^n)=N(\cdot|a^n,\hat b^n)\}|=2^{n\cdot(H(\hat A|\hat B)-f_C(n))}.
\end{align}
The function $f_C$ satisfies $\lim_{n\to\infty}f_C(n)=0$.
\end{lem}
The following Lemma is basically taken from \cite{csiszar-koerner}. It would generally be completely sufficient for proving all our statements in sufficient generality.
\begin{lem}\label{lemma:continuity-of-entropy}
Let $D(p\|q)\leq\delta$. For the function $f_1:[0,1/2]\to\mathbb R_+$ defined by $f_1(x):=-\sqrt{x/2}\log(x|\mathcal Z|^2)$ we have that
\begin{align}
|H(p)-H(q)|\leq f_4(\delta).
\end{align}
Clearly, $\lim_{\delta\to0}f_4(\delta)=0$.
\end{lem}
Note that $p(x)=0$ implies $p'(x|s)=0$ for all $s\in\mathcal S$, by construction.
\begin{proof}
From Pinsker's inequality we have $\|p-q\|_1\leq\sqrt{2\delta}$ and, accordingly, by Lemma 2.7 in \cite{csiszar-koerner},
$|H(p)-H(q)|\leq-\sqrt{2\delta}\log(\sqrt{2\delta}/|\mathcal Z|)$.
\end{proof}
We did however feel that it would be interesting to use a slightly more general version of Lemma \ref{lemma:continuity-of-entropy}, which led us to prove the following Lemma:
\begin{lem}[Continuity of conditional entropy with respect to averaged
norm]\label{lemma:continuity-of-conditional-entropy-with-respect-to-averaged-norm}
Let $p\in\mathcal P(\mathcal X)$ and channels $w,r:\mathcal P(\mathcal X)\to\mathcal P(\mathcal Z)$ be given such that
\begin{align}
\sum_{x\in\mathcal X}p(x)\|w(\cdot|x)-r(\cdot|x)\|_1\leq\delta\leq1.
\end{align}
Then
\begin{align}
|H(w|p)-H(r|p)|\leq f_1(\delta),
\end{align}
where $f_1(\delta):=|\mathcal Z|\cdot h(\frac{\delta}{|\mathcal Z|})$.
\end{lem}
\begin{proof}[Proof of Lemma \ref{lemma:continuity-of-conditional-entropy-with-respect-to-averaged-norm}]
As in \cite{csiszar-koerner}, set $\nu(t):=-t\log t$ and observe that $\nu$ is concave and satisfies $\nu(0)=\nu(1)=0$. This brings with it the
property that for all $s,\lambda\in[0,1]$ we have
\begin{align}
\nu(\lambda\cdot a)&\geq\lambda\cdot\nu(a),\qquad\nu(\lambda\cdot a+1-\lambda)\geq\lambda\cdot\nu(a).
\end{align}
We wish to obtain a meaningful bound on $|\nu(s)-\nu(t)|$ in terms of $|s-t|$. To this end, assume without loss of generality that $s\leq t$.
Observe that this implies that $|t-s|=t-s$, so that both
\begin{align}
\nu(|t-s|)+\nu(s)&=\nu(t\cdot\frac{t-s}{t})+\nu(t\cdot\frac{s}{t})\\
&\geq \frac{t-s}{t}\cdot\nu(t)+\frac{s}{t}\cdot\nu(t)\\
&=\nu(t)
\end{align}
and with $\lambda:=\frac{t-s}{1-s}$ satisfying $0\leq\lambda\leq1$ we have
\begin{align}
\nu(1-|t-s|)+\nu(t)&=\nu(\lambda\cdot s+1-\lambda)+\nu(\lambda+(1-\lambda)\cdot s)\\
&\geq\lambda\nu(s)+(1-\lambda)\nu(s)\\
&=\nu(s),
\end{align}
so that in total we get for every two number $s,t\in[0,1]$:
\begin{align}
|\nu(t)-\nu(s)|&\leq\max\{\nu(|t-s|),\nu(1-|t-s|)\}\\
&\leq\nu(|t-s|)+\nu(1-|t-s|)\\
&=h(|t-s|)
\end{align}
where $h$ denotes the binary entropy. Then for every $(\epsilon_x)_{x\in\mathcal X}\in[-1,1]^{|\mathcal X|}$ and $(t_x)_{x\in\mathcal
X}\in[0,1]^{|\mathcal X|}$ such that $t_x+\epsilon_x\in[0,1]$ for all $x\in\mathcal X$ we get:
\begin{align}
|\sum_{x\in\mathcal X}p(x)(\nu(t_x)-\nu(t_x+\epsilon_x))|&\leq\sum_{x\in\mathcal X}p(x)|\nu(t_x)-\nu(t_x+\epsilon_x)|\\
&\leq\sum_{x\in\mathcal X}p(x)h(|\epsilon_x|)\\
&\leq h(\sum_{x\in\mathcal X}p(x)|\epsilon_x|).
\end{align}
Then, we write $t_{xz}:=w(z|x)$ and $\epsilon_{xz}:=-w(z|x)+r(z|x)$. This leads to the bound we ultimately need:
\begin{align}
|\sum_{x\in\mathcal X}p(x)H(w(\cdot|x))-H(r(\cdot|x))|&=|\sum_{z\in\mathcal Z}\sum_{x\in\mathcal X}p(x)(\nu(w(z|x))-\nu(r(z|x)))|\\
&\leq\sum_{z\in\mathcal Z}|\sum_{x\in\mathcal X}p(x)(\nu(t_{xz}-\nu(t_{xz}+\epsilon_{xz})|\\
&\leq\sum_{z\in\mathcal Z}h(\sum_{x\in\mathcal X}p(x)|\epsilon_{xz}|)\\
&\leq|\mathcal Z|\cdot h(\frac{1}{|\mathcal Z|}\sum_{x\in\mathcal X}\sum_{z\in\mathcal Z}p(x)|\epsilon_{xz}|)\\
&=|\mathcal Z|\cdot h(\frac{1}{|\mathcal Z|}\delta)
\end{align}
\end{proof}
\end{section}
\emph{Acknowledgements.}
J.N. wants to thank Prakash Narayan, Aylin Yener, Ebrahim MolavianJazi and Mohamed Nafea for fruitful and lively discussions. The authors are grateful to their unknown referees for helping them to increase the quality of the manuscript. This work was supported by the DFG via grant BO 1734/20-1 (H.B.) and by the BMBF via the grants 01BQ1050 and 16KIS0118 (H.B., J.N.).\\
Further funding (J.N.) was provided by the ERC Advanced Grant IRQUAT, the Spanish MINECO Project No. FIS2013-40627-P and the Generalitat de Catalunya CIRIT Project No. 2014 SGR 966.

\end{document}